%% file: main.tex
\documentclass[12pt]{article} %

\usepackage{etoolbox}
\newtoggle{preprint}
\newtoggle{jasa}
\newcommand{\preprint}[1]{\iftoggle{preprint}{#1}{}}
\newcommand{\jasa}[1]{\iftoggle{jasa}{#1}{}}

\toggletrue{preprint} %
\togglefalse{jasa} %

\preprint{
\usepackage{arxiv}

\include{preamble}

\include{preprint-preamble}

\bibliographystyle{abbrvnat}

\makeatletter
\newcommand{\printfnsymbol}[1]{%
  \textsuperscript{\@fnsymbol{#1}}%
}
\makeatother

\title{Stochastic Convergence Rates and Applications of Adaptive Quadrature in Bayesian Inference}

\date{}           %

\author[1]{Blair Bilodeau}
\author[2]{Alex Stringer}
\author[3]{Yanbo Tang}
\affil[1]{Department of Statistical Sciences, University of Toronto}
\affil[2]{Department of Statistics and Actuarial Science, University of Waterloo}
\affil[3]{Department of Mathematics, Imperial College London}

}%

\jasa{
\include{preamble}
\usepackage{amsthm}
\usepackage{setspace}

\usepackage[pagewise]{lineno}
\linenumbers

\AtBeginDocument{%
  \newcommand*\patchAmsMathEnvironmentForLineno[1]{%
    \expandafter\let\csname old#1\expandafter\endcsname\csname #1\endcsname
    \expandafter\let\csname oldend#1\expandafter\endcsname\csname end#1\endcsname
    \renewenvironment{#1}%
                     {\linenomath\csname old#1\endcsname}%
                     {\csname oldend#1\endcsname\endlinenomath}%
  }%
  \newcommand*\patchBothAmsMathEnvironmentsForLineno[1]{%
    \patchAmsMathEnvironmentForLineno{#1}%
    \patchAmsMathEnvironmentForLineno{#1*}%
  }%
  \patchBothAmsMathEnvironmentsForLineno{equation}%
  \patchBothAmsMathEnvironmentsForLineno{align}%
  \patchBothAmsMathEnvironmentsForLineno{flalign}%
  \patchBothAmsMathEnvironmentsForLineno{alignat}%
  \patchBothAmsMathEnvironmentsForLineno{gather}%
  \patchBothAmsMathEnvironmentsForLineno{multline}%
}

\newcommand{\blind}{0}

\usepackage[left=1in,right=1in,bottom=2in,top=1in]{geometry}
\addtolength{\textheight}{-.3in}%

\newtheorem{definition}{Definition}
\newtheorem{corollary}{Corollary}
\newtheorem{remark}{Remark}
\newtheorem{lemma}{Lemma}
\newtheorem{proposition}{Proposition}
\newtheorem{assumption}{Assumption}

\newcommand{\manualendproof}{}
\usepackage{sectsty}
\usepackage{titlecaps}
\usepackage{titlesec}

\renewcommand{\thefootnote}{\fnsymbol{footnote}}

}%

\usepackage[hide=false,setmargin=true,marginparwidth=0.75in]{marginalia}

\begin{document}

\preprint{
\maketitle
}
\jasa{
\def\spacingset#1{\renewcommand{\baselinestretch}%
{#1}\small\normalsize} \spacingset{1}

\title{\bf Stochastic Convergence Rates and Applications of Adaptive Quadrature\\ in Bayesian Inference}
\if0\blind
{
\author{$\text{Blair Bilodeau}^{1}$, $\text{Alex Stringer}^{2}$, $\text{Yanbo Tang}^{1}$\thanks{
Authors listed in alphabetical order.
BB is supported by an NSERC Canada Graduate Scholarship and the Vector Institute.
AS is supported by an NSERC Postgraduate Scholarship and the Centre for Global Health Research at St. Michael's Hospital, Toronto, Canada.
YT is supported by an NSERC Postgraduate Scholarship and the Vector Institute.
We thank Jeffrey Negrea, Nancy Reid, Daniel Roy, and Jamie Stafford for helpful comments and suggestions.}
\hspace{.2cm}\\
Department of Statistical Sciences, $\text{University of Toronto}^{1}$\\
Department of Statistics and Actuarial Science, $\text{University of Waterloo}^{2}$}
} \fi
\if 1\blind
{
\author{\emph{Authors Blinded}}
} \fi
\date{}
\maketitle

} %

\preprint{
\begingroup\renewcommand\thefootnote{*}
\footnotetext{Authors listed alphabetically.}
}

\jasa{\bigskip}
\begin{abstract}
We provide the first stochastic convergence rates for a family of adaptive quadrature rules used to normalize the posterior distribution in Bayesian models. 
Our results apply to the uniform relative error in the approximate posterior density, the coverage probabilities of approximate credible sets, and approximate moments and quantiles, therefore guaranteeing fast asymptotic convergence of approximate summary statistics used in practice. 
The family of quadrature rules includes adaptive Gauss-Hermite quadrature, and we apply this rule in two challenging low-dimensional examples.
Further, we demonstrate how adaptive quadrature can be used as a crucial component of a modern approximate Bayesian inference procedure for high-dimensional additive models.
The method is implemented and made publicly available in the \texttt{aghq} package for the \texttt{R} language, available on \texttt{CRAN}. 
\end{abstract}

\jasa{
	\noindent%
\vfill

\newpage
\spacingset{1.5} %
}%

\jasa{
\sectionfont{\MakeUppercase}
\titlelabel{\thetitle.\quad}
\renewcommand{\thefootnote}{\arabic{footnote}}
}
\preprint{\renewcommand{\thefootnote}{\arabic{footnote}}}

\preprint{
\newpage
\tableofcontents
}
\preprint{\newpage}

\input{sections/introduction}

\input{sections/preliminaries}

\input{sections/convergence-with-statistics}

\input{sections/examples-low-dim}

\input{sections/examples-high-dim}

\input{sections/discussion}

\preprint{

\input{sections/acknowledgements}
\bibliography{biblio.bib}
\newpage
}

\appendix

\jasa{
\bibliography{biblio.bib}

\newpage
\setstretch{1}
\newgeometry{left=1in,right=1in,bottom=1in,top=1in}
\begin{center}
{\Large\bf SUPPLEMENTARY MATERIAL}
\end{center}

\setcounter{section}{0}
\renewcommand{\theHsection}{S.\arabic{section}}
\renewcommand{\thesection}{S.\arabic{section}}
\crefalias{appendix}{supplement}
\crefalias{subappendix}{supplement}
}

\input{sections/notation-and-assumptions}

\input{sections/proof-B-main}

\input{sections/proof-C-submain}

\input{sections/proof-D-big-lemma}
\input{sections/proof-A-summaries}

\input{sections/computational}

\input{sections/simulations}

\input{sections/loaloa-simulations}
\input{sections/glossary}

\end{document}

%% file: preamble.tex
\usepackage{amsmath,amssymb,bbm}
\usepackage{amsthm}
\usepackage{thmtools}
\usepackage{natbib}
\bibliographystyle{chicago}
\usepackage[hyperfootnotes=false, hidelinks]{hyperref}
\usepackage[capitalize,nameinlink]{cleveref}

\usepackage{breqn}
\usepackage{changepage}
\usepackage{dsfont}
\usepackage{multirow}
\usepackage{makecell}
\usepackage{booktabs}
\usepackage{array}
\usepackage{mathtools}
\usepackage{thmtools}
\usepackage{esvect}
\usepackage{url}
\usepackage{algorithm}
\usepackage{algorithmic}
\usepackage{bm}
\usepackage{wrapfig}
\usepackage{mathrsfs}
\usepackage{subfiles}
\usepackage[caption = false]{subfig} 
\usepackage{graphicx}
\usepackage{float}
\usepackage{tabularx}
\usepackage[symbol]{footmisc}
\usepackage{enumerate}
\usepackage{comment}
\usepackage{appendix}
\makeatletter
\patchcmd{\@makecaption}
  {\parbox}
  {\advance\@tempdima-\fontdimen2} %
  {}{}

\newcommand{\thickhline}{%
    \noalign {\ifnum 0=`}\fi \hrule height 1pt
    \futurelet \reserved@a \@xhline
}
\makeatother 

\def\[#1\]{\begin{equation}\begin{aligned}#1\end{aligned}\end{equation}}
\def\*[#1\]{\begin{equation*}\begin{aligned}#1\end{aligned}\end{equation*}}

\usepackage{crossreftools}
\pdfstringdefDisableCommands{%
    \let\Cref\crtCref
    \let\cref\crtcref
}

\hypersetup{%
    bookmarksnumbered, bookmarksopen=true, bookmarksopenlevel=1,colorlinks=false%
}

\crefname{subsection}{Section}{Sections}
\crefname{lemma}{Lemma}{Lemmas}
\crefname{corollary}{Corollary}{Corollaries}
\crefname{theorem}{Theorem}{Theorems}
\crefname{definition}{Definition}{Definitions}
\crefname{assumption}{Assumption}{Assumptions}
\crefname{supplement}{Supplement}{Supplements}

\newcommand{\N}{\mathbb{N}}
\newcommand{\NN}{\mathcal{N}}
\newcommand{\JJ}{\mathcal{J}}

\newcommand{\R}{\mathbb{R}}

\newcommand{\GP}{\mathcal{GP}}
\newcommand{\VE}{\mathcal{E}}
\newcommand{\PolyP}{\mathcal{P}}

\newcommand{\eps}{\varepsilon}

\newcommand{\dee}{\mathrm{d}}
\newcommand{\ind}{\mathbb{I}}

\newcommand{\sT}{\mathrm{\ s.t. \ \ }}
\newcommand{\negsetdelim}{\vert}
\newcommand{\setdelim}{\ \vert \ }

\newcommand{\Bigsetdelim}{\ \Big\vert \ }

\newcommand{\tpose}{^{\top}}

\DeclareMathOperator{\EE}{\mathbb{E}}
\DeclareMathOperator{\PP}{\mathbb{P}}

\DeclareMathOperator{\argmax}{\arg\max}

\newcommand{\absbig}[1]{\Big\lvert #1 \Big\rvert}
\newcommand{\abssmall}[1]{\lvert #1 \rvert}
\newcommand{\Norm}[1]{\left\lVert#1\right\rVert}

\newcommand{\isim}{\overset{iid}\sim}
\newcommand{\Tr}{{}^\intercal}

\newcommand{\Ord}{\mathcal{O}}
\newcommand{\Ordp}{\mathcal{O}_{P}}

\newcommand{\iid}{\text{i.i.d.}}

\newcommand{\AGHQ}{{\upshape AGHQ}}
\newcommand{\GHQ}{{\upshape GHQ}}
\newcommand{\MCMC}{{\upshape MCMC}}

\newcommand{\property}[2]{\mathscr{P}(#1,#2)}

\newcommand{\mleaccent}[1]{{#1}^{\mathrm{\scriptscriptstyle MLE}}}
\newcommand{\modeaccent}[1]{\widehat{#1}}
\newcommand{\intapproxaccent}[1]{\check{#1}}

\newcommand{\dist}{\pi} %
\newcommand{\distfunc}{\pi^{\meanfunc}} %
\newcommand{\margdist}[1]{\pi^{#1}}

\newcommand{\llhood}{\ell_n}

\newcommand{\llandp}{{\ell^\dist_n}}

\newcommand{\logposthess}{\mb{H}_{n}}
\newcommand{\logposthessfunc}{\mb{H}^{\meanfunc}_{n}}
\newcommand{\logposthessmarg}{\mb{H}^{\margparamval}_{n}}

\newcommand{\logpostchol}{\modeaccent{\mb{L}}_{n}}
\newcommand{\logpostcholfunc}{\modeaccent{\mb{L}}^{\meanfunc}_{n}}
\newcommand{\logpostcholmarg}{\modeaccent{\mb{L}}^{\margparamval}_{n}}

\newcommand{\approxdist}{\widetilde \dist} %
\newcommand{\approxmargdist}[1]{{\widetilde \dist}^{#1}} %
\newcommand{\quadapproxdist}[2]{\widetilde \dist_{(#1,#2)}}
\newcommand{\approxint}{\widetilde \genintFunc}
\newcommand{\quadapproxint}[2]{\widetilde \genintFunc_{(#1,#2)}}
\newcommand{\quadadaptapproxint}[2]{\widetilde \genintFunc^{\adapt}_{(#1,#2)}}

\newcommand{\adapt}{\mathcal{A}}

\newcommand{\quadadaptmarg}[2]{{\quadapproxdist{#1}{#2}^{\adapt}(\data)}}

\newcommand{\quadadaptapprox}[3]{\quadapproxdist{#1}{#2}^{\adapt}(#3 \setdelim \data)}

\newcommand{\GHmarg}{\approxdist^{\text{\upshape\tiny AGHQ}}(\data)}
\newcommand{\GHapprox}[1]{\approxdist^{\text{\upshape\tiny AGHQ}}(#1 \setdelim \data)} %
\newcommand{\LAapprox}{\approxdist_{\text{\upshape\tiny LA}}}
\newcommand{\GGapprox}{\approxdist_{\text{\upshape\tiny G}}}
\newcommand{\LAGHapprox}[1]{\approxdist_{\text{\upshape\tiny LA}}^{\text{\upshape\tiny AGHQ}}(#1 \setdelim \data)}
\newcommand{\GHapproxidx}[2]{\approxdist^{\text{\upshape\tiny AGHQ}}_{#1}(#2 \setdelim \data)}
\newcommand{\GHapproxcdfidx}[2]{\widetilde{F}^{\text{\upshape\tiny POLY}}_{#1}(#2 \setdelim \data)}
\newcommand{\GHapproxidxpoly}[2]{\approxdist^{\text{\upshape\tiny POLY}}_{#1}(#2 \setdelim \data)}

\newcommand{\dataidx}{\mb{Y}}

\newcommand{\data}{\mb{Y}^{(n)}} %
\newcommand{\paramtrue}{\mb{\theta}^{*}} %
\newcommand{\paramtrueidx}{\theta^{*}} %
\newcommand{\paramdum}{\mb{\theta^\prime}}

\newcommand{\datatrueprobn}{\PP^*_{\!n}}

\newcommand{\post}[1]{\dist(#1 \setdelim \data)}
\newcommand{\margpost}[2]{\dist^{(#1)}(#2 \setdelim \data)}
\newcommand{\intapproxmargpost}[2]{\intapproxaccent{\dist}^{(#1)}(#2 \setdelim \data)}
\newcommand{\postidx}[2]{\dist_{#1}(#2 \setdelim \data)}
\newcommand{\postcdfidx}[2]{F_{#1}(#2 \setdelim \data)}
\newcommand{\approxpost}[1]{\approxdist(#1 \setdelim \data)}
\newcommand{\poststar}[1]{\dist(#1 , \data)}

\newcommand{\multitaylorfunc}{A_n}
\newcommand{\midmultitaylorfunc}{\overline{A_n}}

\newcommand{\parammle}{\mleaccent{\param}_n} %

\newcommand{\parammode}{\modeaccent{\param}_n} %
\newcommand{\parammodemarg}{\modeaccent{\param}_n^{\margparamval}} 
\newcommand{\margparammode}{\modeaccent{\margparam}_n}
\newcommand{\nuisparammode}{\modeaccent{\nuis}_n}

\newcommand{\parammodefunc}{\modeaccent{\param}_n^{\meanfunc}} %
\newcommand{\parammid}{\mb{\vartheta}_n}
\newcommand{\parammodemid}{\mb{\modeaccent{\vartheta}}_n}
\newcommand{\midweight}{\tau}
\newcommand{\modemidweight}{\modeaccent{\tau}}

\newcommand{\dummydim}{p}

\newcommand{\dummytau}{b}
\newcommand{\paramdim}{p}
\newcommand{\paramdimbig}{m}
\newcommand{\datadim}{d}

\newcommand{\param}{\mb{\theta}}
\newcommand{\highdimparam}{\mb{w}}
\newcommand{\paramidx}{\theta}

\newcommand{\paramspace}{\Theta}

\newcommand{\hermite}{H}
\newcommand{\hermitezero}{z^*}
\newcommand{\hermitezerovec}{\mb{z}}

\newcommand{\hermitezerobound}{\overline{\mb{z}}_{H}}
\newcommand{\hermiteset}{\mathcal{H}}

\newcommand{\weight}{\mb{\omega}}
\newcommand{\weightidx}{\omega}

\newcommand{\quadpointvec}{\mb{z}}
\newcommand{\quadpointidx}{z}
\newcommand{\quadpointbound}{\overline{\mb{z}}}
\newcommand{\quadpointset}{\mathcal{Q}}
\newcommand{\quadpointsetidx}{Q}
\newcommand{\quadrule}[2]{\mathfrak{R}(#1, #2)}

\newcommand{\quadnum}{k}
\newcommand{\quadnumadj}{\kappa}

\newcommand{\derivnum}{m} %
\newcommand{\derivbound}{M}
\newcommand{\derivboundmode}{\modeaccent{M_n}}
\newcommand{\dumderivvec}{\mb{\alpha}}
\newcommand{\dumderiv}{\alpha}
\newcommand{\multitaylorexpbound}{\modeaccent{\Lambda}_n}
\newcommand{\llhoodmargin}{b}
\newcommand{\priorbig}{c_2}
\newcommand{\priorsmall}{c_1}

\newcommand{\eigen}{\lambda}

\newcommand{\conmargin}{\beta}

\newcommand{\shrinkingrad}{\gamma_n}
\newcommand{\shrinkingconst}{\gamma}
\newcommand{\hessbig}{\overline{\eta}}
\newcommand{\hesssmallmean}{\underline{\eta}^\meanfunc}
\newcommand{\hessbigmean}{\overline{\eta}^\meanfunc}
\newcommand{\hessbigmode}{\overline{\eta_n}}
\newcommand{\hesssmallmode}{\underline{\eta_n}}
\newcommand{\hessbigtrue}{\overline{\eta^*_n}}
\newcommand{\hesssmalltrue}{\underline{\eta^*_n}}
\newcommand{\hesssmall}{\underline{\eta}}

\newcommand{\boundofinterest}{B_{n,\quadnum}}

\newcommand{\normaldens}[3]{\phi(#1; #2, #3)}
\newcommand{\stdnormaldens}[1]{\phi(#1)}

\newcommand{\nuis}{\mb{\lambda}}
\newcommand{\interestdim}{q}

\newcommand{\meanfunc}{g}
\newcommand{\genintfunc}{f}
\newcommand{\genintFunc}{F}
\newcommand{\polyfunc}{P}

\newcommand{\genintmode}{\modeaccent{\param}}
\newcommand{\genloghess}{\mb{H}}
\newcommand{\genlogchol}{\modeaccent{\mb{L}}}

\newcommand{\ball}[3]{B^{#1}_{#2}(#3)} %
\newcommand{\ballc}[3]{[B^{#1}_{#2}(#3)]^c} %

\newcommand{\constt}{C}
\newcommand{\paramconst}{C_{\quadnum, \paramdim}}

\newcommand{\Reals}{\mathbb{R}}
\newcommand{\PosReals}{\mathbb{R}_{+}}
\newcommand{\Nats}{\mathbb{N}}

\newcommand{\PosInts}{\mathbb{Z}_{+}}
\newcommand{\Borel}[1]{\mathcal{B}\left(\R^{#1}\right)}
\newcommand{\Borelset}{\mathcal{K}}

\newcommand{\norm}[1]{\left\lVert#1\right\rVert}

\newcommand{\approxEE}{\widetilde{\EE}}
\newcommand{\intapproxEE}{\intapproxaccent{\EE}}
\newcommand{\GHapproxEE}{\widetilde{\EE}^{\text{\upshape\tiny AGHQ}}}

\newcommand{\intapproxcredfunc}{\intapproxaccent{\mathcal{K}}}
\newcommand{\credfunc}{\mathcal{K}}

\newcommand{\margparam}{\mb{\psi}}
\newcommand{\margparamval}{\margparam_{0}}
\newcommand{\identmat}{\mb{I}}

\newcommand{\quant}[2]{q_n^{(#1)}(#2)}
\newcommand{\approxquant}[2]{\widetilde{q}_n^{(#1)}(#2)} 
\newcommand{\intapproxquant}[2]{\intapproxaccent{q}_n^{(#1)}(#2)} 
\newcommand{\approxquantpoly}[2]{\widetilde{q}^{\text{\upshape\tiny POLY}}_{#1}(#2)}

\newcommand{\cdfpost}[1]{F^{(#1)}_n}
\newcommand{\cdfpostinv}[1]{[F^{(#1)}_n]^{-1}}
\newcommand{\cdfapprox}[1]{\widetilde{F}^{(#1)}_n}

\newcommand{\cdfintapprox}[1]{\intapproxaccent{F}^{(#1)}_n}
\newcommand{\cdfintapproxinv}[1]{[\intapproxaccent{F}^{(#1)}_n]^{-1}}
\newcommand{\cdfderivbound}{L}
\newcommand{\universalradius}{\delta}

\newcommand{\multicoeff}[2]{\mathbb{M}_{#1, #2}}
\newcommand{\tvec}{\bm{t}}

\newcommand{\smalltset}[1]{\tau^{(#1)}_{<2\quadnum}}
\newcommand{\smalltsetdef}[2][1]{ \tau^{(#1)}_{{ < #2 }} }
\newcommand{\equaltsetdef}[2][1]{ \tau^{(#1)}_{{ = #2 }} }
\newcommand{\bigtset}[1]{\tau^{(#1)}_{\geq2\quadnum}}
\newcommand{\bigtsetdef}[2][1]{\tau^{(#1)}_{\geq #2}}
\newcommand{\customtset}[2]{\tau^{(#1)}_{\geq#2}}
\newcommand{\customtsetequal}[2]{\tau^{(#1)}_{=#2}}
\newcommand{\tsum}{\tau(\tvec)}

\newcommand{\bracevec}[1]{\left\{ #1 \right\}}

\newcommand{\mb}[1]{\boldsymbol{#1}}

\newcommand{\relerror}{\text{E}_{\mathrm{rel}}}

\newcommand{\TVerror}{\text{E}_{\mathrm{TV}}}

\newcommand{\samplesizes}{\NN}
\newcommand{\simsize}{M}

%% file: preprint-preamble.tex
\usepackage{authblk}

\makeatletter
\renewcommand{\maketag@@@}[1]{\hbox{\m@th\normalsize\normalfont#1}}%
\makeatother

\makeatletter
\let\reftagform@=\tagform@
\def\tagform@#1{\maketag@@@{\ignorespaces\textcolor{gray}{(#1)}\unskip\@@italiccorr}}
\renewcommand{\eqref}[1]{\textup{\reftagform@{\ref{#1}}}}
\makeatother

\newtheoremstyle{thmstyle}
{10pt} %
{10pt} %
{} %
{} %
{\bfseries} %
{.} %
{.5em} %
{} %

\newcommand{\manualendproof}{\hfill\qedsymbol\\[2mm]}

\declaretheorem[name=Theorem]{theorem}

\declaretheorem[name=Lemma]{lemma}
\declaretheorem[name=Proposition]{proposition}
\declaretheorem[name=Corollary]{corollary}

\declaretheorem[name=Definition]{definition}
\declaretheorem[name=Assumption]{assumption}
\declaretheorem[name=Assumption, numbered=no]{assumption*}

\declaretheorem[qed=$\triangleleft$,name=Remark]{remark}

\bibpunct[, ]{(}{)}{;}{a}{}{,}
\usepackage{tikz}

\usepackage{titlesec}
\usepackage[subfigure]{tocloft}
\setlength{\cftbeforesecskip}{5pt}
	\titlelabel{\thetitle.\quad}
	\titleformat*{\section}{\large\bfseries}
	\titleformat*{\subsection}{\bfseries}
	\titleformat*{\subsubsection}{\itshape}

\titlespacing*\section{0pt}{12pt plus 4pt minus 2pt}{4pt plus 2pt minus 2pt}
\titlespacing*\subsection{0pt}{12pt plus 4pt minus 2pt}{4pt plus 2pt minus 2pt}
\titlespacing*\subsubsection{0pt}{12pt plus 4pt minus 2pt}{4pt plus 2pt minus 2pt}

%% file: sections/introduction.tex
\section{Introduction}\label{sec:intro}

The central challenge of Bayesian inference is computing the \emph{posterior distribution}, which requires evaluating an integral---the \emph{normalizing constant} or \emph{marginal likelihood}---that is intractable in all but the simplest models. 
\emph{Numerical quadrature} (or \emph{cubature} in multiple dimensions) comprises a range of techniques for approximating deterministic integrals via function evaluations at finitely many points, and is a mature field of study in applied mathematics; see \citet{numint} for an overview.
However, in Bayesian inference, the integrand is necessarily changing with the observed data, and consequently fixed quadrature rules can perform arbitrarily poorly by failing to capture the shifting mass of the integrand. 
To address this limitation, \emph{adaptive quadrature}
techniques based on shifting and scaling fixed quadrature rules using the mode and curvature of the integrand have been proposed in Bayesian inference since at least \citet{nayloradaptive}.
More recently, adaptive quadrature has been employed as a fundamental component of approximate Bayesian inference in the popular INLA framework \citep{inla}, for integrating out random effects \citep{non-linear-latent,latent}, and in the context of Bayesian inversion \citep{schillings}. 
In the present work, we study stochastic convergence rates as the sample size tends to infinity for fixed parameter dimension and number of quadrature points.
\cref{THM:MAINRESULT} and its corollaries provide the first stochastic convergence rates for the error in normalizing Bayesian posterior densities with adaptive quadrature, and for computing approximate moments and marginal densities.

Despite its broad applicability and usage by practitioners, relatively little is known about the theoretical properties of adaptive quadrature for Bayesian inference.
\cite{nayloradaptive} discuss the practical application of what \citet{laplace} call \emph{adaptive Gauss--Hermite quadrature} (\AGHQ{}) in Bayesian inference, arguing that it is a useful tool for normalizing posterior distributions and computing approximate summary statistics, but do not provide any theoretical guarantees on its accuracy. 
\citet{laplace} use \AGHQ{} to renormalize a Laplace approximation (which itself corresponds to \AGHQ{} with a single quadrature point) of the marginal posterior density, however they do not discuss the effect that the error introduced by this numerical renormalization may have on their convergence rate. 
\citet{validitylaplace} rigorously establish convergence rates for the Laplace approximation in Bayesian inference, but their proof only applies to one-dimensional parameters, and they do not provide results for \AGHQ{} with multiple quadrature points. 
More recently, for a specific, restricted class of one-dimensional functions that vary only through a scaling parameter $n$, \citet{adaptive_GH_2020} expand upon the work of \citet{adaptive_GH_1994} to show that \AGHQ{} with $\quadnum\in\Nats$ quadrature points converges at relative rate $\Ord(n^{-\lfloor (\quadnum+2)/3\rfloor})$, and comment that this rate holds in multiple dimensions using a specific extension of the univariate rule.

In statistical problems,
the integrand
varies through $n$ both as a scaling factor and the data observed at time $n$, which is not captured by
the class of functions considered by \citet{adaptive_GH_1994} and \citet{adaptive_GH_2020}.
A \emph{stochastic convergence} perspective is needed in order to quantify the behaviour of \AGHQ{} when used in fitting any statistical model, including when used to approximate the normalizing constant in Bayesian inference. Beyond the univariate Laplace analysis by \citet{validitylaplace}, 
a stochastic convergence perspective is also taken by \citet{schillings} in the context of Bayesian inversion of operator equations,
where
they study convergence as the variance of the data tends to zero. 
\citet{dick19inversion} study Bayesian PDE inversion using a Quasi-Monte Carlo variant of numerical quadrature.
As is standard in the numerical analysis literature, they suppose that the data (rather than the model) satisfies certain regularity conditions, and provide a convergence rate as the number of quadrature points tends to infinity.
Convergence limits as the data variability tends to zero and as the number of quadrature points tends to infinity for a fixed data sequence are both distinct from the asymptotic statistics perspective we take of letting the sample size $n$ tend to infinity, and our results explicitly identify the effect of $\quadnum$ on the approximation error.

In this work, we quantify the stochastic error of using adaptive quadrature rules to approximate the 
posterior distribution and summary statistics based on it.
We demonstrate that the crucial property of \AGHQ{} is that the underlying quadrature rule, Gauss-Hermite quadrature, exactly integrates the product of the Gaussian density and any polynomial of \emph{total order} $2\quadnum-1$ or less for a well-chosen integer $\quadnum$ in $\paramdim$ dimensions; we call this property $\property{\quadnum}{\paramdim}$ (see \cref{def:exact-integrate-property}).
More precisely, under standard regularity assumptions, we prove that if the posterior normalizing constant is approximated using the adaptive form of \emph{any} quadrature rule that satisfies $\property{\quadnum}{\paramdim}$, then the relative error of the approximate normalizing constant to the true normalizing constant 
converges in probability at rate $\Ordp(n^{-\lfloor (\quadnum+2)/3\rfloor})$, where $n$ is the number of observed data points (\cref{thm:mainresult}). 
The main technical contribution that enables this result is a precise quantification of how well the posterior distribution locally approximates a Gaussian distribution as a function of the data.

We further describe how to approximate marginal posterior densities and moments using a second application of \AGHQ{}, and show that the convergence rate is preserved for the error of these approximate summary statistics (\cref{fact:pos-marg-compute,fact:pos-mean-compute}).
For approximate quantiles and credible sets, we show that the rate is preserved in the ideal case where one can exactly integrate the \AGHQ{} approximation (\cref{cor:applications}), and provide a computational method to approximate these.
We include a simulation study (\cref{sec:simulations}) that illustrates the stochastic nature of the convergence rates and demonstrates an example of a simple model in which our stated rate is achieved in finite samples empirically. 

To illustrate the breadth of models for which adaptive quadrature provides a useful tool for inference, 
in \cref{sec:exampleslowdim}
we directly apply the method to two challenging examples that satisfy the conditions of our theoretical results.
For a distance-based, individual-level model for the spread of infectious disease, we show that the results obtained using \AGHQ{} are nearly identical to those obtained via traditional sampling-based approaches by \citet{epi}, at a substantial reduction in computational time. 
Additionally, we apply \AGHQ{} to estimate the mass of the Milky Way galaxy using multivariate position and velocity measurements of star clusters and a complex astrophysical model.

In addition to the low-dimensional examples of \cref{sec:exampleslowdim}, 
in \cref{sec:highdimexamples} we demonstrate the applicability of adaptive quadrature for
high-dimensional models by employing it within a broader method to fit a complex zero-inflated geostatistical binomial regression model, which we use to infer the spatial risk of contracting a certain tropical disease in West Africa, and for which Bayesian inferences had not previously been made. This example is particularly challenging and is not compatible with INLA or (to our knowledge) any other existing (non-\MCMC{}) framework for making approximate Bayesian inferences, and we discuss some observed difficulties with applying \MCMC{} to it as well.

\AGHQ{} and the corresponding high-dimensional method are implemented in the \texttt{aghq} package 
in the \texttt{R} statistical programming language, made publicly available on the \texttt{CRAN} package repository. 
All code for the examples and simulations in \preprint{the present}\jasa{this} paper is available \preprint{at \href{https://github.com/awstringer1/aghq-paper-code}{https://github.com/awstringer1/aghq-paper-code}}\jasa{in the supplementary material}.
\preprint{For more details on \texttt{aghq}, see the vignette by \citet{aghqsoftware}.}

%% file: sections/preliminaries.tex
\section{Preliminaries}\label{sec:prelim}

\subsection{Bayesian Inference}\label{subsec:posterior}

Suppose that we observe a dataset $\data = (\dataidx_1,\dots,\dataidx_n) \subseteq \R^{\datadim}$ generated from some unknown probability distribution.
We fix a \emph{model} for the data defined by a \emph{parameter space} $\paramspace \subseteq \R^{\paramdim}$ and \emph{likelihood} $\dist(\data \setdelim \param)$.
Often, the model assumes independent and identically distributed (\iid{}) data, in which case the likelihood factors into $\dist(\data \setdelim \param) = \prod_{i=1}^n \dist(\dataidx_i \setdelim \param)$, but we do not require this restriction. 
Further, we do not require that the model be well-specified; the only constraint on the data-generating distribution is that it satisfies certain regularity assumptions for the chosen model.

For some \emph{prior density} $\dist(\cdot)$ on $\paramspace$, the object of inferential interest is the \emph{posterior density}
\[
	\post{\param}
  = \frac{\poststar{\param}}{\int_{\paramspace}\poststar{\paramdum}\dee\paramdum},
\label{eqn:posterior}
\]
where $\poststar{\param} = \dist(\param)\dist(\data \setdelim \param)$ is the \emph{unnormalized posterior density}.
Inference for $\param$ is based upon point and interval estimates computed from $\post{\param}$.

All posterior summaries (marginals, moments, quantiles, etc.)
require knowledge of the normalized posterior distribution, which requires computing the denominator of \cref{eqn:posterior}, referred to as the \emph{marginal likelihood} or \emph{normalizing constant},
\*[
	\dist(\data) = \int_{\paramspace}\poststar{\paramdum}\dee\paramdum.
\]
The computation of this integral---as well as the further integration required to compute posterior summary statistics---is typically not analytically tractable, so inference is instead conducted using integral approximations.
\subsection{Numerical Quadrature}\label{subsec:adaptivequadrature}

A \emph{quadrature rule} for 
approximating an integral $\genintFunc = \int \genintfunc(\param)\dee\param$ of a function $\genintfunc:\paramspace\to\Reals$
is a collection of \emph{points} $\quadpointset \subseteq \paramspace$ and a \emph{weight function} $\weight: \quadpointset \to \Reals$, and is denoted by $\quadrule{\quadpointset}{\weight}$. The approximate 
integral
under such a rule is denoted by
\*[
	\quadapproxint{\quadpointset}{\weight}
	= 
	\sum_{\quadpointvec\in\quadpointset}\genintfunc(\quadpointvec) \, \weight(\quadpointvec),
\]
which we denote by just $\approxint$ when $\quadrule{\quadpointset}{\weight}$ is clear.
Often, quadrature rules are designed to be exact for specific functions of interest, most often polynomials. A $\paramdim$-dimensional polynomial $\polyfunc$, defined by  
\*[
	\polyfunc(x_1,\dots,x_\paramdim) = \sum_{t=1}^T a_t \prod_{d=1}^\paramdim x_d^{j_{t,d}}
\]
for some $T \in \Nats$, $(a_1,\dots,a_T) \in \Reals^T$, and $(j_{t,1},\dots,j_{t,\paramdim})_{t\in[T]} \subseteq \Nats^{\paramdim}$,
is said to have \emph{total order} $\max_{t\in[T]} \sum_{d=1}^\paramdim j_{t,d}$ \citep{heiss08sparse}.

\begin{definition}\label{def:exact-integrate-property}
For any $\quadnum,\paramdim\in\Nats$, a quadrature rule $\quadrule{\quadpointset}{\weight}$ satisfies $\property{\quadnum}{\paramdim}$ if for all $\paramdim$-dimensional polynomials $\polyfunc:\paramspace\to\Reals$ of total order $2\quadnum-1$ or less,
\[\label{eqn:exact-integration}
	\int_{\paramspace} \normaldens{\param}{0}{\identmat_\paramdim} \polyfunc(\param) \dee \param = \sum_{\quadpointvec \in \quadpointset} \normaldens{\quadpointvec}{0}{\identmat_\paramdim} \polyfunc(\quadpointvec) \weight(\quadpointvec),
\]
where $\normaldens{\param}{0}{\identmat_\paramdim}$ is the standard $\paramdim$-dimensional Gaussian density.
\end{definition}

The choice of the multivariate Gaussian density in \cref{def:exact-integrate-property} is strategic, since in many parametric models the posterior distribution asymptotically (in sample size) converges to a Gaussian distribution by the Bernstein von Mises theorem \citep[Chapter 10]{vaart_1998}.
For a non-standard limit, another baseline density can be used in place of the Gaussian in \cref{def:exact-integrate-property}, but we do not consider such models at present.

Consider the univariate case (i.e., $\paramdim=1$) with $\paramspace=\Reals$. For $\quadnum \in \N$, let $\hermite_\quadnum$ be the $\quadnum^{\text{th}}$ Hermite polynomial, defined for all $\quadpointidx \in \Reals$ by
\[\label{eqn:hermitepoly}
	\hermite_\quadnum(\quadpointidx) = (-1)^\quadnum e^{\quadpointidx^2/2} \frac{\dee^\quadnum e^{-\quadpointidx^2/2}}{\dee \quadpointidx^\quadnum},
\]
and denote its zeroes by $\hermitezero_1,\dots,\hermitezero_\quadnum$. These zeroes are distinct, symmetric about $0$, and include $0$ if and only if $\quadnum$ is odd. The Gauss--Hermite quadrature (\GHQ) rule uses the points $\quadpointset = (\hermitezero_j)_{j\in[\quadnum]}$ and the weight function
\[\label{eqn:weights}
	\weight(\hermitezero_j) = \frac{\quadnum!}{[H_{\quadnum+1}(\hermitezero_j)]^{2}\stdnormaldens{\hermitezero_j}},
\]
where $\stdnormaldens{\cdot}$ denotes the standard normal density. It is well known \citep[e.g., Eq. (3.6.11) of][]{numint} that \GHQ{} satisfies $\property{\quadnum}{1}$ using the smallest number of points $\quadnum$ possible.
Hence, a naive application of \GHQ{} for integrating $\genintfunc$ may be expected to perform well in the case that $\genintfunc$ is centred at $\paramidx = 0$ and $\genintfunc(\paramidx) / \stdnormaldens{\paramidx}$ is well-approximated by a univariate polynomial of total order $2\quadnum-1$ or less.
For a strict subset $\paramspace \subset \Reals$, there exist finite interval Gauss quadrature rules (Section~2.7 of \citeauthor{numint}, \citeyear{numint}; Theorem~1 of \citeauthor{bojanov01interval}, \citeyear{bojanov01interval}) that satisfy $\property{\quadnum}{\paramdim}$.

Quadrature rules in $\paramdim$ dimensions are formed by combining a univariate quadrature rule for each dimension. The most common combination technique is the \emph{product rule}. If $\quadpointsetidx \subseteq \Reals$ is a collection of univariate quadrature points with weight function $\weightidx:\Reals \to \Reals$, the product rule induced multivariate quadrature rule is defined by
$\quadpointset = \quadpointsetidx^\paramdim$ and $\weight(\quadpointvec) = \prod_{j=1}^\paramdim \weightidx(\quadpointidx_{i_j})$
for any $\quadpointvec = (\quadpointidx_{i_1},\dots,\quadpointidx_{i_\paramdim}) \in \quadpointset$. \GHQ{} with the product rule satisfies $\property{\quadnum}{\paramdim}$ \citep[see Section~5.6 of][]{numint}, and requires $\quadnum^\paramdim$ quadrature points. 
Sparse rules are also available that satisfy $\property{\quadnum}{\paramdim}$, see \cref{sec:convergence}.

In addition to $\property{\quadnum}{\paramdim}$, a commonly desired property of quadrature rules is that they are \emph{symmetric}, which we formalize now in the context of our problem. We note that \GHQ{} with both product rule and the sparse rule we consider is symmetric, and hence restricting ourselves to symmetric rules is benign.
\begin{definition}
$\quadrule{\quadpointset}{\weight}$ is \emph{symmetric} if for all $\quadpointvec\in\quadpointset$, $-\quadpointvec\in\quadpointset$ and $\weight(\quadpointvec) = \weight(-\quadpointvec)$.
\end{definition}

For Bayesian inference, we are interested in integrating functions that depend on $\data$, such as $\genintfunc(\param) = \poststar{\param}$.
A limitation of numerical quadrature rules is that the points and weights remain fixed regardless of the shape and location of $\poststar{\param}$. As $n\to\infty$, standard Bayesian asymptotic theory \citep[][Ch.\ 10]{vaart_1998} guarantees that with high probability over $\data$, the posterior mode concentrates to some $\paramtrue$ and the variance at the mode tends to $0$. 
Consequently, any fixed rule will miss most of the mass of $\poststar{\param}$, and this problem worsens as $n\to\infty$. 
A procedure that explicitly \emph{adapts} to the changing location and shape of the posterior density is necessary to obtain statistical performance guarantees that hold assuming only standard regularity conditions on the model.

\subsection{Adaptive Quadrature}\label{subsec:adaptivegausshermitequadrature}

\citet{nayloradaptive} introduced a technique that was eventually named \emph{adaptive} Gauss-Hermite quadrature (\AGHQ{}) by \citet{laplace}, which we extend here to adaptive quadrature in general.

Given a function $\genintfunc$ to integrate (which one expects is well-approximated by a Gaussian density), 
define the mode, curvature at the mode, and Cholesky decomposition of the inverse curvature by
\[\label{eqn:adaptive-objects}
	\genintmode = \argmax_{\param\in\paramspace}\genintfunc(\param);
	\quad
	\genloghess(\genintmode) = -\partial^{2}_{\param}\log\genintfunc(\genintmode);
	\quad
	\big[\genloghess(\genintmode)\big]^{-1} = \genlogchol\genlogchol\tpose.
\]
While we focus on the Cholesky decomposition for concreteness, any other matrix decomposition of this form could be used in our results.
For any quadrature rule $\quadrule{\quadpointset}{\weight}$, the adapted integral approximation under this rule is
\*[
	\quadadaptapproxint{\quadpointset}{\weight} = 
	\abssmall{\genlogchol} \sum\limits_{\quadpointvec\in\quadpointset} 
	\genintfunc(\genlogchol \, \quadpointvec + \genintmode)
	\, \weight(\quadpointvec) .
\]
\subsection{Approximate Bayesian Inference}\label{sec:approximatebayesianinference}
We will use adaptive quadrature to approximate three integrals for Bayesian inference: the normalizing constant to obtain an approximate posterior density, and then the further integration needed to obtain approximate marginal posterior densities and moments.
For any approximation $\approxdist(\data)$ of $\dist(\data)$, the approximate posterior distribution is %
\[\label{eqn:approxposterior}
	\approxpost{\param} = \frac{\poststar{\param}}{\widetilde{\dist}(\data)}.
\]

First, to approximate the normalizing constant, denote the analogous quantities of \cref{eqn:adaptive-objects} for the function $\genintfunc(\param) = \poststar{\param}$ by $\parammode$, $\logposthess(\parammode)$, and $\logpostchol$.
Then, for any quadrature rule $\quadrule{\quadpointset}{\weight}$, the adapted approximate normalizing constant under this rule is
\[
	\quadadaptmarg{\quadpointset}{\weight} = 
	\abssmall{\logpostchol} \sum\limits_{\quadpointvec\in\quadpointset} 
	\poststar{
	\logpostchol \, \quadpointvec + \parammode}
	\, \weight(\quadpointvec) .
\label{eqn:aghq-normalizing}
\]
When $\quadrule{\quadpointset}{\weight}$ is \GHQ{}, the adaptive form is \AGHQ{} by definition, and we denote the approximate normalizing constant by $\GHmarg$. When $\quadnum=1$, $\hermitezerovec = 0$ and $\weight(\hermitezerovec) = (2\pi)^{\paramdim/2}$, so the \AGHQ{} approximation is actually a Laplace approximation \citep{laplace} and may be applied 
to integrals of any dimension without computational difficulties. 
In \cref{sec:highdimexamples}, we provide an example of how to combine \AGHQ{} on low-dimensional parameters of interest with a Laplace approximation for high-dimensional nuisance parameters.

Second, to approximate the marginal posterior distribution, we apply \AGHQ{} twice.
Suppose the parameter can be decomposed into $\param = (\margparam, \nuis)$ where $\margparam\in\R^{\interestdim}$, and we are interested in computing the marginal posterior density at $\margparam=\margparamval$,
\*[
	\dist(\margparamval \setdelim \data) = \frac{\int \dist(\margparamval, \nuis,\data) \dee\nuis}{\int \dist(\param,\data)\dee \param}.
\]
Define 
 $\parammodemarg = \argmax_{\param\in\paramspace(\margparamval) }\dist(\margparamval, \nuis,\data)$ for $\paramspace(\margparamval) = \{\param\in\paramspace: \param=(\margparamval,\nuis)\}$, $\logposthessmarg(\param) = -\partial^{2}_{\nuis}\log\dist(\margparamval, \nuis,\data)$, and $\logposthessmarg(\parammodemarg)^{-1} = \logpostcholmarg(\logpostcholmarg)\tpose.$
Then, given a fixed $\quadrule{\quadpointset}{\weight}$ and $\quadrule{\quadpointset^\prime}{\weight^\prime}$ of dimensions $\paramdim$ and $\paramdim-\interestdim$ respectively, the approximate marginal density is
\[\label{eqn:approx-marg-dist-defn}
	\approxdist(\margparamval \setdelim \data)
	= \frac{\abssmall{\logpostcholmarg} \sum_{\quadpointvec^\prime \in\quadpointset^\prime } \dist\left( (0, \logpostcholmarg \, \quadpointvec\Tr)\Tr  
	+ \parammodemarg, \data \right) \weight^\prime(\quadpointvec^\prime) }{\abssmall{\logpostchol} \sum_{\quadpointvec\in\quadpointset} 
	\poststar{
	\logpostchol \, \quadpointvec + \parammode}
	\, \weight(\quadpointvec)},
\]
which we denote by $\GHapprox{\margparamval}$ when both rules correspond to \AGHQ{}.

Third, let $\meanfunc:\paramspace \to \PosReals$ be any nonnegative function satisfying $\int \meanfunc(\param)\dist(\param)\dee \param < \infty$ . 
Denote the analogous quantities of \cref{eqn:adaptive-objects} for the function $\genintfunc(\param) = \distfunc(\param,\data) = \dist(\param, \data)\meanfunc(\param)$ by
 $\parammodefunc$, $\logposthessfunc(\param)$, and $\logpostcholfunc$, when they exist 
 (see \cref{sec:proof-pos-mean-compute}).
Then, given a fixed $\quadrule{\quadpointset}{\weight}$, define
\[\label{eqn:approx-moment-defn}
	\approxEE[\meanfunc(\param) \setdelim \data]
	= \frac{\abssmall{\logpostcholfunc} \sum_{\quadpointvec\in\quadpointset} \distfunc(\logpostcholfunc \, \quadpointvec + \parammodefunc, \data) \weight(\quadpointvec)}{\abssmall{\logpostchol} \sum_{\quadpointvec\in\quadpointset} 
	\poststar{
	\logpostchol \, \quadpointvec + \parammode}
	\, \weight(\quadpointvec)},
\]
which we denote by $\GHapproxEE[\meanfunc(\param) \setdelim \data]$ when the rule is \AGHQ{}.

%% file: sections/convergence-with-statistics.tex
\section{Convergence Rates}\label{sec:convergence}

In this section,
we provide stochastic convergence rates for adaptive quadrature applied to Bayesian inference as well as stochastic convergence rates for various summary statistics of inferential interest. 
All proofs are deferred to the \preprint{appendix}\jasa{supplement}.
We denote probability under the true data-generating distribution by $\datatrueprobn$, and use $\constt$ to denote a generic constant in $n$ that may otherwise depend on $\paramdim$, $\quadnum$, and the universal constants in \cref{sec:assumptions}.

\subsection{Approximate Posterior}

\begin{restatable}{theorem}{MainResult}
\label{thm:mainresult}\label{THM:MAINRESULT}
Suppose there exists $\derivnum\geq 4$ such that the likelihood of the data $\dist(\data \setdelim \param)$ is $\derivnum$-times differentiable as a function of $\param$ and the regularity assumptions of \cref{sec:assumptions} hold. 
For $1 \leq \quadnum \leq \lfloor \derivnum/2 \rfloor$,
if $\quadrule{\quadpointset}{\weight}$ is a symmetric quadrature rule satisfying $\property{\quadnum}{\paramdim}$ then
\*[
	\lim_{n\to\infty}\datatrueprobn\left( \absbig{\frac{\dist(\data)}{\quadadaptmarg{\quadpointset}{\weight}} - 1} \leq \constt \, n^{-\lfloor \frac{\quadnum+2}{3} \rfloor}\right) = 1.
\]
\end{restatable}

\begin{remark}
For \AGHQ{} with $\quadnum=1$,
\cref{THM:MAINRESULT} recovers the known $\Ordp(n^{-1})$ rate for the Laplace approximation \citep{validitylaplace}.
\end{remark}

\begin{remark}
If $\quadnum > \lfloor \derivnum/2 \rfloor$ the rate for $\quadnum = \lfloor \derivnum/2 \rfloor$ applies; this can be seen by reproducing the proof with a Taylor expansion of order $\derivnum$ rather than one of order $2\quadnum$.
\end{remark}

The following corollary demonstrates the utility of \cref{thm:mainresult}, and follows immediately from the algebra of \cref{sec:ABI-accuracy} and the definition of $\quadadaptapprox{\quadpointset}{\weight}{\param}$.
\begin{corollary}\label{cor:rateofothererrors}
Under the conditions of \cref{thm:mainresult},
\*[
	\lim_{n\to\infty}\datatrueprobn\left( \sup_{\param\in\paramspace}\absbig{\frac{\post{\param}}{\quadadaptapprox{\quadpointset}{\weight}{\param}} - 1} \leq \constt \, n^{-\lfloor \frac{\quadnum+2}{3} \rfloor}\right) = 1
\]
and
\*[
	\lim_{n\to\infty}\datatrueprobn\left( \norm{\post{\cdot} - \quadadaptapprox{\quadpointset}{\weight}{\cdot}}_{\mathrm{TV}} \leq \constt \, n^{-\lfloor \frac{\quadnum+2}{3} \rfloor}\right) = 1.
\]
\end{corollary}

The results of this section apply to any symmetric quadrature rule satisfying $\property{\quadnum}{\paramdim}$.
In our applications, we focus on \AGHQ{} defined using the product rule due to its simplicity and the fact that we have provided a robust implementation in the {\upshape\texttt{aghq}} package.
It is of interest to compare alternatives empirically, such as the nested rule considered by \cite{genz1996fully}, which is a Gaussian extension of the Gauss–Kronrod–Patterson construction.

Further, for multidimensional posteriors, our theoretical results apply to symmetric quadrature rules based on ``sparse grids'' as long as they satisfy $\property{\quadnum}{\paramdim}$. For example, Smolyak's quadrature rule satisfies this criteria \citep[Theorem 1]{heiss08sparse}, which reduces the dependence on the dimension from exponential to polynomial asymptotically. For specific $\quadnum$ and $\paramdim$, however, sparse rules may actually be more computationally intensive than the product rule; for example, when $\quadnum=5$ and $\paramdim=2$, the product rule uses 25 quadrature points while the Smolyak rule uses 55. 

The convergence rate depends directly on the number of quadrature points as follows. If one uses a product rule extension, then $\quadnum^\paramdim$ quadrature points are needed to satisfy $\property{\quadnum}{\paramdim}$. Hence, if the model is of dimension $\paramdim$ and one uses a product rule extension with $\abssmall{\quadpointset}$ quadrature points, the convergence rate will be $n^{-\lfloor (\abssmall{\quadpointset}^{1/\paramdim}+2)/3 \rfloor}$.

\subsection{Approximate Posterior Summaries}\label{sec:convergence-summary}

We now show that the convergence rate of \cref{thm:mainresult} is realized for the approximations to marginal distributions and moments.
Refer to \cref{sec:approximatebayesianinference} for defintions and see \cref{app:computesummaries} for computational details.

\begin{restatable}{corollary}{PosMargCompute}
\label{fact:pos-marg-compute}
Fix the value of $\margparamval$. 
Suppose the conditions in \cref{thm:mainresult} are satisfied when replacing all instances of $\param$ with $(\margparamval, \nuis)$,  $\paramtrue$ with some constant $\paramtrue_\margparamval = (\margparamval, \nuis_\margparamval^\star)$, $\logposthess$ with $\logposthessmarg$ and $\ball{\paramdim}{\paramtrue}{\cdot}$ with $\ball{\paramdim - \interestdim}{\nuis^\star_\margparamval}{\cdot}$.
Then 
\*[
	\lim_{n \to \infty}
	\datatrueprobn
	\bigg(
	\absbig{\frac{\dist(\margparamval \setdelim \data)}{\approxdist(\margparamval \setdelim \data)} - 1} \leq \constt \, \, n^{-\lfloor \frac{\quadnum+2}{3} \rfloor}
	\bigg)
	= 1. 
\]
\end{restatable}

\begin{restatable}{corollary}{PosMeanCompute}
\label{fact:pos-mean-compute}
Suppose $\meanfunc:\paramspace \to \PosReals$ satisfies 
assumptions (M1) through (M3) from \cref{sec:proof-pos-mean-compute}.
Then, if the conditions of \cref{thm:mainresult} also hold,
\*[
	\lim_{n \to \infty}
	\datatrueprobn
	\bigg(
	\absbig{\frac{\EE[\meanfunc(\param) \setdelim \data]}{\approxEE[\meanfunc(\param) \setdelim \data]} - 1} \leq \constt \, \, n^{-\lfloor \frac{\quadnum+2}{3} \rfloor}
	\bigg)
	= 1. 
\]
\end{restatable}

Both \cref{fact:pos-marg-compute,fact:pos-mean-compute} require additional assumptions to be verified. In \cref{sec:proof-pos-marg-compute,sec:proof-pos-mean-compute}, we show that (a) \cref{fact:pos-marg-compute} applies to all values of $\margparamval$ in a $n^{-1/2}$-neighbourhood of the unrestricted posterior mean (\cref{fact:pos-marg-vals}) and (b) \cref{fact:pos-mean-compute} applies to all marginal posterior moments (\cref{fact:pos-mean-compute-verify}).

\subsection{Proof Sketch of \cref{THM:MAINRESULT}}

Finally, we provide a brief sketch of the proof of \cref{thm:mainresult}
to
highlight the intuition for how convergence rates of the posterior inform the ultimate approximation error rate, 
and contrast our result with previous analyses of adaptive quadrature rules.

\begin{proof}[Proof sketch (informal) of \cref{thm:mainresult}]
The proof quantifies the rate at which the posterior behaves locally Gaussian with polynomial error, combines this with the polynomial exactness property $\property{\quadnum}{\paramdim}$, and then quantifies that the contribution to the posterior mass outside of this local neighbourhood is negligible with high probability.
Specifically, the proof of \cref{thm:mainresult} is composed of demonstrating that the following key facts hold with high probability asymptotically. We refer to the corresponding formal statements by their location in the \preprint{appendix}\jasa{supplement}.

\begin{enumerate}[(1)]
\item There exists a neighbourhood centered at a fixed parameter with radius defined by the curvature of the likelihood such that the likelihood is exponentially small outside of the neighbourhood. See \cref{lem:fixed_numer} for the precise statement.
\item Within this neighbourhood, there exists a smaller neighbourhood with radius decaying at rate $\sqrt{\log(n)/n}$ such that the likelihood is polynomially small within the annulus outside of the shrinking neighbourhood.
See \cref{lem:shrinking_numer} for the precise statement.
\item A Taylor series expansion (with order depending on $\quadnum$) of the unnormalized posterior provides an accurate polynomial approximation within the shrinking neighbourhood. See the proof of \cref{lem:normalizing_constant_as} for details. \qedhere
\end{enumerate}
\end{proof}

The two most relevant works to our result are \citet{validitylaplace} and \citet{adaptive_GH_2020}. We now contrast our proof with the analyses in both to highlight our technical contribution. 
First, we note that all of these works have only proved results for specifically \AGHQ{}, while we have distilled the rate down to a simpler set of assumptions satisfied by more rules.
Second, \citet{validitylaplace} only prove the $\quadnum=\paramdim=1$ case. They remark $\paramdim>1$ is trivial, however, multivariate Taylor expansions lead to a product of sums rather than simply a sum of products, which consequently must be further upper bounded (see \cref{eqn:post-cancel-all-terms}). 

Third, for the $\quadnum>1$ case, higher-order derivatives are required, and \citet{adaptive_GH_2020} sketch a proof for a limited class of functions in this setting. However, their analysis requires limiting assumptions that rule out posterior functions. 
Specifically, \citet{adaptive_GH_2020} \citep[inheriting the assumptions of][]{adaptive_GH_1994} only allow integrands of the form $\exp\{n \ell(\param)\}$ rather than $\exp\{n \llhood(\param)\}$, meaning $\ell$ cannot depend on $n$ and consequently also cannot depend on data; this eliminates \emph{all} log-likelihoods. 

Finally, neither \citet{adaptive_GH_2020} nor \citet{adaptive_GH_1994} provide an explicit argument for the order of the remainder terms from a Taylor series expansion. 
Handling these remainder terms is highly nontrivial, as can be seen by the proofs in \citet{validitylaplace} and \preprint{our appendix}\jasa{our supplement}. Specifically, this requires a) identifying whether the remainder term is odd or even, b) controlling the higher-order derivatives of the likelihood at various distances from the target parameter \citep[undefined for][]{adaptive_GH_2020}, and c) applying probabilistic concentration results depending on the order of the remainder \citep[also undefined for][who take a deterministic approach]{adaptive_GH_2020}.

%% file: sections/examples-low-dim.tex
\section{Low-Dimensional Parameter Spaces}\label{sec:exampleslowdim}

In the next two sections, we complement the theoretical results of \cref{sec:convergence} through three challenging examples, demonstrating the attractive computational properties of \AGHQ{} for approximate Bayesian inference. In all of the examples, the quadrature rule we use is \AGHQ{} with the product rule extension to multiple dimensions. Because \MCMC{} is arguably the most widely researched method for making approximate Bayesian inferences, and enjoys robust implementation in open-source software, we pay attention to the practical advantages of \AGHQ{} compared to state of the art \MCMC{} methods for the chosen examples.

\subsection{Example: Modelling Infectious Disease Spread}\label{subsec:infectiousdisease}

We consider the popular \emph{Susceptible, Infectious, Removed} (SIR) model for infectious disease spread as implemented in the {\tt EpiILMCT} package in {\tt R} \citep{epi}. Despite the low dimension of the parameter space, \MCMC{} is the methodology of choice for fitting these models, leading to long run times and the need for specialized tuning and practical assessment of convergence. We demonstrate here that \AGHQ{} gives fast and stable results that closely match the output of \MCMC{} in a small fraction of the run time.

\citet{epi} consider an outbreak of Tomato Spotted Wilt Virus in $n = 520$ plants. Plants were grown on an even grid and checked for the virus every 14 days, a total of 7 times. There were $n_{0} = 327$ plants infected by the end of the study period. For each plant we observe infection times $I_{1} \leq \ldots \leq I_{n_{0}}$ and $I_{i} = \infty$ for $i = n_{0} + 1,\ldots,n$. Plants may infect other plants while they are infected, and we observe associated removal times $R_{i},i\in[n]$ where a plant is no longer infectious. The likelihood function for these observed infection and removal times is given by
\*[
\dist(\mb{I},\mb{R} | \alpha,\beta) = \prod_{j=2}^{n_{0}}\left( \sum_{i:I_{i} < I_{j} \leq R_{i}}\lambda_{ij}\right)\exp\bracevec{-\sum_{i=1}^{n_{0}}\sum_{j=1}^{n}[\text{min}(R_{i},I_{j}) - \text{min}(I_{i},I_{j})]\lambda_{ij}},
\]
where $\lambda_{ij} = \alpha d_{ij}^{-\beta}$
is the \emph{infectivity rate}: the rate at which an infectious plant $i$ passes the disease to a susceptible plant $j$. Here $d_{ij}$ is the Euclidean distance between plants $i$ and $j$, and $\alpha,\beta > 0$ are the parameters of inferential interest. Independent $\text{Exponential}(.01)$ priors are placed on $\alpha,\beta$. As discussed in \cref{sec:computational}, we transform the parameters as $\theta_{1} = \log\alpha$ and $\theta_{2} = \log\beta$, perform the quadrature on this transformed scale, and then transform back when reporting results.

\begin{figure}
\centering
	\subfloat[$\alpha, k = 3$]{\includegraphics[width=.3\textwidth]{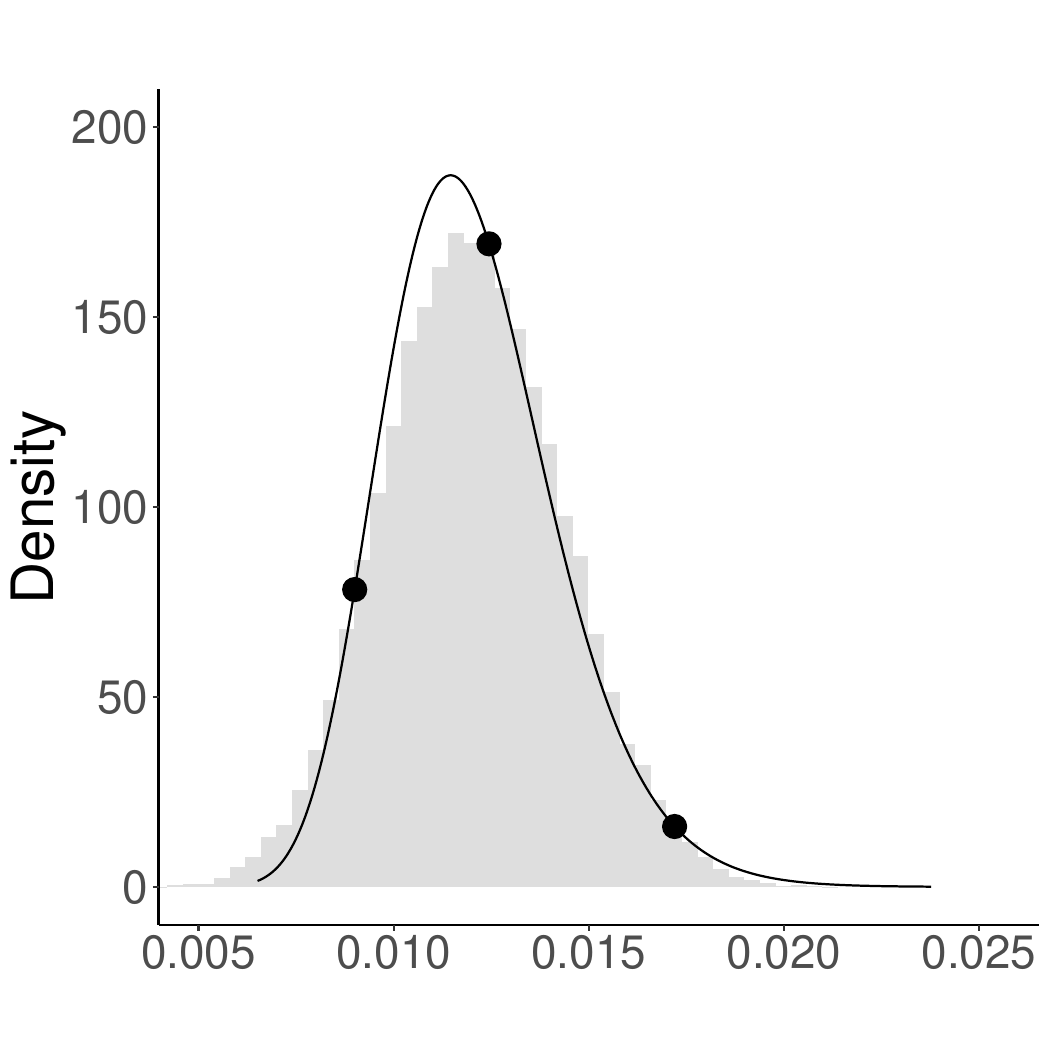}}
	\subfloat[$\alpha, k = 5$]{\includegraphics[width=.3\textwidth]{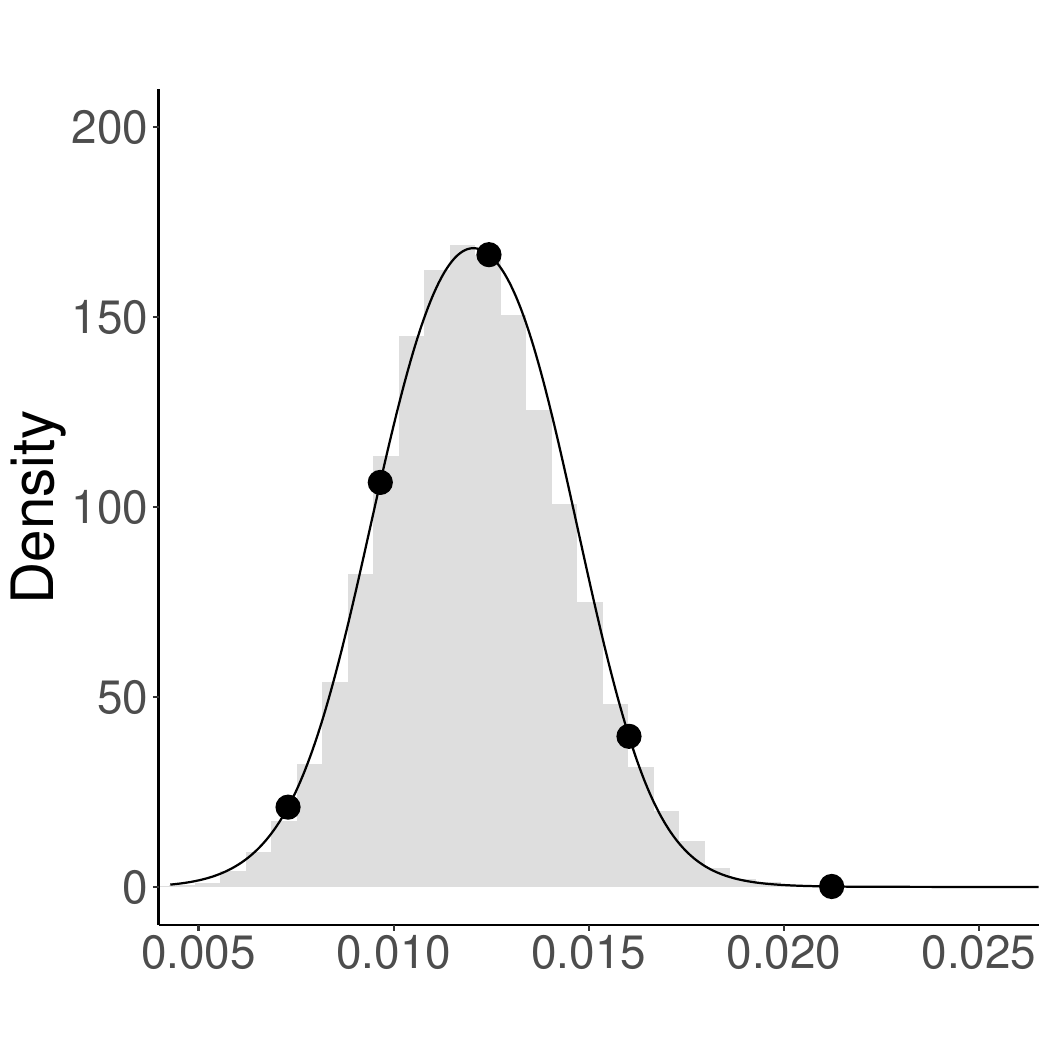}}
	\subfloat[$\alpha, k = 7$]{\includegraphics[width=.3\textwidth]{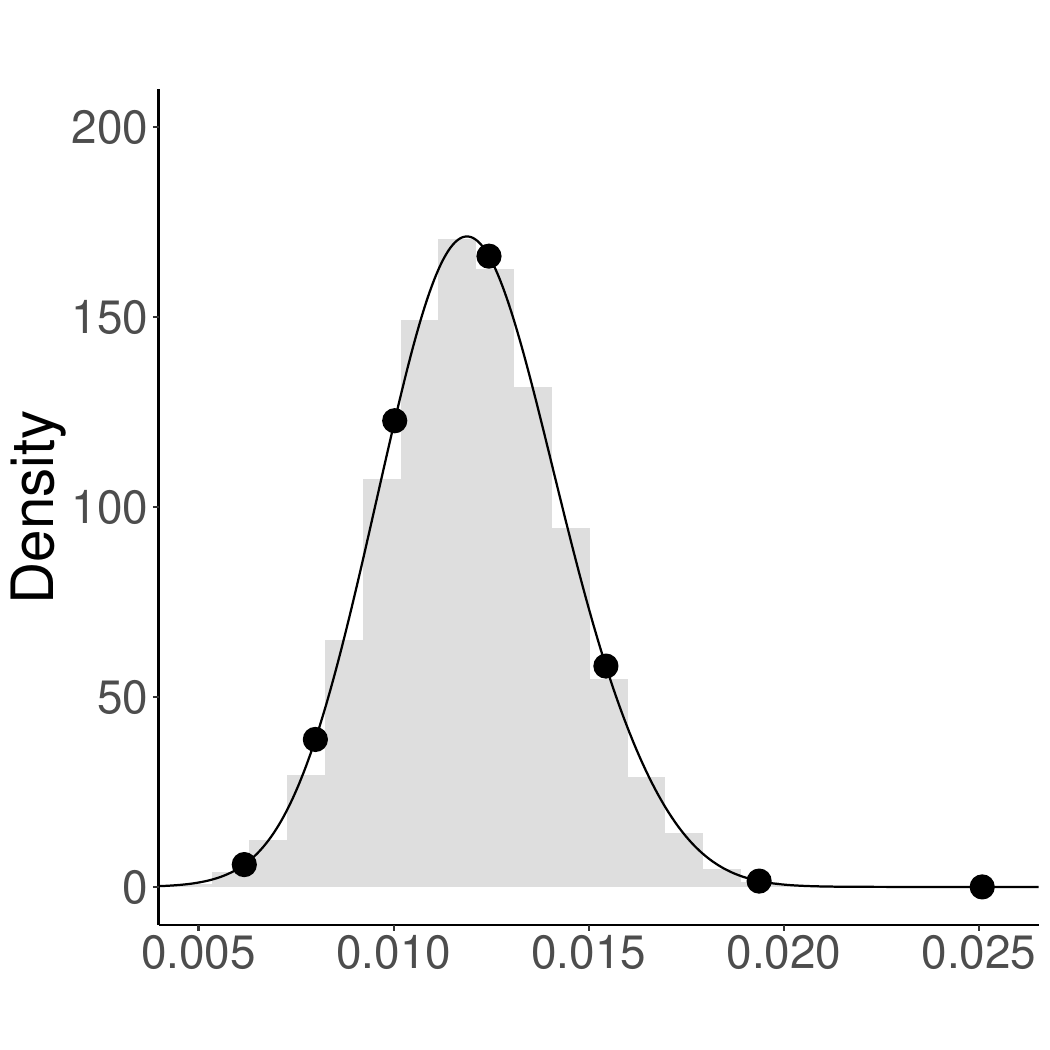}}\\
	\subfloat[$\beta, k = 3$]{\includegraphics[width=.3\textwidth]{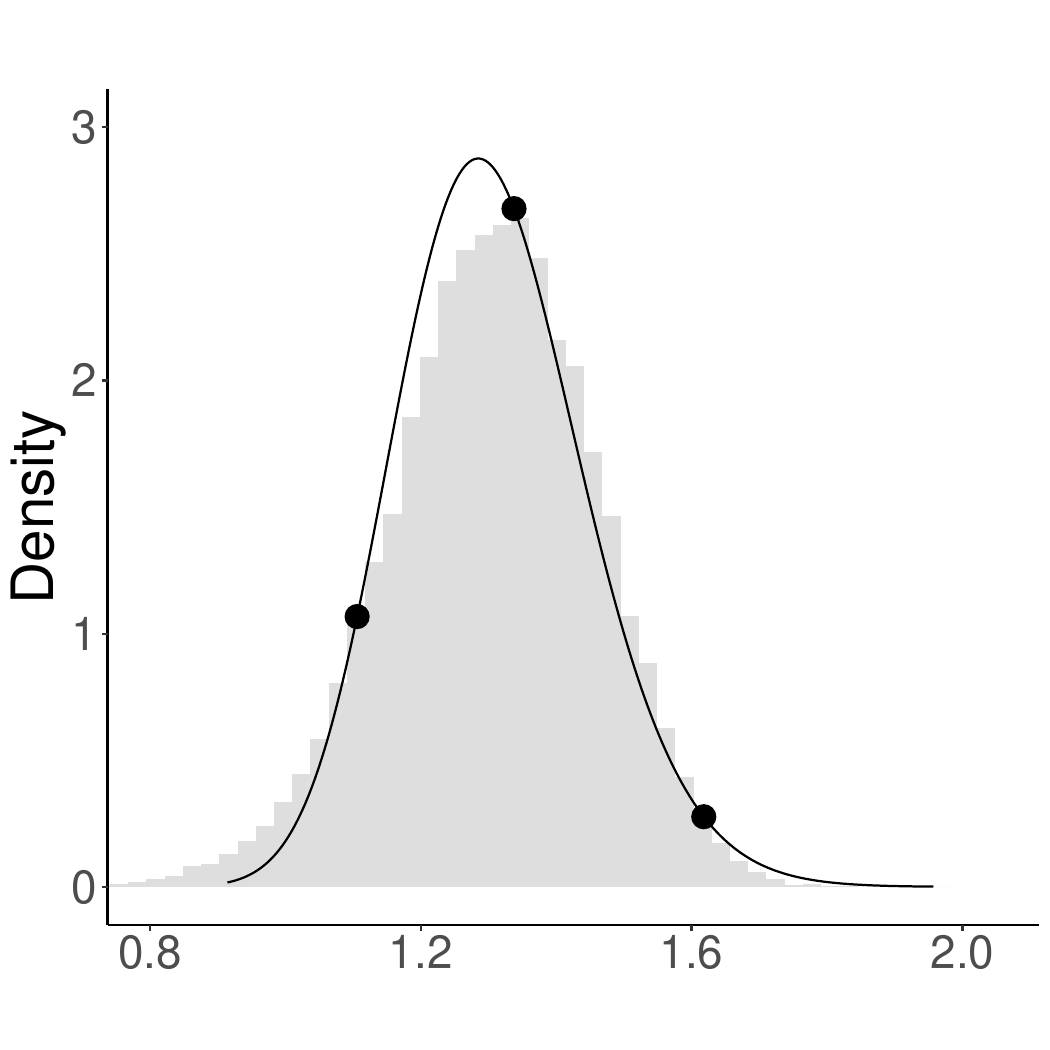}}
	\subfloat[$\beta, k = 5$]{\includegraphics[width=.3\textwidth]{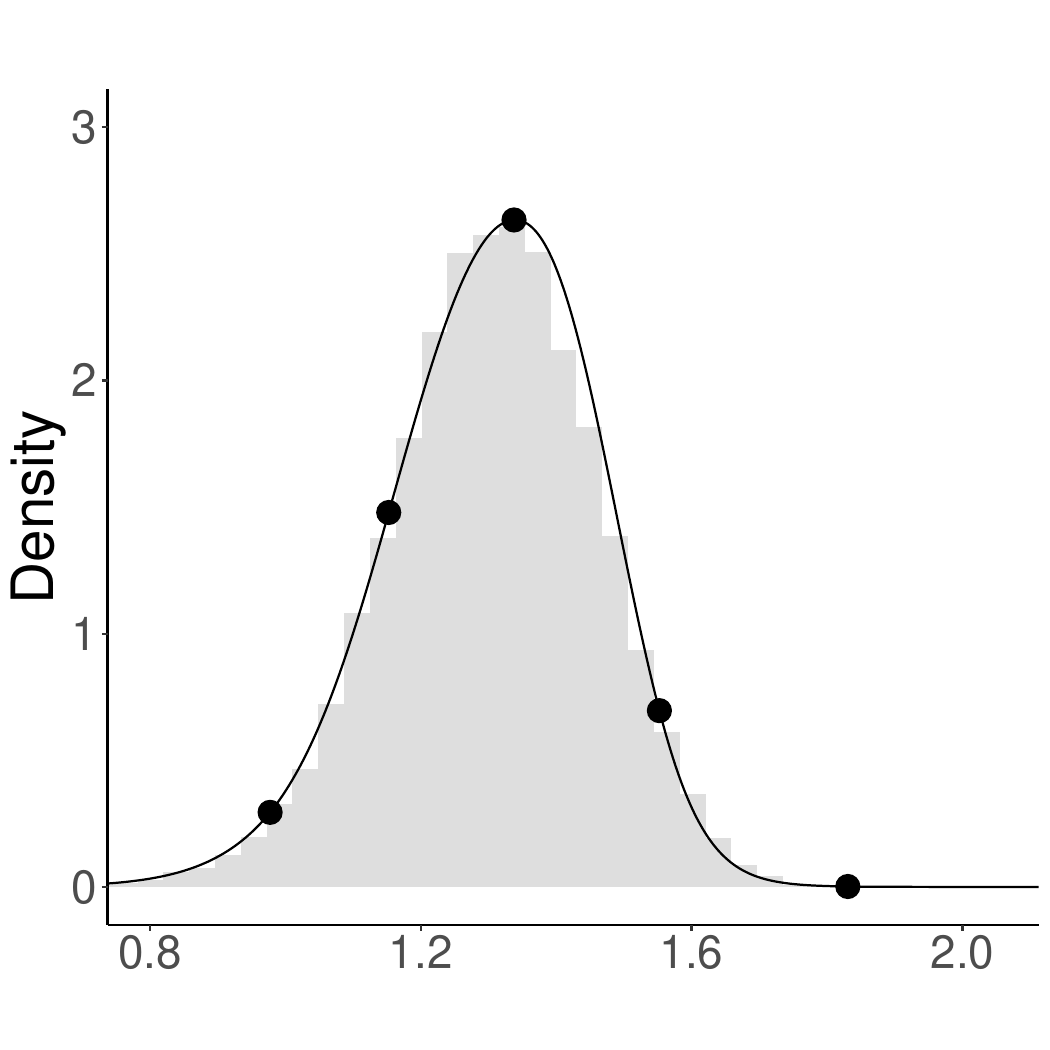}}
	\subfloat[$\beta, k = 7$]{\includegraphics[width=.3\textwidth]{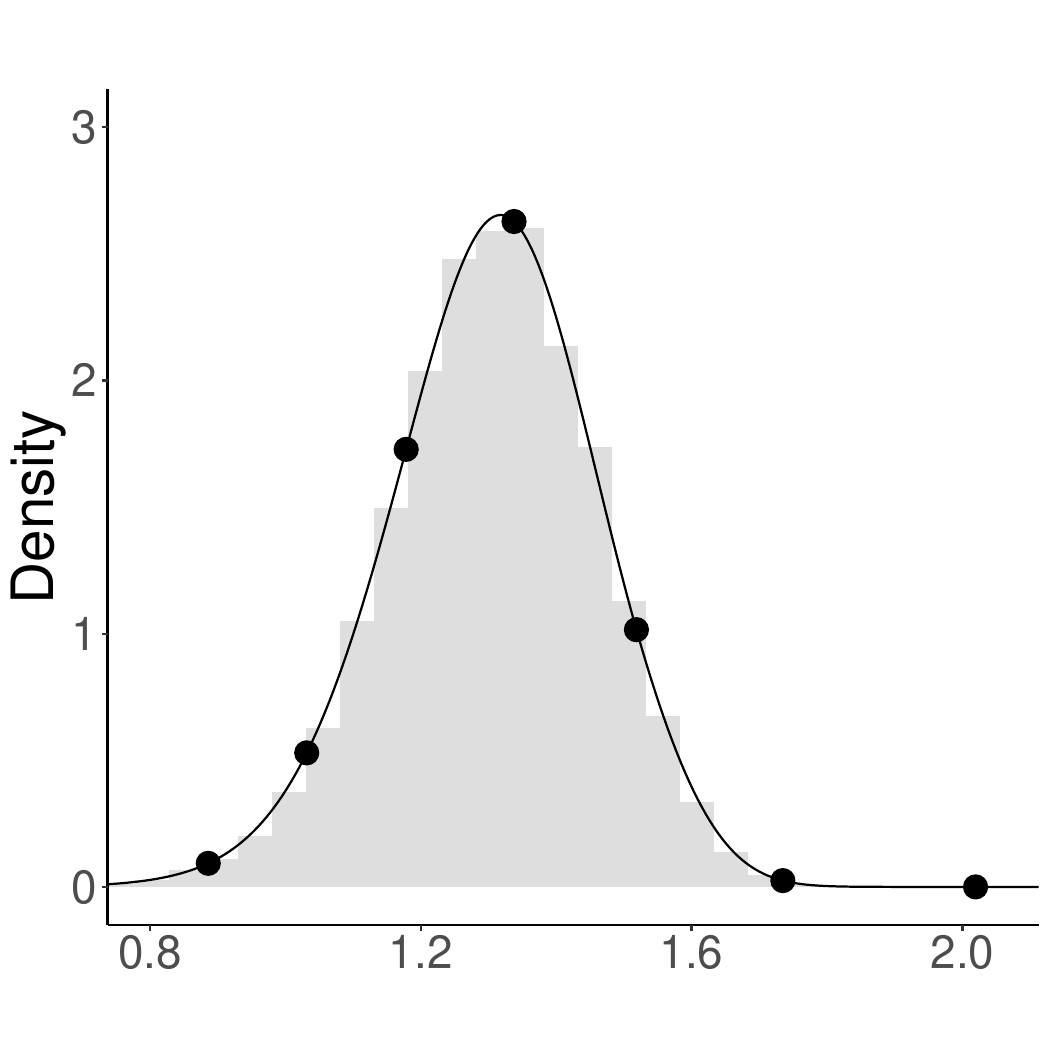}}
\caption{\AGHQ{} ($\bullet$, ---) and \MCMC{} (\textcolor{lightgray}{$\blacksquare$}) results for the infectious disease data of \cref{subsec:infectiousdisease}.}
\label{fig:diseasemcmccompare}
\end{figure}

\newcolumntype{?}{!{\vrule width 1pt}}
\begin{table}
\centering
\caption{\label{tab:diseasetable} Posterior summaries are reported for \AGHQ{} applied to the infectious disease data of \cref{subsec:infectiousdisease}. Comparison with \MCMC{} is reported using the Kolmogorov-Smirnov (KS) distance. Mean, SD, and quantiles for $\alpha$ are multiplied by 100.
}
\begin{tabular}{crrrrrrrrrr}
\hline
 & \multicolumn{2}{r}{Mean} & \multicolumn{2}{r}{SD} & \multicolumn{2}{r}{$2.5\%$} & \multicolumn{2}{r}{$97.5\%$} & \multicolumn{2}{c}{KS} \\
$\quadnum$ & $\alpha$ & $\beta$ & $\alpha$ & $\beta$ & $\alpha$ & $\beta$ & $\alpha$ & $\beta$ & $\alpha$ & $\beta$ \\
\hline
3 & 1.21 & 1.31 & 0.239 & 0.148 & 0.829 & 1.06 & 1.70 & 1.63 & 0.0326 & 0.0362 \\
5 & 1.20 & 1.30 & 0.232 & 0.152 & 0.750 & 0.982 & 1.60 & 1.55 & 0.0234 & 0.0258 \\
7 & 1.20 & 1.30 & 0.233 & 0.153 & 0.758 & 0.984 & 1.67 & 1.59 & 0.0133 & 0.0129 \\
9 & 1.20 & 1.30 & 0.233 & 0.153 & 0.759 & 0.985 & 1.66 & 1.58 & 0.0139 & 0.0131 \\
11 & 1.20 & 1.30 & 0.233 & 0.153 & 0.758 & 0.984 & 1.66 & 1.58 & 0.0126 & 0.0158 \\
13 & 1.20 & 1.30 & 0.233 & 0.153 & 0.757 & 0.984 & 1.66 & 1.58 & 0.0168 & 0.0157 \\
\MCMC{} & 1.20 & 1.30 & 0.228 & 0.151 & 0.761 & 0.986 & 1.65 & 1.58 & - & - \\
\hline
\end{tabular}
\end{table}

\begin{table}
\caption{\label{tab:diseasetableruntimes} Median (over 100 replications) run times to compute the marginal posteriors of both $\alpha$ and $\beta$ using \AGHQ{} for the infectious disease data of \cref{subsec:infectiousdisease}. Effective iterations are the number of \MCMC{} iterations that could have been performed in the same time it took to run \AGHQ{}, calculated based on a maximum run time of $84$ seconds for $10,000$ \MCMC{} iterations using $4$ parallel chains in \texttt{tmbstan}.}
\centering
\begin{tabular}{crrrrrr}
\hline
$\quadnum$ & 3 & 5 & 7 & 9 & 11 & 13 \\
\hline
Time (Seconds) 
& 0.101 
& 0.160 
& 0.224 
& 0.285 
& 0.357 
& 0.441 \\
Effective Iterations 
& 12
& 19
& 27
& 34 
& 42 
& 52 \\
\hline
\end{tabular}
\end{table}

\cref{fig:diseasemcmccompare} shows the posterior density estimates obtained using $\quadnum = 3$, $5$, and $7$, and a comparison to a long \MCMC{} run. The $\quadnum = 5$ and $7$ results are visually indistinguishable from the density obtained through \MCMC{} (and each other). Table \ref{tab:diseasetable} makes this more precise, with comparisons of posterior summaries of interest for $\quadnum = 3,5,7,9,11,13$. 

\cref{tab:diseasetableruntimes} shows the dramatic improvement in run time of \AGHQ{}, as measured by the number of \MCMC{} iterations (not including any time spent tuning the sampler) that could have been run in the same amount of time as it took to run the full \AGHQ{} procedure. Running \MCMC{} for the maximum such number of iterations resulted in all such iterations being marked as divergent and \texttt{NaN} estimates for the number of effective parameters. This demonstrates the substantial computational gains attained by \AGHQ{} in this simple example when compared to \MCMC{}.

In practice, choosing $\quadnum$ remains an open question. 
As helpfully suggested by a referee, one strategy is to fit the model with successively increasing $\quadnum$ until inferences no longer change with $\quadnum$. \cref{tab:diseasetable} shows this occurring for the infectious disease example, where up to $\quadnum=13$ was fit, with similar estimates from about $\quadnum=7$ or so. Adding up the first row of \cref{tab:diseasetableruntimes}, we see that the total time for this entire strategy is about $1.568$ seconds, or 186 total \MCMC{} iterations, still a dramatic computational gain.

\subsection{Example: Estimating the Mass of the Milky Way}\label{subsec:astro}

Estimating the mass of the Milky Way Galaxy (hereafter the ``Galaxy'') is of importance to astrophysicists interested in determining the amount of Dark Matter in the universe, among other things. \citet{gwen2} describe a probabilistic model for estimating and, importantly, quantifying uncertainty in the mass of the Galaxy using Bayesian inference. They use three-dimensional observed position and velocity measurements of star clusters in orbit of the Galaxy within a probabilistic physical model whose parameters determine the mass of the Galaxy at any radial distance from its centre. The parameters are subject to nonlinear constraints and are found to have strongly correlated, highly skewed posteriors with mode lying on or near the boundary of the parameter space \citep{gwen2}. Care is required in implementing \AGHQ{} for this problem.

The choice of priors was observed to have a substantial effect on inference in this problem \citep{gwen2}, and a large body of knowledge on how to do this is available from the underlying physics. \citet{gwen2} consider many different choices of priors and subsets of their data and the effect that this has on the estimated mass of the Galaxy. They use \MCMC{} for inference, where each new model fit in their application requires careful tuning and assessment of the chains as well as potentially inconvenient run times. We find that \AGHQ{} exhibits fast and stable performance in this challenging problem, although we note that our present implementation with \texttt{tmbstan} \citep[software that was not available at the time][was written]{gwen2} seems to avoid some of the reported challenges with \MCMC{} as well. Nonetheless, this example serves to illustrate the application of \AGHQ{} in a challenging applied problem.

Let $\dataidx_{i} = (y_{i1},y_{i2},y_{i3})$ denote the three measurements for each star cluster: position, radial velocity, and tangential velocity relative to the centre of the Galaxy (referred to as \emph{galactocentric} measurements), and let the full matrix of data be $\data = \bracevec{\dataidx_{i}:i\in[n]}$. There are $n = 70$ clusters with complete measurements. The probability density for $\dataidx_{i}$ is $$f(\dataidx_{i};\Psi_{0},\gamma,\alpha,\beta) = \frac{L_{i}^{-2\beta}\VE_{i}^{\frac{\beta(\gamma-2)}{\gamma} + \frac{\alpha}{\gamma} - \frac{3}{2}}\Gamma\left( \frac{\alpha}{\gamma} - \frac{2\beta}{\gamma} + 1\right)} {\sqrt{8\pi^{3}2^{-2\beta}} \Psi_{0}^{-\frac{2\beta}{\gamma} + \frac{\alpha}{\gamma}}\Gamma\left( \frac{\beta(\gamma - 2)}{\gamma} + \frac{\alpha}{\gamma} - \frac{1}{2}\right)},$$ where $L_{i} = y_{i1}y_{i3}$, $\VE_{i} = \Psi_{0}y_{i1}^{1-\gamma} - (y_{i2}^{2} + y_{i3}^{2})/2$, and $i\in[n]$. The parameters $\Xi = (\Psi_{0},\gamma,\alpha,\beta)$ determine the mass of the Galaxy at radial distance $r$ kiloparsecs (kpc) from its centre according to $M(r) = \Psi_{0}\gamma r^{1-\gamma}$. While $M(r)$ only directly depends on $\Psi_{0}$ and $\gamma$, its posterior will depend indirectly on all four parameters due to correlation between them.

\citet{gwen2} consider many different strongly informative priors for the four model parameters. We choose one configuration of theirs: $\Psi_{0}\sim\text{Unif}(1,200)$, $\gamma\sim\text{Unif}(0.3,0.7)$, $\alpha - 3 \sim\text{Gamma}(1,4.6)$, and $\beta\sim\text{Unif}(-0.5,1)$. The parameters are further subject to nonlinear constraints $\alpha > \gamma, \alpha > \beta(2 - \gamma) + \gamma/2$, and $\VE_{i} > 0, i\in[n]$.

We find the following transformations convenient in this example:
\*[
\theta_{j} = \log\left(-\log\left[ \frac{\Xi_{j} - a_{j}}{b_{j} - a_{j}}\right]\right), j=1,2,4, \qquad \theta_{3} = \log(\alpha - 3), \\
\]
where $(a_{j},b_{j})$ are the endpoints of the uniform prior for $\Xi_{j}$. We let $\mb{\theta} = (\theta_{1},\theta_{2},\theta_{3},\theta_{4})$ and normalize the posterior $\pi(\mb{\theta}|\data)$ using \AGHQ{} with $\quadnum=5$. We emphasize that these transformations are not required to apply the theoretical results of \cref{sec:convergence}, and refer the reader to \cref{sec:computational} for further discussion about implementation. To find the posterior mode accounting for the remaining non-linear constraints, we perform a constrained optimization using the \texttt{IPOPT} software \citep{ipopt} with derivatives of the log-likelihood and constraints provided by the \texttt{TMB} software \citep{tmb}.

\begin{figure}[h!]
\centering
\subfloat[$\Psi_{0}$]{\includegraphics[width=0.3\textwidth]{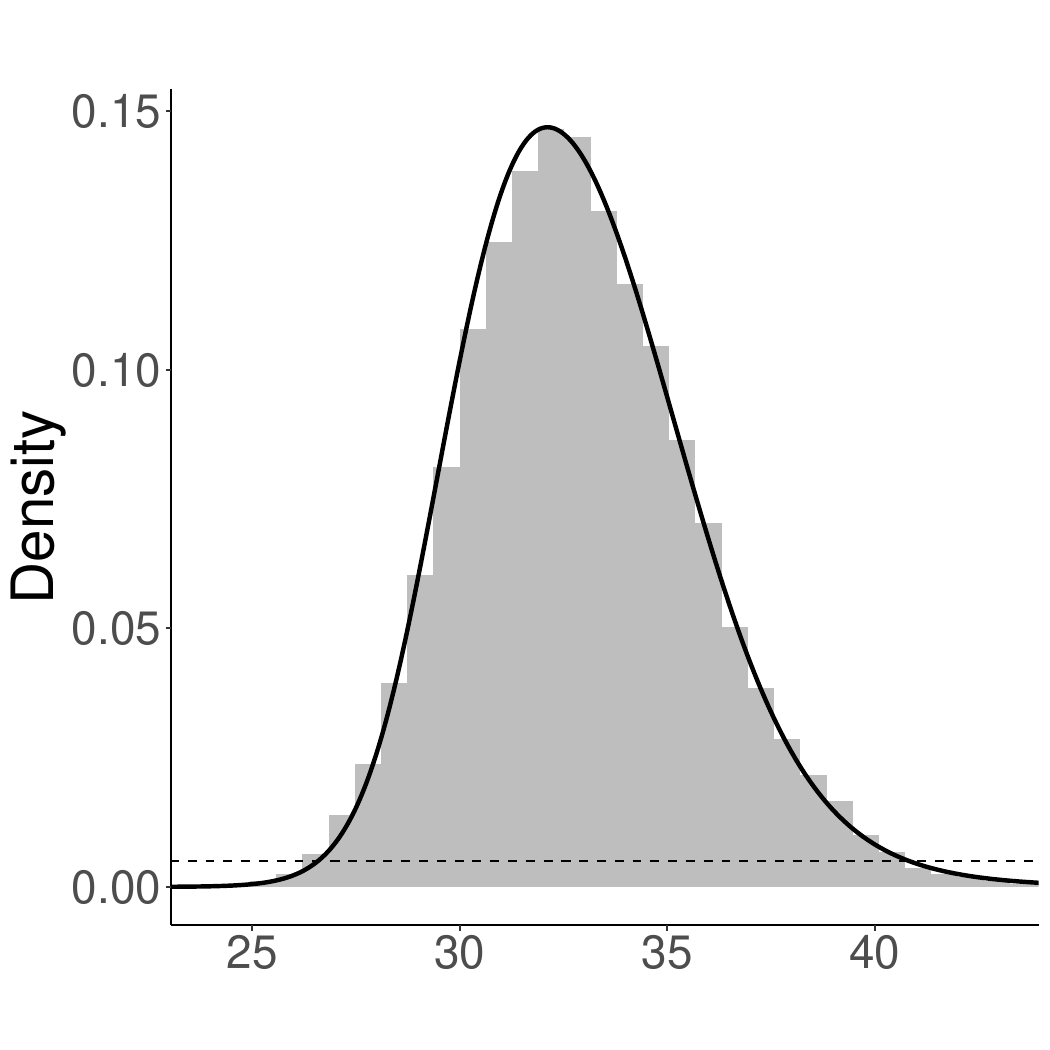}}
\subfloat[$\gamma$]{\includegraphics[width=0.3\textwidth]{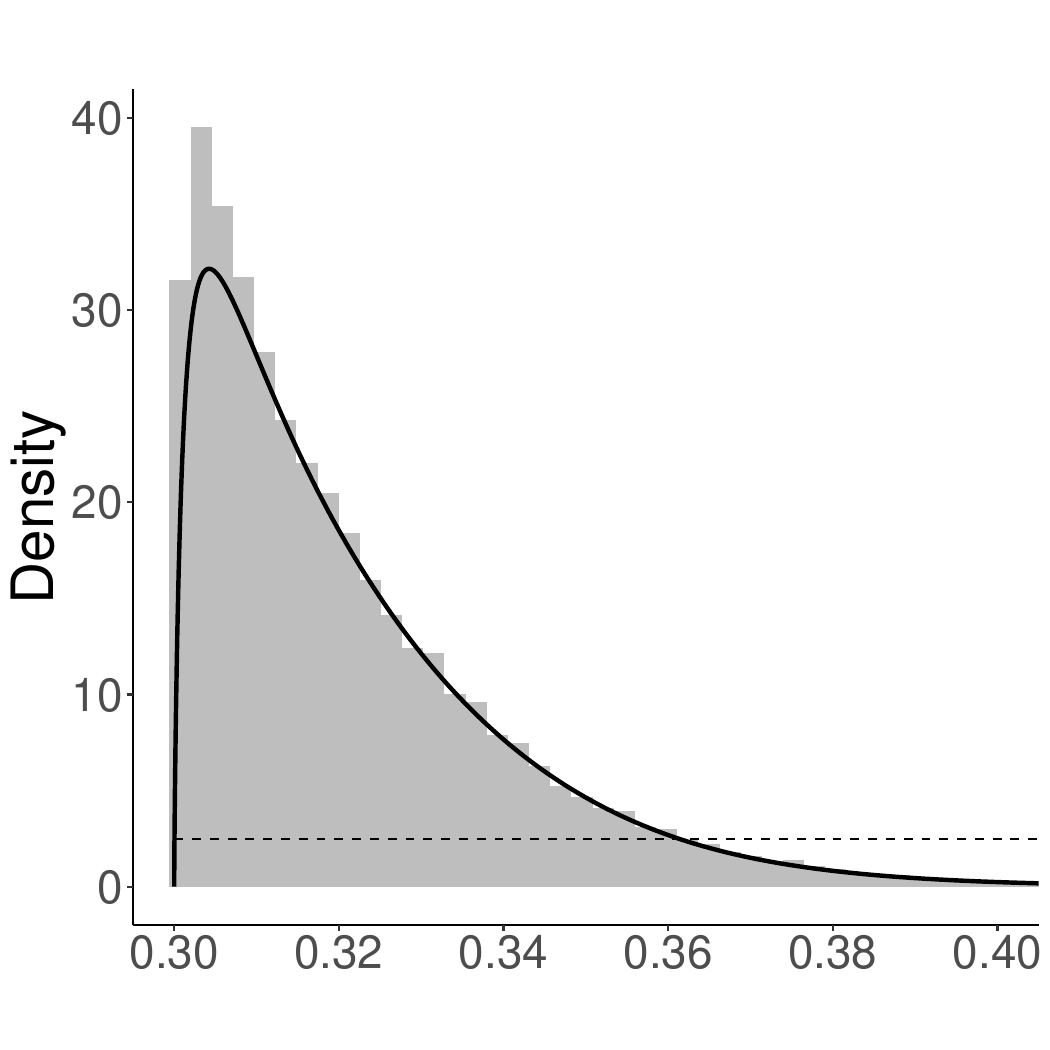}}\\
\subfloat[$\alpha$]{\includegraphics[width=0.3\textwidth]{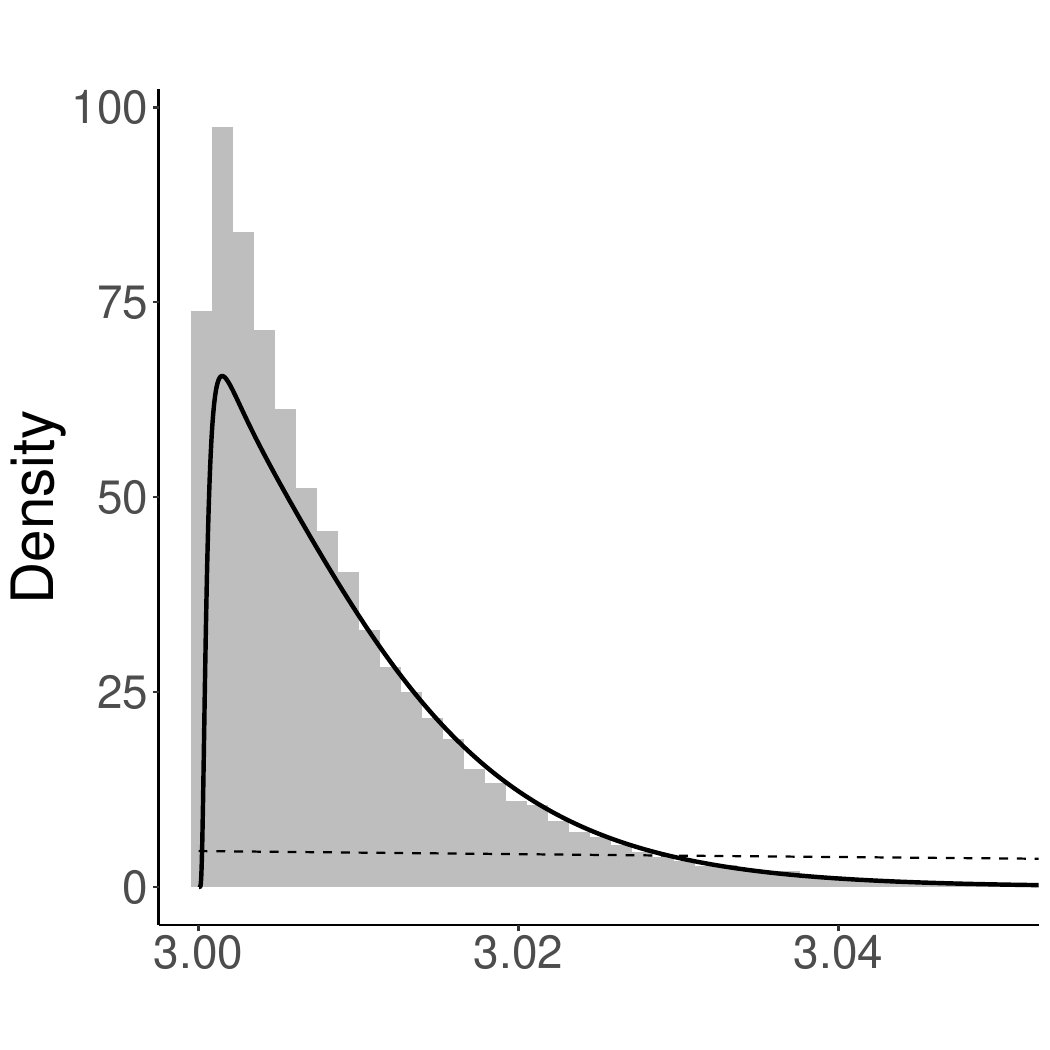}}
\subfloat[$\beta$]{\includegraphics[width=0.3\textwidth]{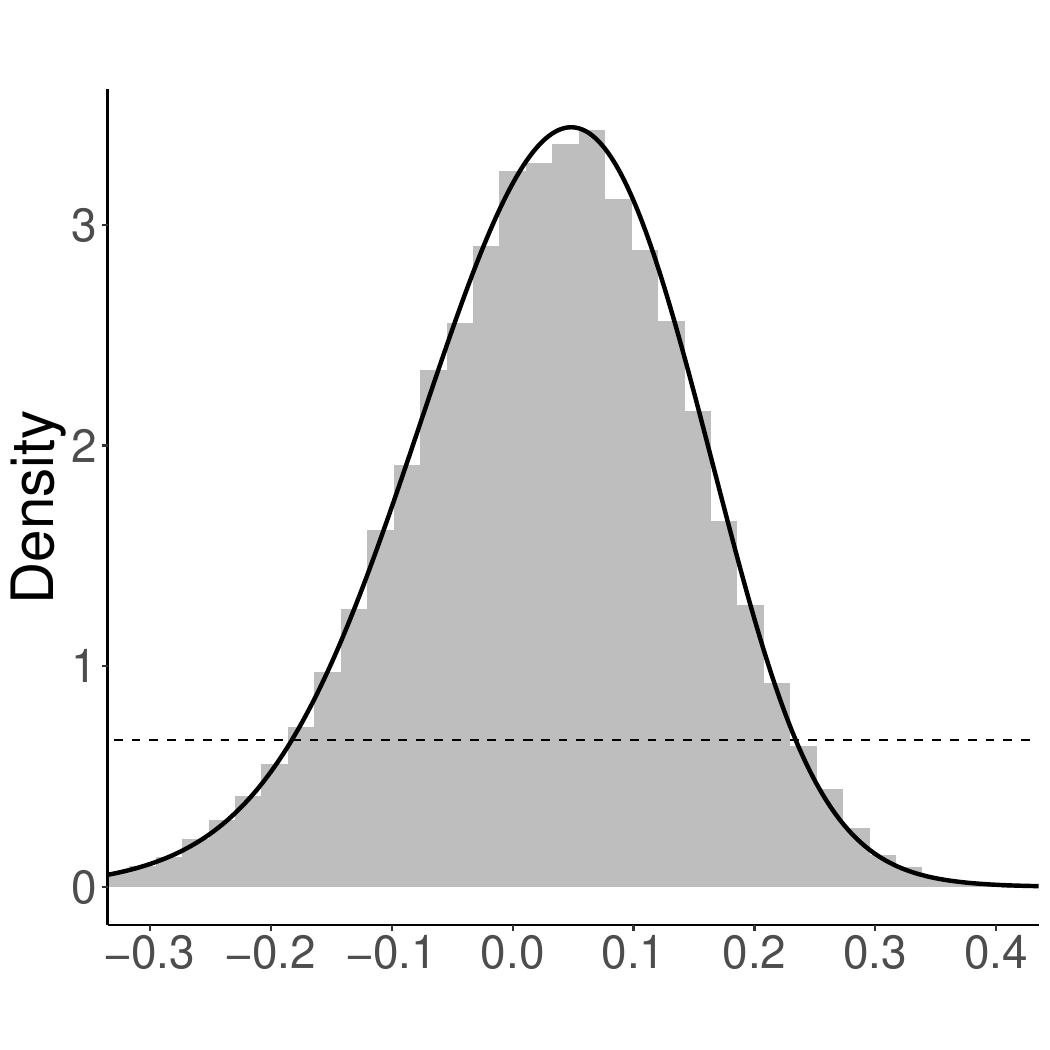}}
\caption{(a) -- (d): prior (- - -), and \AGHQ{} (---) and \MCMC{} ($\textcolor{lightgray}{\blacksquare}$) approximate posterior distributions for the four parameters from the astronomy data of \cref{subsec:astro}. The marginal posteriors for $\gamma$ (b) and $\alpha$ (c) are particularily skewed, and the approximation appears very accurate in the tails of both distributions, which is important for accurately quantifying uncertainty using marginal credible intervals.}
\label{fig:astroparam}
\end{figure}
\begin{figure}[h!]
\centering
\includegraphics[width=0.4\textwidth]{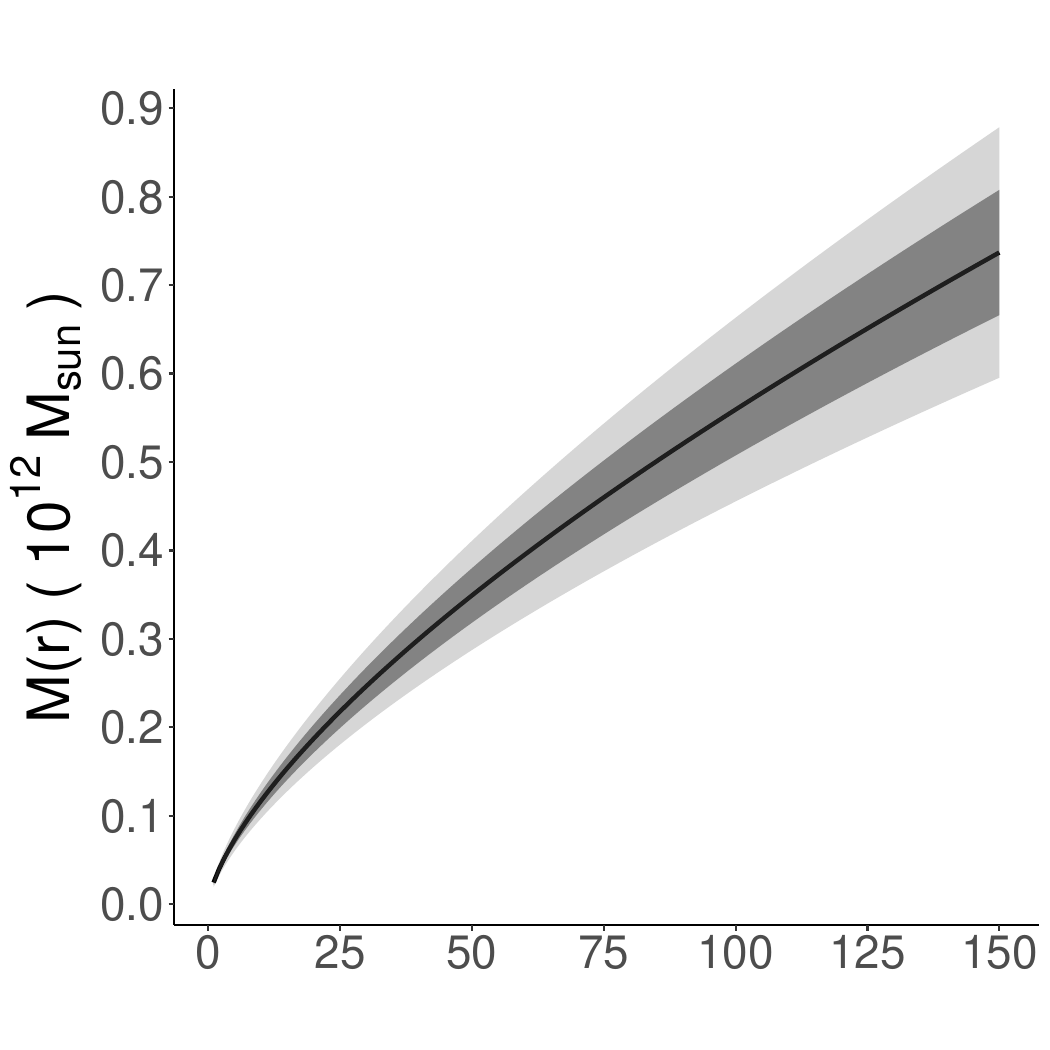}
\caption{\AGHQ{} estimated posterior mean mass (---), relative to the mass of the sun, of the Milky Way galaxy as a function of radial distance from galaxy centre (kpc), with one- (\textcolor{gray}{$\blacksquare$}) and two- (\textcolor{lightgray}{$\blacksquare$}) standard deviation bands for the astronomy data of \cref{subsec:astro}.}
\label{fig:astromass}
\end{figure}

\cref{fig:astroparam,fig:astromass} show the marginal posteriors of $\Xi$ and \cref{fig:astromass} the posterior mean and standard deviation of $M(r)$ respectively using $\quadnum = 5$, and hence \cref{thm:mainresult} prescribes an $\Ordp(n^{-2})$ relative error rate. The total computation time for the optimization, quadrature, and computation of marginal posteriors was around 1.3 seconds on a modern laptop using the \texttt{aghq} package. \cref{tab:astroks} shows the estimated KS statistic between the \AGHQ{} and \MCMC{} approximate empirical CDFs. 
\AGHQ{} is generally quite accurate, with slight disagreement in the middle of the posterior for $\alpha$, although the tail appears accurately estimated, which is reflected both visually in \cref{fig:astroparam} and numerically in \cref{tab:astroks}.

\begin{table}
\caption{\label{tab:astroks} Comparison of \AGHQ{} with $\quadnum=5$ to \MCMC{} using the KS distance for the astronomy data of \cref{subsec:astro}.}
\centering
\begin{tabular}{crrrr}
\hline
Param. & $\Psi_{0}$ & $\gamma$ & $\alpha$ & $\beta$ \\
\hline
KS(\AGHQ{},\MCMC{})
& 0.00872
& 0.00844
& 0.0358 
& 0.00739 \\
\hline
\end{tabular}
\end{table}
An interesting computational challenge emerges in this example: we observe that using a larger number of quadrature points to satisfy $\quadnum>5$ results in points outside the constraint regions and is hence infeasible. A similar challenge is observed when computing the marginal posterior for $\beta$, and in this case only we report results of a simpler method based on reuse of the original adapted points. These challenges may be due to the low sample size: $\logposthess^{-1}$ has a wide spectrum that causes the quadrature points to be spread far apart. As $n$ becomes larger this spectrum would be expected to become smaller and hence a larger number of quadrature points may be expected to lie inside the constraint region. However, we reiterate that $\quadnum = 5$ still yields a very fast $\Ordp(n^{-2})$ relative error rate by \cref{thm:mainresult} as well as empirically accurate results in this example (\cref{tab:astroks}).

%% file: sections/examples-high-dim.tex
\section{High-Dimensional Parameter Spaces}\label{sec:highdimexamples}

Adaptive quadrature is an increasingly popular technique in modern Bayesian statistics as one important component of more complicated methods for approximate posteriror inference in models with high-dimensional parameter spaces. In this section we demonstrate the use of one such type of method, based on the INLA method of \citet{inla}, through fitting a spatial model for zero-inflated counts, for which MCMC-based inference is observed to be challenging. 

The methods described in this section have no known convergence theory, and their usefulness in applied Bayesian statistics makes development of such theory a topic of substantial current interest. \cref{THM:MAINRESULT}, which describes the convergence properties of the adaptive quadrature rules used at the core of these methods, is a first step in this direction.

\subsection{High-Dimensional Approximation Method}\label{subsec:marginallaplace}

Consider a parameter vector $(\highdimparam,\param)$ where $\param\in\R^{\paramdim}$ and $\highdimparam\in\R^{\paramdimbig}$ with $\paramdim\ll \paramdimbig$. Bayesian inferences for these parameters are made using the posterior distributions:
\*[
\post{\param} &= \frac{\int\dist(\highdimparam,\param,\data)\dee\highdimparam}{\int\int\dist(\highdimparam,\param,\data)\dee\highdimparam\dee\param}, \\
\post{\highdimparam} &= \int\dist(\highdimparam \setdelim \param,\data)\dist(\param \setdelim \data)\dee\param.
\]
It is assumed that $\paramdim$ is small enough to make it computationally feasible to directly apply adaptive quadrature to $\dee\param$ integrals, but that $\paramdimbig$ is large enough for this to be infeasible for $\dee\highdimparam$ integrals, even using sparse grids or other non-product rule extensions to multiple dimensions. This occurs, for example, in hierarchical models (\citealt{kassandsteffy,inla,smoothestimation,geirsson20lgm}; \cref{SUBSEC:ZEROINF}) where $\highdimparam$ typically relate to the mean response, and $\param$ are variance components.

For any fixed $\param$, \citet{laplace} suggest approximating $\post{\param}\approx\LAGHapprox{\param}$ by first approximating $\int\dist(\highdimparam,\param,\data)\dee\highdimparam\approx\LAapprox(\param,\data)$ using \AGHQ{} with $k=1$ (a Laplace approximation), 
and then renormalizing the result using numerical integration, for which they also use \AGHQ{} in their experiments. \citet{casecrossover} combine this approximation with a Gaussian approximation $\dist(\highdimparam \setdelim \param,\data) \approx \GGapprox(\highdimparam \setdelim \param,\data)$, obtaining
\[\label{eqn:Wapprox}
\approxdist(\highdimparam \setdelim \data) \approx
\int
\GGapprox(\highdimparam \setdelim \param,\data)
\LAGHapprox{\param}
\dee\param.
\]
The integration in \cref{eqn:Wapprox} is approximated with the same \AGHQ{} points and weights used to obtain $\LAGHapprox{\param}$, so that $\approxdist(\highdimparam \setdelim \data)$ corresponds to a discrete mixture of Gaussian approximations with weights determined by \AGHQ{}. Inferences for $\highdimparam$ are then made by sampling from this Gaussian mixture. The INLA method of \citet{inla} uses an alternative adaptive quadrature rule for the renormalization, and then another Laplace approximation to approximate the marginal distributions $\dist(w_{j}|\param,\data)$.

There is a growing body of evidence suggesting that approximations based on \cref{eqn:Wapprox} give results empirically similar to those returned by \MCMC{} and other methods \citep{inla,geostatsp,inlamcmc,casecrossover,simplifiedinla} in faster computational times. In \cref{subsec:zeroinf} we show an example of a model for which a state-of-the-art MCMC algorithm runs for days and fails to converge (\cref{subsec:mcmc-results}) to a suitable solution, while \cref{eqn:Wapprox} provides a potentially suitable (\cref{subsec:empirical-accuracy-aghq}) solution in minutes. However, we stress that the convergence properties of \cref{eqn:Wapprox} are not known, and the apparent practical utility of this approximation makes establishing such properties an important area of research. Because \AGHQ{} is used several times in the computation of \cref{eqn:Wapprox}, \cref{THM:MAINRESULT} is a first step towards this broader goal.

\subsection{Example: Zero-Inflated Geostatistical Binomial Regression}\label{subsec:zeroinf}\label{SUBSEC:ZEROINF}

\citet{geostatlowresource} introduce a zero-inflated geostatistical binomial regression model, where both the incidence rate and suitability of infection (zero-inflation probability) varies spatially. They argue that such models are of substantial importance in the mapping of tropical diseases, and make frequentist inferences for the parameters of interest. Here we make Bayesian inferences for the spatial patterns in indidence and suitability of infection of a tropical disease in Nigeria and Cameroon, based on a dataset of subjects who tested positive in $n=190$ villages in this region \citep{loaloazero}. Data are obtained from the \texttt{loaloa} object in the \texttt{geostatsp} package \citep{geostatsp}. A simpler model that does not allow for zero-inflation has been fit using INLA \citep{geostatsp} as well as \MCMC{} and maximum likelihood \citep{prevmap}. 
To our knowledge, no previous Bayesian implementation of this zero-inflated model exists.

We apply \cref{eqn:Wapprox} to fit this model. Let $0 \leq y_{i} \leq N_{i},i\in[n]$ represent the counts of people infected and total number of people in the $i^{th}$ village out of the $n=190$ included in the data, and let $\mb{s}_{i}\in\R^{2}$ denote the geographical coordinates of this village. For every location $\mb{s}\in\R^{2}$, let $\phi(\mb{s})$ denote the probability that this location is capable of disease transmission (the \emph{suitability} probability), and $p(\mb{s})$ denote the probability that transmission occurs at this location, conditional on it being suitable (the \emph{incidence} probability). 
\citet{geostatlowresource} stress the practical importance of allowing observed zero counts $y_{i} = 0$ to either be \emph{haphazard} zeroes arising from sampling variability, or \emph{structural} zeroes arising from a location being unsuitable for disease transmission. They also discuss how this makes joint inference of the underlying spatial fields governing suitability and incidence very challenging. The full model is 
\*[
\mathbb{P}\left[Y_{i} = y_{i} | p(\mb{s}_{i}),\phi(\mb{s}_{i})\right] &= \left[1 - \phi(\mb{s}_{i})\right]\text{I}\left(y_{i} = 0\right) + \phi(\mb{s}_{i})\times\text{Binomial}[y_{i};N_{i},p(\mb{s}_{i})], \\
\log\left[\frac{\phi(\mb{s})}{1-\phi(\mb{s})}\right] &= \beta_{\texttt{suit}} + u(\mb{s}); \ \log\left[\frac{p(\mb{s})}{1-p(\mb{s})}\right] = \beta_{\texttt{inc}} + v(\mb{s}), \mb{s}\in\R^{2}\\
u(\cdot) | \param &\sim \GP(0,\text{C}_{\param}); \ v(\cdot) | \param \sim \GP(0,\text{C}_{\param}), \\
\]
where the unknown functions $u(\cdot),v(\cdot)$ are modelled as independent Gaussian Processes with the same Mat\'{e}rn covariance function, $\text{C}_{\param}$, with $\param = (\sigma,\rho)$  and the two intercepts are given independent Gaussian priors with variance $1000$. We assign
$\sigma$ and $\rho^{-2}$ independent exponential priors satisfying $\PP(\rho < 200\text{km}) = \PP(\sigma < 4) = 97.5\%$, following \citet{geostatsp} and \citet{pcpriormatern}. 

Inference for $u(\cdot)$ and $v(\cdot)$ is based on their values at the observed locations $\mb{s}_{i}$, and then posterior distributions for their values at any new location $\mb{s}\in\R^{2}$ are obtained using standard methods for spatial interpolation. 
Define $\mb{U} = \bracevec{u(\mb{s}_{i}):i\in[n]}$, $\mb{V} = \bracevec{v(\mb{s}_{i}):i\in[n]}$, and let $\highdimparam = (\mb{U},\beta_{\texttt{suit}},\mb{V},\beta_{\texttt{inc}})\in\R^{\paramdimbig},\paramdimbig=2n+2$. The Gaussian process priors on $u(\cdot)\negsetdelim\param$ and $v(\cdot)\negsetdelim\param$ imply that $\mb{U}\negsetdelim\param\sim\text{N}\left[0,\mb{\Sigma}(\param)\right]$ and $\mb{V}|\param\sim\text{N}\left[0,\mb{\Sigma}(\param)\right]$ independently, where $\left[\mb{\Sigma}(\param)\right]_{ij} = \text{C}_{\param}\left( \norm{\mb{s}_{i} - \mb{s}_{j}}\right), i,j\in[n]$. To infer $\mb{U}^{*} \equiv \bracevec{u(\mb{s}^{*}_{t}):t\in[T]}$ and $\mb{V}^{*} \equiv \bracevec{v(\mb{s}^{*}_{t}):t\in[T]}$ for any set of new locations $\bracevec{\mb{s}^{*}_{t}:t\in[T]}\subseteq\R^{2}$, we simulate from the predictive distribution $(\mb{U}^{*},\mb{V}^{*})|\mb{Y}$ by first drawing $\highdimparam$ from $\widetilde{\pi}(\highdimparam|\data)$ using standard methods \citep{fastsamplinggmrf}, and then sampling from $(\mb{U}^{*},\mb{V}^{*})\negsetdelim\highdimparam$ using existing algorithms for conditional simulation of Gaussian fields, implemented in the \texttt{geostatsp} \citep{geostatsp} and \texttt{RandomFields} \citep{randomfields} packages. 

We fit the model using \AGHQ{} with $\quadnum=7$ and the approximations described in \cref{subsec:marginallaplace}, and show the resulting spatial interpolations on a fine grid in \cref{fig:loaloaresults}. Total computation time for parameter estimation was 225 seconds. The predicted incidence probabilities appear visually similar to those reported by \citet{geostatsp} and \citet{prevmap} for the simpler model without zero-inflation, and the novel plot of predicted suitability probabilities identifies a cluster of villages that have a low posterior probability of being suitable for transmission. Owing to the lack of available convergence theory in this problem, we include a brief simulation study in \cref{subsec:empirical-accuracy-aghq} to assess the empirical accuracy of this procedure for this model and these data.

To better illustrate the difficulty of fitting this model with existing methods, we fit the model using \MCMC{} by running the ``NUTS'' sampler \citep{nuts} through the \texttt{tmbstan} package \citep{tmbstan} using the default settings. Eight chains of $10,000$ iterations each (including a $1,000$ iteration warmup) were run in parallel on a remote server at a total ``wall'' computation time of $66$ hours. The resulting chains exhibited divergent transitions according to \texttt{STAN}'s built in diagnostics. We investigated this in \cref{subsec:mcmc-results}, finding that $\beta_{\texttt{suit}}$ is poorly identified by the sampler. We ran both \cref{eqn:Wapprox} and \MCMC{} with $\beta_{\texttt{suit}}$ and $\beta_{\texttt{inc}}$ fixed at their estimated posterior means obtained from the initial fit of \cref{eqn:Wapprox}. This sampler converged without warnings in just over 19 hours for 10,000 iterations. The Kolmogorov-Smirnov (KS) statistics for the difference between approximate marginal CDFs from \MCMC{} and \cref{eqn:Wapprox} indicate that the two procedures provide mostly comparable inferences, with disagreement in a small number of villages. See \cref{subsec:mcmc-beta-fixed} for further details. We re-iterate that \MCMC{} did not  produce a complete answer for $\beta_{\texttt{suit}}$ in this problem.

Inferences made using \cref{eqn:Wapprox} produce a complete answer in around three and a half minutes on a modern server for this problem of substantial practical importance \citep{geostatlowresource}. In this same problem and on the same hardware, \MCMC{} either (a) runs for almost a day and produces an incomplete answer, or (b) runs for almost 3 days and fails. This example illustrates why these types of approximations have such high potential value in applied statistics, and why convergence theory for \cref{eqn:Wapprox} is of such importance. \cref{THM:MAINRESULT} provides a first step towards this goal.

\begin{figure}[t]
\caption{\label{fig:loaloaresults} \AGHQ{} estimated posterior mean (a) suitability probabilities and (b) incidence rates for the \texttt{loaloa} example of \cref{sec:highdimexamples}. 
}
\centering
\subfloat[{$\EE\left[\phi(\cdot) | \data \right]$}]{\includegraphics[width=0.45\textwidth]{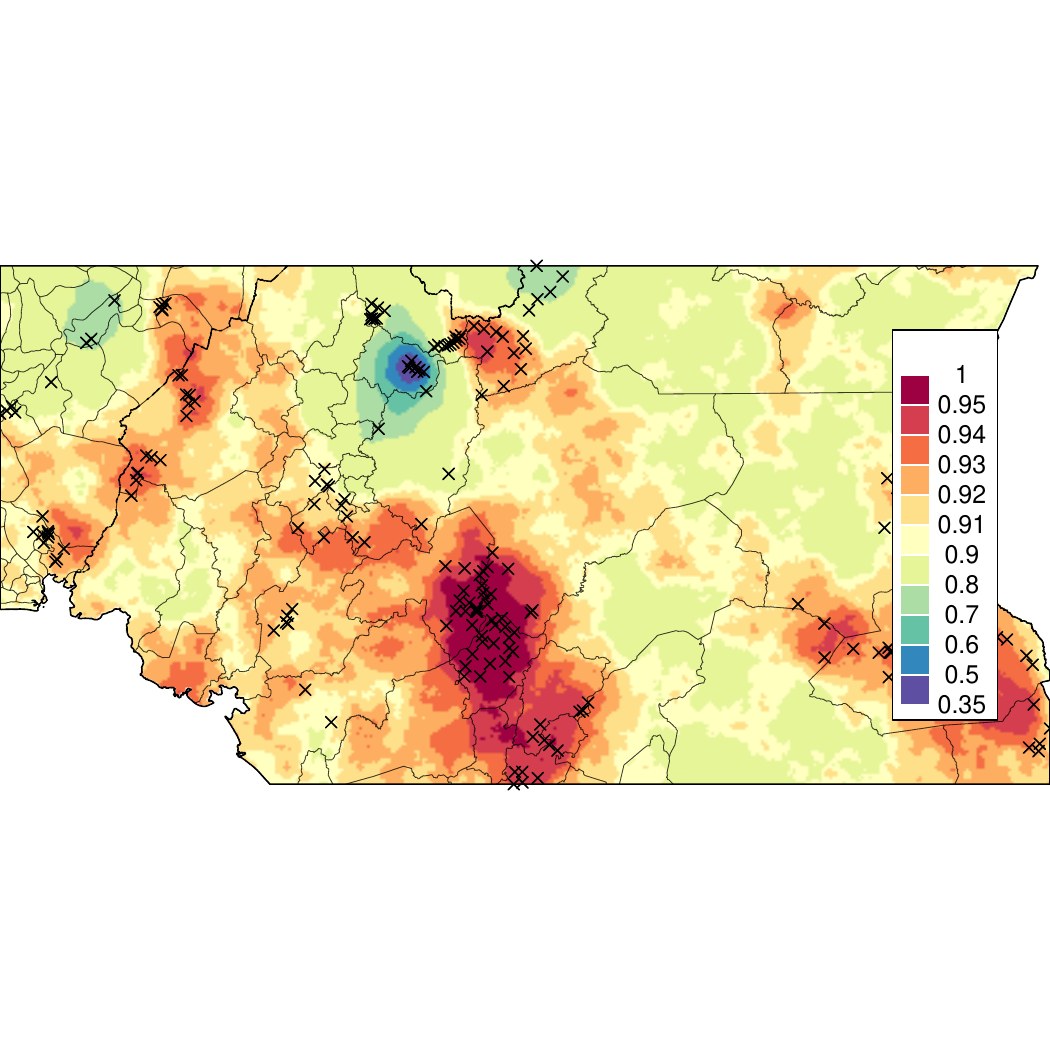}}
\subfloat[{$\EE\left[\phi(\cdot)\times p(\cdot) | \data\right]$}]{\includegraphics[width=0.45\textwidth]{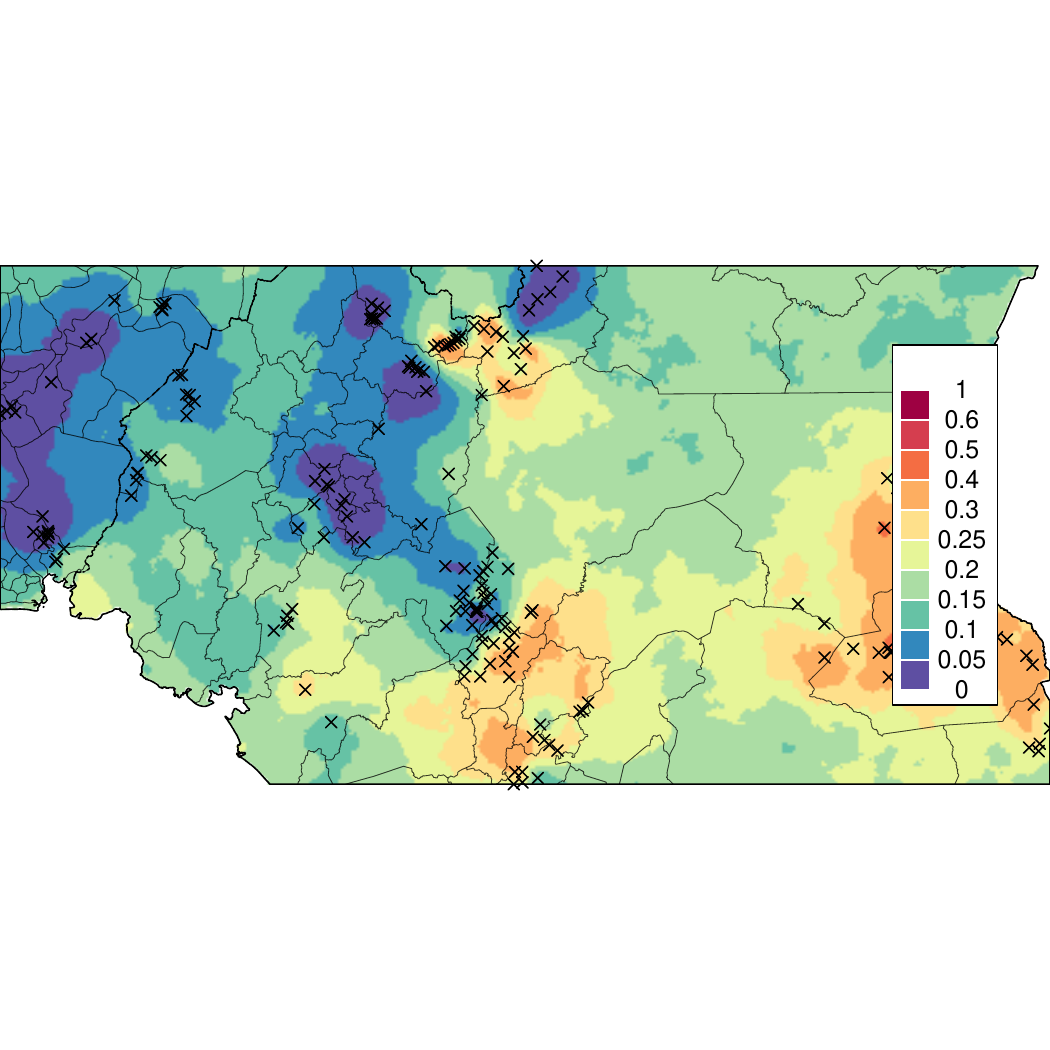}} \\
\end{figure}

%% file: sections/discussion.tex
\section{Discussion}\label{sec:discussion}

Using standard regularity assumptions, we have provided the first stochastic convergence rate for adaptive quadrature in Bayesian inference, and showed that this rate applies to the approximate normalizing constant, posterior density,
moments, and marginal densities. 
Using our \texttt{R} package \texttt{aghq}, available on \texttt{CRAN},
we demonstrated the use of \AGHQ{} for Bayesian inference in two challenging low-dimensional models and one high-dimensional model. 
We now briefly discuss five open problems for the theory of adaptive quadrature in Bayesian inference.

First, computing approximate quantiles and credible sets requires further integration of the approximate posterior over a subset of the parameter space, and hence a quadrature rule is needed that satisfies a truncated version of $\property{\paramdim}{\quadnum}$.
Providing a robust method for this computation with corresponding theoretical guarantees (analogous to \cref{fact:pos-marg-compute,fact:pos-mean-compute}) will complete the justification of using \AGHQ{} for all facets of Bayesian inference in low-dimensional models.
The current implementation uses an interpolation-based method with no theoretical guarantees, but appears to provide reasonable output in challenging examples.
Second, for high-dimensional models, the current implementation uses a Gaussian approximation and an adaptive quadrature approximation with reused points and weights. 
Providing full theoretical guarantees for the output of this entire procedure remains an open problem, and will not only validate the use of the the \texttt{aghq} package for such models but also provide the first theoretical guarantees for \texttt{INLA}-like methods; we believe that \cref{thm:mainresult} is an important first step towards this goal.
Third, our theoretical guarantees are all asymptotic and worst-case subject to the regularity assumptions. 
A challenging open problem is to provide theoretical guarantees that hold for finite samples and adapt to properties such as smoothness and sparsity, leading to improved performance for ``benign'' data and models.
Fourth, a principled choice of $\quadnum$ in any given practical application, for any given data set, remains an open problem. The recommendation from \cref{subsec:infectiousdisease} is feasible due to the fast run time of \AGHQ{}, and a more formally-motivated approach based on this could lead to a clearer and more useful practical recommendation.
Lastly, developing methods with comparable accuracy to AGHQ that are computationally feasible in very high dimensions remains a challenging open problem.

%% file: sections/acknowledgements.tex
\section*{Acknowledgements}\label{sec:acknowledgements}

BB acknowledges support from an NSERC Canada Graduate Scholarship and the Vector Institute.
AS acknowledges support from an NSERC Postgraduate Scholarship and the Centre for Global Health Research at St. Michael's Hospital, Toronto, Canada.
YT acknowledges support from an NSERC Postgraduate Scholarship and the Vector Institute.
We thank Jeffrey Negrea, Nancy Reid, Daniel Roy, and Jamie Stafford for helpful comments and suggestions.

%% file: sections/notation-and-assumptions.tex
\section{Regularity Assumptions}\label{sec:assumptions}

We state here some more notation and the required modelling assumptions for \cref{thm:mainresult}.
The log-likelihood of a parameter $\param\in\R^\paramdim$ is denoted by $\llhood(\param;\data) = \log \dist(\data \setdelim \param)$.
When the dependence on the data is clear, we may use $\llhood(\param)$ for brevity. Denote the 
log-posterior (unnormalized) by $\llandp(\param) = \log \poststar{\param} = \llhood(\param) + \log\dist(\param)$.
The maximum likelihood estimator is $\parammle = \argmax_{\param\in\paramspace}\llhood(\param)$, and the posterior mode is $\parammode = \argmax_{\param\in\paramspace}\llandp(\param)$.
The 
negative Hessian of the log-posterior is 
\*[
\logposthess(\param) = -\frac{\partial^{2}}{\partial \param\partial \param\tpose}\llandp(\param).
\] 
Further, we make frequent use of
the Cholesky decomposition of the inverse curvature of
the log-posterior at $\parammode$, which for symmetric, positive-definite $\logposthess(\parammode)$ is the unique lower-triangular matrix that satisfies
\*[
	\Big[\logposthess(\parammode)\Big]^{-1} = \logpostchol\logpostchol\tpose.
\]

For any $x_0 \in \Reals^\dummydim$ and $\universalradius>0$, let $\ball{\dummydim}{x_0}{\universalradius} = \left\{ x\in\R^{\dummydim} : \norm{x - x_0}_{2} < \universalradius \right\}$ denote the open ball in $\R^{\dummydim}$ of radius $\universalradius$ centred at $x_0$ with respect to the Euclidean norm. Let $\normaldens{x}{\mb{\mu}}{\mb{\Sigma}}$ denote the multivariate normal density evaluated at $\mb{x}$ with mean $\mb{\mu}$ and variance $\mb{\Sigma}$.
For a positive-definite $\dummydim\times\dummydim$ matrix $A$, let $\eigen_{1}(A) \geq \cdots \geq \eigen_{\dummydim}(A) > 0$ denote its ordered eigenvalues. For any $f: \Reals^\dummydim \rightarrow \Reals$, $\dumderivvec \subseteq \N^\dummydim$, and $\mb{x} \in \Reals^\dummydim$, we define
\*[
\abssmall{\dumderivvec} = \sum_{j = 1}^{\dummydim} \alpha_i, \quad \dumderivvec! = \prod_{j =1}^\dummydim \alpha_j!, \quad \mb{x}^{\dumderivvec} =\mb{x}_{\dumderivvec} =\prod_{j = 1}^\dummydim x_j^{\alpha_j}, \text{ and }\\
\partial^{\dumderivvec} f(\mb{x}) = \partial x_1^{\alpha_1} \partial x_2^{\alpha_2} \cdots \partial x_\dummydim^{\alpha_\dummydim} f( x) = \frac{\partial^{\abssmall{\dumderivvec}} f(x)} {\partial x_1^{\alpha_1} \partial x_2^{\alpha_2}\cdots \partial x_\dummydim^{\alpha_\dummydim}}. 
\]

For any data-generating distribution $\datatrueprobn$, we say the following assumptions hold if there exist $\universalradius>0$ and $\paramtrue\in\paramspace$ such that all five statements are true.

\begin{assumption}\label{assn:kderiv}
There exists $\derivnum,\derivbound>0$ such that for all $\dumderivvec \subseteq \N^\paramdim$ with $0 \leq \abssmall{\dumderivvec} \leq \derivnum $,
\*[
	\lim_{n \to \infty} \datatrueprobn\Big[\sup_{\param \in \ball{\paramdim}{\paramtrue}{\universalradius}}\absbig{\partial^{\dumderivvec}\llandp(\param)} < n \derivbound\Big]=1.
\]
\end{assumption}

\begin{assumption} \label{assn:hessian}
There exist $0 < \hesssmall \leq \hessbig < \infty$ such that 
\*[
	\lim_{n \to \infty} \datatrueprobn\Big[n \hesssmall \leq \inf_{\param \in \ball{\paramdim}{\paramtrue}{\universalradius}}\eigen_{\paramdim}(\logposthess(\param)) \leq \sup_{\param \in \ball{\paramdim}{\paramtrue}{\universalradius}} \eigen_{1}(\logposthess(\param))  \leq n \hessbig\Big] = 1.
\]
\end{assumption}

\begin{assumption}\label{assn:limsup}
There exists $\llhoodmargin>0$ such that 
\*[
	\lim_{n \to \infty} \datatrueprobn\Big[\sup_{\param \in \ballc{\paramdim}{\paramtrue}{\universalradius}} \llhood(\param) - \llhood(\paramtrue) \leq -n \llhoodmargin\Big] = 1.
\]
\end{assumption}

\begin{assumption}\label{assn:consistency}
For any $\conmargin > 0$ and function $G(n)$ such that $\lim_{n \rightarrow \infty} G(n) = \infty$,
\*[
  \lim_{n \to \infty} \datatrueprobn \left[ \frac{\sqrt{n}}{G(n)} \Norm{\parammode - \paramtrue}_2 > \conmargin \right] = 0.
\]
\end{assumption}

\begin{assumption}\label{assn:prior}
There exist  $0 < \priorsmall < \priorbig < \infty$ such that
\*[
   \priorsmall \leq \inf_{\param \in \ball{\paramdim}{\paramtrue}{\universalradius}} \dist(\param) \leq \sup_{\param \in \ball{\paramdim}{\paramtrue}{\universalradius}} \dist(\param) \leq \priorbig.
\]
\end{assumption}

\begin{remark}
\cref{assn:kderiv,assn:hessian,assn:prior} are standard assumptions that can be found in the asymptotic literature.
\cref{assn:limsup} corresponds to a consistency condition for the MLE (see the paragraph before Theorem~8 in \citealt{validitylaplace}). \cref{assn:consistency} is implied by $n^{1/2}$ the consistency of the MLE and \cref{assn:prior}.
In the presence of \cref{assn:prior}, \cref{assn:kderiv,assn:hessian,assn:limsup,assn:consistency} are equivalent to analogous assumptions on the log-likelihood and the MLE. 
\end{remark}

\begin{remark}
Our assumptions are similar to those found in Section 3 of \citet{validitylaplace}, with the exception that the number of derivatives we require can potentially be higher since $\quadnum \geq 1$, and our assumptions hold in probability rather than almost surely.
\end{remark}

\begin{remark}
Assumptions \ref{assn:kderiv}--\ref{assn:hessian} (where $m \geq 2$ in Assumption \ref{assn:kderiv}) and  \ref{assn:consistency}--\ref{assn:prior} are sufficient to imply the Bernstein-von Mises theorem holds for our model, meaning that the posterior distribution is asymptotically Gaussian.
Using Theorem 10.1 in \cite{vaart_1998}, the conditions on the model are: differentiability in quadratic mean, invertability of the Fisher information matrix at $\paramtrue$, continuity and positivity of the prior distribution at $\paramtrue$, and finally the existence of tests $\zeta_n$ such that for every $\epsilon > 0$:
\*[ 
\lim_{n \to \infty} \datatrueprobn\zeta_n = 0 \quad\text{ and } \quad \lim_{n \to \infty} \sup_{\norm{\param - \paramtrue}_2 \geq \eps}(1 - \zeta_n) = 0.
\] 
Assumptions \ref{assn:kderiv}, \ref{assn:hessian} and \ref{assn:prior} directly imply the first three conditions, as for the final requirement, let 
\*[
\zeta_n = \mathbb{I}\left\{ \norm{\parammle - \paramtrue}_2 > \frac{\log(n)}{n^{1/2}} \right\},
\]
then by Assumption \ref{assn:consistency} this sequence of test satisfies the final condition.
\end{remark}

%% file: sections/proof-B-main.tex
\section{Proof of \cref{THM:MAINRESULT}}\label{sec:main-thm-proof}

\subsection{Quantifying Accuracy for Approximate Bayesian Inference}\label{sec:ABI-accuracy}

We measure the accuracy of a normalizing--constant approximation by the \emph{relative error},
\*[
	\relerror(\data) = \absbig{\frac{\dist(\data)}{\approxdist(\data)} - 1}.
\]
Since we are ultimately interested in summary statistics of the posterior for Bayesian inference, we require further integration of the approximate posterior density. 
To measure this error, we use the \emph{total variation error}, 
\*[
	\TVerror(\data) = \sup_{\Borelset\in\Borel{\paramdim}}\absbig{\int_{\Borelset}\approxpost{\param} - \post{\param}\dee\param}.
\] 
Fortunately, by the definition of $\approxpost{\param}$, $\TVerror(\data)$ simplifies to
\*[
	\sup_{\Borelset\in\Borel{\paramdim}}\absbig{\int_{\Borelset}\approxpost{\param} - \post{\param}\dee\param} =&\absbig{\frac{1}{\approxdist(\data)} - \frac{1}{\dist(\data)}} \sup_{\Borelset\in\Borel{\paramdim}}\absbig{\int_{\Borelset}\poststar{\param}\dee\param} \\
	=& \absbig{\frac{1}{\approxdist(\data)} - \frac{1}{\dist(\data)}} \underbrace{\int_{\paramspace} \poststar{\param} \dee\param}_{\dist(\data)} \\
	=& \absbig{\frac{\dist(\data)}{\approxdist(\data)} - 1}, 
\]
so $\TVerror(\data) = \relerror(\data)$ and it suffices to analyse $\relerror(\data)$.

Finally, the choice of positioning the approximation in the numerator or denominator in the definition of $\relerror(\data)$ does not affect the discussion of asymptotic rates, as is made clear in the following Lemma.
\begin{lemma}\label{lem:relative-helper}
For any sequences of random variables $(A_n)$ and $(B_n)$ such that\\
\phantom{}\hspace{\parindent} {\upshape a)} $\PP_n(A_n > 0) = \PP_n(B_n > 0) = 1$ for all $n$, and\\ 
\phantom{}\hspace{\parindent} {\upshape b)} there exists $r>0$ and $\constt>0$ satisfying
\*[
	\lim_{n \to \infty} \PP_n\left(\absbig{\frac{A_n}{B_n}-1} \leq \constt n^{-r}\right) = 1,
\]
it holds that
\*[
	\lim_{n \to \infty} \PP_n\left(\absbig{\frac{B_n}{A_n}-1} \leq 2\constt n^{-r}\right) = 1.
\]
\end{lemma}
\begin{proof}[Proof of \cref{lem:relative-helper}]
By assumption $\PP_n(A_n =  0) = \PP_n(B_n =  0) = 0$ for all $n$, so in what follows we work on the event $\{ \abssmall{A_n} > 0  \} \cap \{ \abssmall{B_n} > 0  \}$. 
For any $z \in [0,1)$, $\frac{1}{1-z} \geq 1+z$. So, for all $n > (2\constt)^{1/r}$,
\*[
	\PP_n\left(\frac{B_n}{A_n} < 1 - 2\constt n^{-r}\right) 
	=& \PP_n\left(\frac{A_n}{B_n} \geq \Big[1 - 2\constt n^{-r}\Big]^{-1}\right) \\
	\leq& \PP_n\left(\frac{A_n}{B_n} \geq 1 + 2\constt n^{-r}\right) \\
	\leq& \PP_n\left(\absbig{\frac{A_n}{B_n} - 1} \geq \constt n^{-r}\right).
\]  
Similarly, for any $z \in [0,1]$, $\frac{1}{1+z} \leq 1-z/2$. So, for all $n>(2\constt)^{1/r}$,
\*[
	\PP_n\left(\frac{B_n}{A_n} > 1 + 2\constt n^{-r}\right) 
	=& \PP_n\left(\frac{A_n}{B_n} \leq \Big[1 + 2\constt n^{-r}\Big]^{-1}\right) \\
	\leq& \PP_n\left(\frac{A_n}{B_n} \leq 1 - \constt n^{-r}\right) \\
	\leq& \PP_n\left(\absbig{\frac{A_n}{B_n} - 1} \geq \constt n^{-r}\right).
\]  
Thus,
\*[
	&\hspace{-2em}\lim_{n \to \infty} \PP_n\left(\absbig{\frac{B_n}{A_n}-1} > 2\constt n^{-r}\right) \\
	&= \lim_{n \to \infty} \PP_n\left(\frac{B_n}{A_n} > 1 + 2\constt n^{-r} \ \bigcup \ \frac{B_n}{A_n} < 1 - 2\constt n^{-r} \right) \\
	&\leq \lim_{n \to \infty} \PP_n\left(\frac{B_n}{A_n} > 1 + 2\constt n^{-r}\right) 
		+ \lim_{n \to \infty} \PP_n\left( \frac{B_n}{A_n} < 1 - 2\constt n^{-r} \right) \\
	&\leq 2 \lim_{n \to \infty} \PP_n\left(\absbig{\frac{A_n}{B_n} - 1} \geq \constt n^{-r}\right) \\
	&= 0.
\]
\end{proof}

\MainResult*

The proof of \cref{thm:mainresult} follows directly from the combination of the following two lemmas. 

\begin{lemma} \label{lem:normalizing_constant_approx}
Under \cref{assn:kderiv,assn:hessian,assn:limsup,assn:prior,assn:consistency}, 
for all $1 \leq \quadnum \leq \lfloor \derivnum/2 \rfloor$,
if $\quadrule{\quadpointset}{\weight}$ is a quadrature rule satisfying $\property{\quadnum}{\paramdim}$ then
there exists a constant $\constt>0$ such that
\*[
	\lim_{n \to \infty}
	\datatrueprobn\bigg(\frac{\abssmall{\dist(\data) - \quadadaptmarg{\quadpointset}{\weight}}}{\dist(\data \setdelim \parammode)}
	\leq \constt \frac{1}{n^{\paramdim/2 + \lfloor (\quadnum +2)/3 \rfloor }}\bigg) = 1.
\]
\end{lemma}

\begin{lemma}\label{lem:lowerbound_quad}
Under \cref{assn:prior,assn:hessian,assn:consistency}, 
for all $1 \leq \quadnum \leq \lfloor \derivnum/2 \rfloor$,
if $\quadrule{\quadpointset}{\weight}$ is a quadrature rule satisfying $\property{\quadnum}{\paramdim}$ then
there exists a constant $\constt>0$ such that
\*[
	\lim_{n \to \infty} \datatrueprobn
	\left(
	\frac{\quadadaptmarg{\quadpointset}{\weight}}{\dist(\data \setdelim \parammode)}
	\geq \constt
	\frac{1}{n^{\paramdim/2}}
	\right)
	= 1.
\]
\end{lemma}

The rest of this section is devoted to proving \cref{lem:normalizing_constant_approx,lem:lowerbound_quad}.
For notational simplicity in the multivariate case, we adopt Einstein notation for tensor products throughout our proofs. 
In particular, when upper and lower indices appear twice in a term, this will denote summation over the relevant range of this index. For example, if $\mb{a},\mb{b} \in \Reals^\dummydim$, then
\*[
	\mb{a}_{i} \mb{b}^{i} = \sum_{i = 1}^\dummydim a_i b_i = \mb{a}\tpose \mb{b}.
\]
More specifically, we use this in the context of multivariate Taylor expansions. That is, for any $j \in [\dummydim]$,
\*[
	a_{i_1\dots i_j} \partial^{i_1\dots i_j} f(x)
	= \sum_{i_1,\dots,i_j \in [\dummydim]} \left(a_{i_1} \cdots a_{i_j} \right) \times \left( \partial_{x_{i_1}} \cdots \partial_{x_{i_j}} f(x) \right).
\]
Finally, to account for the constants in a Taylor series, we introduce the new notation:
\*[
	a_{[i_1\dots i_j]!} \partial^{i_1\dots i_j} f(x)
	= \sum_{i_1,\dots,i_j \in [\dummydim]} \frac{1}{i_1! \cdots i_j!} \left(a_{i_1} \cdots a_{i_j} \right) \times \left( \partial_{x_{i_1}} \cdots \partial_{x_{i_j}} f(x) \right).
\]

To ease notational burden when writing large polynomials, we also define for $\dummytau > 3$:
\*[
	\smalltsetdef[j]{\dummytau} &= \left\{(t_3,\dots,t_{2\quadnum}) \in \PosInts^{2\quadnum-3} \Bigsetdelim \sum_{s=3}^{2\quadnum} t_s = j \text{ and } \sum_{s=3}^{2\quadnum} s t_s \leq\dummytau - 1 \right\},\\
	\equaltsetdef[j]{\dummytau} &= \left\{(t_3,\dots,t_{2\quadnum}) \in \PosInts^{2\quadnum -3} \Bigsetdelim \sum_{s=3}^{2\quadnum} t_s = j \text{ and } \sum_{s=3}^{2\quadnum} s t_s = \dummytau \right\},\\
	\bigtsetdef[j]{\dummytau} &= \left\{(t_3,\dots,t_{2\quadnum}) \in \PosInts^{2\quadnum -3} \Bigsetdelim \sum_{s=3}^{2\quadnum} t_s = j \text{ and } \sum_{s=3}^{2\quadnum} s t_s \geq \dummytau \right\}.
\]
The significance of only considering $s\geq 3$ is made clear in \cref{sec:proof_const_as}, but arises from considering only the higher-order terms of a Taylor series expansion. We also use $\tsum = \sum_{s=3}^{2\quadnum} s t_s$ for any $\tvec = (t_3,\dots,t_{2\quadnum}) \in \PosInts^{2\quadnum -3}$.

%% file: sections/proof-C-submain.tex
\subsection{Proof of \cref{lem:normalizing_constant_approx}}\label{sec:proof-lemmas}

Fix arbitrary $\shrinkingconst >0$ (to be tuned at the end as a function of $\paramdim$ and $\quadnum$) and let $\shrinkingrad = \shrinkingconst \sqrt{(\log n)/n}$ for each $n \in \Nats$. First, expand the fraction of interest, giving
\*[
	&\frac{\absbig{\dist(\data) - \quadadaptmarg{\quadpointset}{\weight}}}{\dist(\data \setdelim \parammode)} \\
	&= 
	\frac{
	\absbig{
		\int_{\paramspace}\dist(\param) \dist(\data \setdelim \param)\dee\param  
		- 
		\abssmall{\logpostchol} \sum\limits_{\quadpointvec\in\quadpointset} \weight(\quadpointvec) \dist\Big(\logpostchol \, \quadpointvec + \parammode\Big) \dist\Big(\data \Bigsetdelim \param = \logpostchol \, \quadpointvec + \parammode\Big)}
	}
	{\dist(\data \setdelim \parammode)} \\
	&= \dist(\parammode) 
	\bigg\lvert
		\int_{\paramspace}\exp\left\{\llandp(\param) - \llandp(\parammode) \right\}\dee\param \\
		&\qquad\qquad\qquad - 
		\abssmall{\logpostchol} \sum\limits_{\quadpointvec\in\quadpointset} \weight(\quadpointvec) \exp\left\{\llandp(\logpostchol \, \quadpointvec + \parammode) - \llandp(\parammode) \right\}
	\bigg\rvert, \]
	which after splitting the region of integration and applying the triangle inequality is:
\[\label{eqn:initial-expansion}
	&\leq \dist(\parammode) 
	\bigg\lvert
		\int_{\ball{\paramdim}{\parammode}{\shrinkingrad}}\exp\left\{\llandp(\param) - \llandp(\parammode) \right\}\dee\param \\
		&\qquad\qquad\qquad - 
		\abssmall{\logpostchol} \sum\limits_{\quadpointvec\in\quadpointset} \weight(\quadpointvec) \exp\left\{\llandp(\logpostchol \, \quadpointvec + \parammode) - \llandp(\parammode) \right\}
	\bigg\rvert \\
	&\qquad + \dist(\parammode) \int_{\ballc{\paramdim}{\parammode}{\shrinkingrad}}\exp\left\{\llandp(\param) - \llandp(\parammode) \right\}\dee\param.
\]

Our strategy is to upper bound \cref{eqn:initial-expansion} using a few key quantities, and then show that the regularity assumptions imply these quantities are of the correct order with probability tending to 1. Specifically, we use
\*[
	\derivboundmode = \sup_{\dumderivvec: \abssmall{\dumderivvec}\leq \derivnum} \sup_{\param\in\ball{\paramdim}{\parammode}{\quadpointbound\lvert\logpostchol\rvert \lor \shrinkingrad}} \absbig{\partial^{\dumderivvec} \llandp(\param)} ,
	\quad 
	\hessbigmode = \frac{\eigen_{1}(\logposthess(\parammode))}{n} ,
	\quad \text{and} \quad
	\hesssmallmode = \frac{\eigen_{\paramdim}(\logposthess(\parammode))}{n},
\]
where $\quadpointbound = \sup_{\quadpointvec \in \quadpointset} \norm{\quadpointvec}_2 < \infty$.

We use these quantities to state the following result, which handles the primary technical difficulties for proving \cref{thm:mainresult}. We defer its proof
to \cref{sec:proof_const_as}.

\begin{lemma} \label{lem:normalizing_constant_as}\label{LEM:NORMALIZING_CONSTANT_AS}
For all $1 \leq \quadnum \leq \lfloor \derivnum/2 \rfloor$,
if $\quadrule{\quadpointset}{\weight}$ is a quadrature rule satisfying $\property{\quadnum}{\paramdim}$ then
there exists a constant $\paramconst>0$ depending only on $\paramdim$ and $\quadnum$ such that for all $n \in \Nats$ it holds $\datatrueprobn$-a.s. that
\[\label{eqn:almost-sure-lemma}
	&\absbig{
		\int_{\ball{\paramdim}{\parammode}{\shrinkingrad}}\exp\left\{\llandp(\param) - \llandp(\parammode) \right\}\dee\param
		- 
		\abssmall{\logpostchol} \sum\limits_{\quadpointvec\in\quadpointset} \weight(\quadpointvec) \exp\left\{\llandp(\logpostchol \, \quadpointvec + \parammode) - \llandp(\parammode) \right\}
	} \\
	&\leq
	\paramconst
	\left(
	\left(\frac{\hessbigmode}{\hesssmallmode}\right)^{\paramdim/2}
	+
	1
	\right) n^{-\paramdim/2} \\
	&\qquad\times
	\Bigg[
	\max_{j \in [\quadnumadj]}
	\max_{\tvec \in \bigtset{j}}
	(\derivboundmode)^{j}  \,
	(\hesssmallmode \, n)^{-\tsum/2}
	+
	\max_{\tvec \in \customtset{\quadnumadj+1}{3(\quadnumadj +1) + \ind\{2\quadnum = 2 \!\!\!\! \pmod 3\}} } (\derivboundmode)^{\quadnumadj + 1} 
	(\hesssmallmode \, n)^{-\tsum/2}
	\\
	&\qquad\qquad +
	(\hesssmallmode)^{-1/2}
	\max_{j \in \{0\} \cup [\quadnumadj]} 
	(\derivboundmode)^{j} 
	n^{-\frac{\shrinkingconst^2 \hesssmallmode + 2}{4}}
	+
	\multitaylorexpbound
	(\derivboundmode)^{\quadnumadj+2} 
	\max_{\tvec \in \customtset{\quadnumadj+2}{3(\quadnumadj+2)}}
	(\hesssmallmode \, n)^{-\tsum/2}
	\Bigg],
\]
where
\*[
	\multitaylorexpbound
	&= \Bigg[
	\left(\frac{\hessbigmode}{\hesssmallmode}\right)^{\paramdim/2} 
	\exp\left\{(2\quadnum) \, \paramdim^{2\quadnum} \max\{1, \shrinkingconst^{2\quadnum}\} \derivboundmode \Big(\frac{\log(n)}{n} \Big)^{3/2} \right\} \\
	&\qquad\qquad+
	\exp\left\{(2\quadnum) \, \paramdim^{2\quadnum} \max\{1, \quadpointbound^{2\quadnum}\} \derivboundmode \max\left\{(\hesssmallmode \, n)^{-3/2}, (\hesssmallmode \, n)^{-\quadnum}\right\}\right\}
	\Bigg]
\]
and $\quadnumadj$ is the smallest integer such that $3(\quadnumadj+1)\geq 2\quadnum$.
\end{lemma}

We then want to make use of the following, which provides the necessary convergence for each of the quantities used in \cref{lem:normalizing_constant_as}.
\begin{lemma}\label{lem:constant_convergence}
Under \cref{assn:kderiv,assn:consistency,assn:hessian,assn:prior}, the following hold:
\begin{enumerate}
\item[\upshape i)]
$\underset{n\to\infty}{\text{\upshape lim}} \datatrueprobn\!\bigg(\hesssmall \leq \hesssmallmode \leq \hessbigmode \leq \hessbig \bigg) = 1.$

\item[\upshape ii)] 
$\underset{n\to\infty}{\text{\upshape lim}} \datatrueprobn\!\bigg(\derivboundmode \leq n \derivbound \bigg) = 1.$

\item[\upshape iii)]
$\underset{n\to\infty}{\text{\upshape lim}} \datatrueprobn\!\bigg(\dist(\parammode) \leq \priorbig \bigg) = 1.$
\end{enumerate}
\end{lemma}
\begin{proof}[Proof of \cref{lem:constant_convergence}]$ $\\
i) 
By \cref{assn:consistency},
$\lim_{n \to \infty} \datatrueprobn \left( \parammode \in \ball{\paramdim}{\paramtrue}{\universalradius} \right) = 1$.
Thus,
\*[ 
	\lim_{n \to \infty} \datatrueprobn\left(\hessbigmode \leq \hessbig \right)
	&= \lim_{n \to \infty} \datatrueprobn\left(\frac{\eigen_1(\logposthess(\parammode))}{n} \leq \hessbig \right) \\
	&\geq \lim_{n \to \infty} \datatrueprobn\left(\parammode \in \ball{\paramdim}{\paramtrue}{\universalradius}, \sup_{\param \in \ball{\paramdim}{\paramtrue}{\universalradius}} \frac{\eigen_{1}(\logposthess(\param))}{n} \leq \hessbig \right) \\
	&= 1,
\]
where the last step uses \cref{assn:hessian}. The inequality for $\hesssmallmode$ is proven in the same fashion. 

\noindent
ii) 
First, recall that $\abssmall{\logpostchol} \leq (\hesssmallmode \, n)^{-\paramdim/2}$. By i), we have
\*[
	\lim_{n \to \infty} \datatrueprobn \left((\hesssmallmode \, n)^{-\paramdim/2} \leq (\hesssmall \, n)^{-\paramdim/2} \right) = 1.
\]
That is, for all $\eps > 0$, 
\*[
	\lim_{n \to \infty} \datatrueprobn \left(\quadpointbound\abssmall{\logpostchol}t \lor \shrinkingrad < \eps \right) = 1,
\]
so by \cref{assn:consistency} again, 
\*[
	\lim_{n \to \infty} \datatrueprobn \left( \ball{\paramdim}{\parammode}{\quadpointbound\lvert\logpostchol\rvert \lor \shrinkingrad} \subseteq \ball{\paramdim}{\paramtrue}{\universalradius} \right) = 1.
\]

Thus,
\*[ 
	&\hspace{-2em}\lim_{n \to \infty} \datatrueprobn\left(\derivboundmode \leq n \derivbound \right) \\
	&= \lim_{n \to \infty} \datatrueprobn\left(\sup_{\dumderivvec: \abssmall{\dumderivvec}\leq \derivnum} \sup_{\param\in \ball{\paramdim}{\parammode}{\quadpointbound\lvert\logpostchol\rvert \lor \shrinkingrad}} \absbig{\partial^{\dumderivvec} \llandp(\param)} \leq n \derivbound \right)  \\
	&\geq \lim_{n \to \infty} \datatrueprobn\left(\ball{\paramdim}{\parammode}{\quadpointbound\lvert\logpostchol\rvert \lor \shrinkingrad} \subseteq \ball{\paramdim}{\paramtrue}{\universalradius}, \sup_{\dumderivvec: \abssmall{\dumderivvec}\leq \derivnum} \sup_{\param\in \ball{\paramdim}{\paramtrue}{\universalradius}} \absbig{\partial^\alpha \llandp(\param)} \leq n \derivbound \right) \\
	&= 1,
\]
where the last step uses \cref{assn:kderiv}. 

\noindent
iii) This follows directly from \cref{assn:consistency,assn:prior}.

\end{proof}

Now, if $\hesssmall \leq \hesssmallmode \leq \hessbigmode \leq \hessbig$ and $\derivboundmode \leq n \derivbound$, for large enough $n$ it holds that $\multitaylorexpbound \leq 2$, so we can further upper bound the RHS of \cref{eqn:almost-sure-lemma} to obtain
\*[
	&\hspace{-2em}\absbig{
		\int_{\ball{\paramdim}{\parammode}{\shrinkingrad}}\exp\left\{\llandp(\param) - \llandp(\parammode) \right\}\dee\param
		- 
		\abssmall{\logpostchol} \sum\limits_{\quadpointvec\in\quadpointset} \weight(\quadpointvec) \exp\left\{\llandp(\logpostchol \, \quadpointvec + \parammode) - \llandp(\parammode) \right\}
	} \\
	&\leq
	\constt n^{-\paramdim/2}
	\Bigg[
	\max_{j \in [\quadnumadj]}
	\max_{\tvec \in \bigtset{j}}
	n^{j-\tsum/2} 
	+
	\max_{\tvec \in \customtset{\quadnumadj+1}{3(\quadnumadj +1) + \ind\{2\quadnum = 2 \!\!\!\! \pmod 3\}} } n^{\quadnumadj+1-\tsum/2} \\
	&\qquad\qquad\qquad+
	\max_{j \in \{0\} \cup [\quadnumadj]} 
	n^{j-\frac{\shrinkingconst^2 \hesssmall + 2}{4}} 
	+
	\max_{\tvec \in \customtset{\quadnumadj+2}{3(\quadnumadj+2)}}
	n^{\quadnumadj+2-\tsum/2}
	\Bigg]
\]
for some constant $\constt>0$ that depends on $\paramdim$ and $\quadnum$ as well as $\hesssmall$, $\hessbig$, and $\derivbound$. 

Clearly, we can take $\shrinkingconst$ arbitrarily large to make the third term as small (polynomially) as we desire. Upon inspection, the first term is maximized at $j = \quadnumadj$ and $\tsum = 2\quadnum$, the second term is maximized at $\tsum = 3(\quadnumadj +1) + \ind\{2\quadnum = 2\pmod 3\}$, and the fourth term is maximized at $\tsum = 3(\quadnumadj+2)$. Thus,
\[\label{eqn:conditional-rate}
	&\hspace{-2em}\absbig{
		\int_{\ball{\paramdim}{\parammode}{\shrinkingrad}}\exp\left\{\llandp(\param) - \llandp(\parammode) \right\}\dee\param
		- 
		\abssmall{\logpostchol} \sum\limits_{\quadpointvec\in\quadpointset} \weight(\quadpointvec) \exp\left\{\llandp(\logpostchol \, \quadpointvec + \parammode) - \llandp(\parammode) \right\}
	} \\
	&\leq
	\constt n^{-\paramdim/2}
	\Bigg[
	\frac{1}{n^{\quadnum - \quadnumadj}} 
	+
	\frac{\ind\{2\quadnum \neq 2\pmod 3\}}{n^{(\quadnumadj+1)/2}} 
	+
	\frac{\ind\{2\quadnum = 2\pmod 3\}}{n^{(\quadnumadj+1)/2 + 1/2}} 
	+ 
	\frac{1}{n^{(\quadnumadj+2)/2}} 
	\Bigg]. 
\]

Now, recall that $\quadnumadj$ is chosen to be the smallest integer such that $3(\quadnumadj+1) \geq 2\quadnum$. There are three cases to consider. 
If $2\quadnum = 0 \pmod 3$, then $\quadnumadj = 2\quadnum/3 - 1$, if $2\quadnum = 1 \pmod 3$, then $\quadnumadj = (2\quadnum-1)/3$, and if $2\quadnum = 2 \pmod 3$, then $\quadnumadj = (2\quadnum-2)/3$. Substituting this into \cref{eqn:conditional-rate} gives
\*[
	&\hspace{-2em}\absbig{
		\int_{\ball{\paramdim}{\parammode}{\shrinkingrad}}\exp\left\{\llandp(\param) - \llandp(\parammode) \right\}\dee\param
		- 
		\abssmall{\logpostchol} \sum\limits_{\quadpointvec\in\quadpointset} \weight(\quadpointvec) \exp\left\{\llandp(\logpostchol \, \quadpointvec + \parammode) - \llandp(\parammode) \right\}
	} \\
	&\leq
	\constt n^{-\paramdim/2}
	\Bigg[
	\frac{\ind\{2\quadnum = 0\pmod 3\}}{n^{\quadnum/3}} 
	+
	\frac{\ind\{2\quadnum = 1\pmod 3\}}{n^{(\quadnum+1)/3}} 
	+
	\frac{\ind\{2\quadnum = 2\pmod 3\}}{n^{(\quadnumadj+2)/3}} 
	\Bigg] \\
	&= \constt n^{-\paramdim/2 - \lfloor \frac{\quadnum+2}{3} \rfloor}.
\]

So, in conjuction with \cref{lem:constant_convergence}, this controls the first term on the RHS of \cref{eqn:initial-expansion}. For the second term of the RHS, \cref{lem:fixed_numer,lem:shrinking_numer} (stated below) together complete the proof of \cref{lem:normalizing_constant_approx} by taking $\shrinkingconst$ large enough (in $\quadnum$).

\manualendproof 

\begin{lemma}\label{lem:fixed_numer}
Under \cref{assn:limsup,assn:prior},
\*[
	\lim_{n \to \infty}
	\datatrueprobn
	\left[
	\int_{\ballc{\paramdim}{\paramtrue}{\universalradius}} \exp\{ { \llandp(\param) - \llandp(\parammode)} \} \dee\param
	\leq
	\frac{e^{-n\llhoodmargin}}{\priorsmall}  
	\right] = 1, 	
\]
\end{lemma}
\begin{proof}[Proof of \cref{lem:fixed_numer}]
Since $\parammode$ maximizes $\llandp$,
\*[
	\int_{\ballc{\paramdim}{\paramtrue}{\universalradius}}
	\exp\{ { \llandp(\param) - \llandp(\parammode)} \}
	\dee \param
	&\leq
	\int_{\ballc{\paramdim}{\paramtrue}{\universalradius}}
	\frac{\dist(\param)}{\dist(\paramtrue)} \exp\{  \llhood(\param) - \llhood(\paramtrue) \}
	\dee \param \\
	&\leq \frac{1}{\dist(\paramtrue)} \sup_{\paramdum \in \ballc{\paramdim}{\paramtrue}{\universalradius}}\exp\{  \llhood(\paramdum) - \llhood(\paramtrue) \} 
	\int_{\paramspace} \dist(\param) \dee\param \\
	&= \frac{1}{\dist(\paramtrue)} \sup_{\paramdum \in \ballc{\paramdim}{\paramtrue}{\universalradius}}\exp\{  \llhood(\paramdum) - \llhood(\paramtrue) \} .
\]
The result then follows by applying \cref{assn:limsup,assn:prior}.
\end{proof}

\begin{lemma}\label{lem:shrinking_numer}
Under \cref{assn:hessian,assn:limsup,assn:prior,assn:consistency}, there exists a constant $\constt>0$ such that
\*[
  \lim_{n \to \infty} \datatrueprobn\left[\int_{\ballc{\paramdim}{\parammode}{\shrinkingrad} \cap \ball{\paramdim}{\paramtrue}{\universalradius} } \exp\{ { \llandp(\param) - \llandp(\parammode)} \} \dee\param
  \leq \constt \frac{1}{n^{ \shrinkingconst^2\hesssmall/4 + \paramdim/2}} \right] = 1.
\]
\end{lemma}
\begin{proof}[Proof of \cref{lem:shrinking_numer}]

Using a second order Taylor expansion for each $\param\in\ballc{\paramdim}{\parammode}{\shrinkingrad} \cap \ball{\paramdim}{\paramtrue}{\universalradius}$ gives
\*[
	\exp\{ { \llandp(\param) - \llandp(\parammode)} \} 
	= \exp\left\{-\frac{1}{2}(\param - \parammode)\Tr \logposthess(\parammid^\param) (\param - \parammode)\right\},
\]
where $\parammid^\param = \modemidweight_\param \parammode + (1-\modemidweight_\param)\param$ for some $\modemidweight_\param \in [0,1]$.
Consider the case where $\parammode \in \ball{\paramdim}{\paramtrue}{\shrinkingrad}$. This implies that, since $\shrinkingrad \to 0$, for large enough $n$
\*[
	&\hspace{-2em}\sup_{\param\in\ball{\paramdim}{\paramtrue}{\universalradius}}
	\Norm{\parammid^\param - \paramtrue}_2 \\
	&= \sup_{\param\in\ball{\paramdim}{\paramtrue}{\universalradius}}
	\Norm{\modemidweight_\param \parammode + (1-\modemidweight_\param)\param - \modemidweight_\param\paramtrue - (1-\modemidweight_\param)\paramtrue}_2 \\
	&\leq \sup_{\param\in\ball{\paramdim}{\paramtrue}{\universalradius}}
	\modemidweight_\param \shrinkingrad + (1-\modemidweight_\param)\universalradius \\
	&\leq \universalradius.
\]
Thus, letting $\hesssmalltrue = \inf_{\param \in \ball{\paramdim}{\paramtrue}{\universalradius}}\eigen_{\paramdim}(\logposthess(\param))/n$, we have
\*[
	\sup_{\param\in\ball{\paramdim}{\paramtrue}{\universalradius}} (\param - \parammode)\Tr \logposthess(\parammid^\param) (\param - \parammode) 
	\geq \hesssmalltrue \, n (\param - \parammode)\Tr (\param - \parammode).
\]

That is,
\*[
	&\hspace{-2em}\int_{\ballc{\paramdim}{\parammode }{\shrinkingrad} \cap \ball{\paramdim}{\paramtrue}{\universalradius}}  \exp\{ {\llandp(\param) - \llandp(\parammode)} \}\dee\param \\
	&\leq  
	\int_{\ballc{\paramdim}{\parammode}{\shrinkingrad} \cap \ball{\paramdim}{\paramtrue}{\universalradius}}  \exp\left\{\frac{-\hesssmalltrue \, n}{2} (\param - \parammode)\Tr (\param - \parammode) \right\} \dee\param \\
	&\leq 
	\int_{\ballc{\paramdim}{\parammode}{\shrinkingrad}}  \exp\left\{\frac{-\hesssmalltrue \, n}{2} (\param - \parammode)\Tr(\param - \parammode) \right\} \dee\param \\
	&=  \left( \frac{2\pi}{\hesssmalltrue \, n } \right)^{\paramdim/2} \PP \left[ \chi_\paramdim^2 \geq  \hesssmalltrue \, n \shrinkingrad^2  \right]  \\
	&=  \left( \frac{2\pi}{\hesssmalltrue \, n} \right)^{\paramdim/2} \PP \left[ \chi_\paramdim^2/\paramdim \geq 1 + \zeta_n  \right],
\]

 where $\zeta_n = \shrinkingconst^2 \log(n) \hesssmalltrue/ \paramdim - 1 $. Then, by Lemma 3 in \cite{Fan},
\*[
	\PP \left[ \chi_\paramdim^2/\paramdim \geq 1 + \zeta_n \right] 
	&\leq \exp\left\{\frac{\paramdim}{2} [\log(1 +\zeta_n )-\zeta_n ] \right\} \\
	&= \exp\left\{\frac{\paramdim}{2}\left[\log\left(\frac{\shrinkingconst^2 \log(n) \hesssmalltrue}{\paramdim}\right) - \frac{\shrinkingconst^2 \log(n) \hesssmalltrue}{\paramdim} + 1\right] \right\} \\
	&\leq \exp\left\{\frac{\paramdim}{2}\left[- \frac{\shrinkingconst^2 \log(n) \hesssmalltrue}{2\paramdim} + 1\right] \right\} \\
	&= e^{\paramdim/2} n^{-\shrinkingconst^2 \hesssmalltrue/4},
\]
where we have used that $\log(x) < x/2$ for all $x>0$.

Thus,
\*[
	&\hspace{-2em}\int_{\ballc{\paramdim}{\parammode }{\shrinkingrad} \cap \ball{\paramdim}{\paramtrue}{\universalradius}}  \exp\{ {\llandp(\param) - \llandp(\parammode)} \}\dee\param \\
	&\leq  
	\left(\frac{2\pi e}{\hesssmalltrue} \right)^{\paramdim/2} n^{-\shrinkingconst^2 \hesssmalltrue/4 - \paramdim/2}.
\]

By \cref{assn:consistency}, we have $ \underset{n\to\infty}{\text{\upshape lim}} \datatrueprobn [ \parammode \in \ball{\paramdim}{\paramtrue}{\shrinkingrad} ] = 1$, and by \cref{assn:hessian} we have $ \underset{n\to\infty}{\text{\upshape lim}} \datatrueprobn [ \hesssmalltrue \geq  \hesssmall] = 1$, giving the statement of the lemma.

\end{proof}

\subsection{Proof of \cref{lem:lowerbound_quad}}
For all $\quadpointvec \in \quadpointset$, a third order Taylor expansion of the posterior around $\parammode$ gives 
\*[
\frac{\poststar{
	\logpostchol \, \quadpointvec + \parammode}}{\poststar{\parammode}} &= \exp\left\{-\frac{1}{2} (\logpostchol \, \quadpointvec)\Tr \logposthess( \parammode ) (\logpostchol \, \quadpointvec) + \frac{1}{6} (\logpostchol \, \quadpointvec)_{i_1 i_2 i_3}\partial^{i_1 i_2 i_3} \llandp(\parammid^{\logpostchol\,\quadpointvec+\parammode}) \right\}\\
	&=\exp\left\{-\frac{1}{2} \quadpointvec\Tr \quadpointvec + \frac{1}{6}(\logpostchol \, \quadpointvec)_{i_1 i_2 i_3}\partial^{i_1 i_2 i_3} \llandp(\parammid^{\logpostchol\,\quadpointvec+\parammode}) \right\},
\]
where $\parammid^{\logpostchol\,\quadpointvec+\parammode} = \midweight_\quadpointvec (\logpostchol\,\quadpointvec+\parammode) + (1-\midweight_\quadpointvec) \parammode$ for some $\midweight_\quadpointvec \in [0,1]$.
If $\parammode \in \ball{\paramdim}{\paramtrue}{\universalradius/2}$ and $\hesssmallmode \geq \hesssmall$, then
\*[
	\Norm{\parammid^{\logpostchol\,\quadpointvec+\parammode} - \paramtrue}_2
	&= \Norm{\midweight_\quadpointvec (\logpostchol \, \quadpointvec + \parammode) + (1-\midweight_\quadpointvec)\parammode - \midweight_\quadpointvec \paramtrue - (1-\midweight_\quadpointvec)\paramtrue}_2 \\
	&\leq \Norm{\logpostchol \, \quadpointvec}_{2} + \Norm{\parammode - \paramtrue}_2 \\
	&\leq [\eigen_{\paramdim}(\logposthess(\parammode))]^{-1/2} \quadpointbound + \universalradius/2 \\
	&\leq (\hesssmallmode \, n)^{-1/2} \quadpointbound + \universalradius/2 \\
	&\leq (\hesssmall \, n)^{-1/2} \quadpointbound + \universalradius/2.
\]

That is, under these conditions, for large enough $n$ we have $\parammid^{\logpostchol\,\quadpointvec+\parammode} \in \ball{\paramdim}{\paramtrue}{\universalradius}$. Further, these conditions imply
\*[
\frac{1}{6}\absbig{ (\logpostchol \, \quadpointvec)_{i_1 i_2 i_3}\partial^{i_1 i_2 i_3} \llandp(\parammid^{\logpostchol\,\quadpointvec+\parammode}) } 
&\leq \frac{1}{6} \paramdim^3 \max_{\quadpointvec\in\quadpointset} \norm{\logpostchol \, \quadpointvec}_2^3
\ \max_{\quadpointvec\in\quadpointset, (i_1, i_2, i_3 ) \in [\paramdim] } \absbig{ \partial^{i_1 i_2 i_3} \llandp(\parammid^{\logpostchol\,\quadpointvec+\parammode}) } \\
&\leq \frac{1}{6} \paramdim^3 (\hesssmall \, n)^{-3/2} \quadpointbound
\ \max_{\quadpointvec\in\quadpointset, (i_1, i_2, i_3 ) \in [\paramdim] } \absbig{ \partial^{i_1 i_2 i_3} \llandp(\parammid^{\logpostchol\,\quadpointvec+\parammode}) }
:= B(n).
\]

Thus, using the crude bound $e^x \geq 1-x$,
\[\label{eqn:lowerbound_frac}
	\frac{\quadadaptmarg{\quadpointset}{\weight}}{\dist(\data \setdelim \parammode)}
	&= \dist(\parammode) 
	\abssmall{\logpostchol} \sum\limits_{\quadpointvec\in\quadpointset} 
	\weight(\quadpointvec) 
	\frac{\poststar{
	\logpostchol \, \quadpointvec + \parammode}}{\poststar{\parammode}} \\
	&= \dist(\parammode)
	\abssmall{\logpostchol} \sum\limits_{\quadpointvec\in\quadpointset} 
	\weight(\quadpointvec) 
	\exp\left\{-\frac{1}{2} \quadpointvec\Tr \quadpointvec +  \frac{1}{6} (\logpostchol \, \quadpointvec)_{i_1 i_2 i_3}\partial^{i_1 i_2 i_3} \llandp(\parammid^{\logpostchol\,\quadpointvec+\parammode}) \right\}\\
	&\geq \dist(\parammode) 
	\abssmall{\logpostchol} \sum\limits_{\quadpointvec\in\quadpointset} 
	\weight(\quadpointvec) 
	\exp\left\{-\frac{1}{2} \quadpointvec\Tr \quadpointvec \right\}[1 - B(n)] \\
	&\geq \dist(\parammode) 
	\abssmall{\logpostchol}\left[ \sum\limits_{\quadpointvec\in\quadpointset} 
	\weight(\quadpointvec) 
	\exp\left\{-\frac{1}{2} \quadpointvec\Tr \quadpointvec \right\}  -  
	B(n) \sum\limits_{\quadpointvec\in\quadpointset} 
	\abssmall{\weight(\quadpointvec)}  \right] \\
	&= \dist(\parammode) 
	\abssmall{\logpostchol}\left[ (2\pi)^{\paramdim/2}  - 
	B(n)\sum\limits_{\quadpointvec\in\quadpointset} 
	\abssmall{\weight(\quadpointvec)}  \right],
\]
where the last step uses $\property{\quadnum}{\paramdim}$.

Note we also have, if $\hessbigmode \leq \hessbig$,
\*[
	\abssmall{\logpostchol}
	= \sqrt{\absbig{[\logposthess(\parammode)]^{-1}}} 
	= \absbig{\logposthess(\parammode)}^{-1/2} 
	\geq \Big[\eigen_1(\logposthess(\parammode))\Big]^{-\paramdim/2}
	= (\hessbigmode \, n)^{-\paramdim/2}
	\geq (\hessbig \, n)^{-\paramdim/2}.
\]

Under these conditions, $B(n) \rightarrow 0 $ and  $\sum_{\quadpointvec\in\quadpointset} \abssmall{\weight(\quadpointvec)}$ is bounded by some constant.
Thus, for large enough $n$,
\cref{eqn:lowerbound_frac} gives that if
$\parammode \in \ball{\paramdim}{\paramtrue}{\universalradius/2}$, $\hesssmallmode \geq \hesssmall$, $\hessbigtrue \leq \hessbig$, and $\hessbigmode \leq \hessbig$, then
\*[
	\frac{\quadadaptmarg{\quadpointset}{\weight}}{\dist(\data \setdelim \parammode)}
	\geq 
	\priorsmall
	(\hessbig \, n)^{-\paramdim/2} \left[ (2\pi)^{\paramdim/2}  - 
	B(n)\sum\limits_{\quadpointvec\in\quadpointset} 
	\abssmall{\weight(\quadpointvec)}  \right] 
	\geq \constt' n^{-\paramdim/2},
\]
where $\constt' > 0$ is some constant.
 The statement of the lemma then follows from \cref{assn:consistency,assn:prior,assn:hessian}, as well as \cref{lem:constant_convergence}.
\manualendproof

%% file: sections/proof-D-big-lemma.tex
\section{Proof of \cref{LEM:NORMALIZING_CONSTANT_AS}}\label{sec:proof_const_as}

We first state the argument for $\quadnum\geq2$, and then address the slight modifications that must be made when $\quadnum=1$ in \cref{sec:laplace_proof}.
For all $\param \in \paramspace$, the $2\quadnum$th order Taylor expansion of $\llandp$ around $\parammode$ gives
\[\label{eqn:main-lemma_first-taylor}
	&\hspace{-1em}\exp\left\{\llandp(\param) -  \llandp(\parammode) \right\} \\
	&=  \exp\left\{ \frac{1}{2} (\param - \parammode)_{i_1 i_2} \partial^{i_1 i_2} \llandp(\parammode)\right\} \times \\
	&\qquad \exp\left\{\sum_{j = 3}^{2\quadnum - 1} (\param - \parammode)_{[i_1\dots i_j]!} \partial^{i_1 \dots i_j  } \llandp(\parammode) +
	(\param - \parammode)_{[i_1\dots i_{2\quadnum}]!}  \partial^{i_1\dots i_{2\quadnum}} \llandp(\parammodemid^\param) \right\},
\]
where $\parammodemid^\param = \modemidweight_{\param,n} \parammode + (1-\modemidweight_{\param,n})\param$ for some $\modemidweight_{\param,n} \in [0,1]$.
Further,
\[\label{eqn:main-lemma_normal-appears}
	\exp\left\{ \frac{1}{2} (\param - \parammode)_{i_1 i_2} \partial^{i_1 i_2} \llandp(\parammode)\right\}
	&= \exp\left\{-\frac{1}{2} (\param - \parammode)\Tr [\logposthess(\parammode)] (\param - \parammode) \right\}	\\
	&= (2\pi)^{\paramdim/2} \absbig{[\logposthess(\parammode)]^{-1}}^{1/2} \normaldens{\param}{\parammode}{[\logposthess(\parammode)]^{-1}} \\
	&= (2\pi)^{\paramdim/2} \abssmall{\logpostchol} \normaldens{\param}{\parammode}{[\logposthess(\parammode)]^{-1}},
\]

For notational simplicity, we define
\*[
	\multitaylorfunc(\param) =  \sum_{s = 3}^{2\quadnum - 1} (\param - \parammode)_{[i_1\dots i_s]!} \partial^{i_1 \dots i_s} \llandp(\parammode) +
	(\param - \parammode)_{[i_1\dots i_{2\quadnum}]!}  \partial^{i_1\dots i_{2\quadnum}} \llandp(\parammodemid^\param). \label{eq:A_def}
\]

Substituting \cref{eqn:main-lemma_first-taylor,eqn:main-lemma_normal-appears} into the difference of interest, we obtain
\[\label{eqn:main-lemma_polynomial-combo}
	&\hspace{-2em}
	\absbig{
		\int_{\ball{\paramdim}{\parammode}{\shrinkingrad}}\exp\left\{\llandp(\param) - \llandp(\parammode) \right\}\dee\param
		- 
		\abssmall{\logpostchol} \sum\limits_{\quadpointvec\in\quadpointset} \weight(\quadpointvec) \exp\left\{\llandp(\logpostchol \, \quadpointvec + \parammode) - \llandp(\parammode) \right\}
	} \\
	&= (2\pi)^{\paramdim/2} \abssmall{\logpostchol} \\
	&\qquad\times\Bigg \lvert 
	\int_{\ball{\paramdim}{\parammode}{\shrinkingrad}}
	\normaldens{\param}{\parammode}{[\logposthess(\parammode)]^{-1}} 
	\exp\left\{\multitaylorfunc(\param) \right\}\dee\param \\
	&\qquad\qquad- 
	\abssmall{\logpostchol} \sum\limits_{\quadpointvec\in\quadpointset} \weight(\quadpointvec) \normaldens{\logpostchol \, \quadpointvec + \parammode}{\parammode}{[\logposthess(\parammode)]^{-1}} \exp\left\{\multitaylorfunc\Big(\logpostchol \, \quadpointvec + \parammode\Big)\right\}
	\Bigg \rvert \\
	&= (2\pi)^{\paramdim/2} \abssmall{\logpostchol} \\
	&\qquad\times\Bigg \lvert 
	\int_{\ball{\paramdim}{\parammode}{\shrinkingrad}}
	\normaldens{\param}{\parammode}{[\logposthess(\parammode)]^{-1}} 
	\exp\left\{\multitaylorfunc(\param) \right\}\dee\param \\
	&\qquad\qquad\qquad - 
	\sum\limits_{\quadpointvec\in\quadpointset} \weight(\quadpointvec) \normaldens{\quadpointvec}{0}{\identmat_\paramdim} \exp\left\{\multitaylorfunc\Big(\logpostchol \, \quadpointvec + \parammode\Big)\right\}
	\Bigg \rvert.
\]

Now, evaluating the $(\quadnumadj+2)$th order Taylor expansion of $e^x$ around zero at $\multitaylorfunc(\param)$ gives
\*[
	\exp\{\multitaylorfunc(\param)\} = \sum_{j = 0}^{\quadnumadj + 1} \frac{\multitaylorfunc(\param)^j}{j!} + \frac{\multitaylorfunc(\param)^{(\quadnumadj+2)}}{(\quadnumadj+2)!} \exp\{ \midmultitaylorfunc(\param) \},
\]
where $\midmultitaylorfunc(\param) = \midweight_n \multitaylorfunc(\param)$ for some $\midweight_n\in [0,1]$. 
For each $j \leq \quadnumadj$ and $\tvec = (t_3,\dots,t_{2\quadnum})$ satisfying $\sum_{s=3}^{2\quadnum} t_s = j$, denote the multinomial coefficient by
\*[
	\multicoeff{j}{\tvec} = {j \choose t_3,\dots,t_{2\quadnum}}.
\]

This simplifies notation, so we can write $\multitaylorfunc(\param)^j$ as
\*[
	\sum_{\tvec: t_3+\dots+t_{2\quadnum} = j} \multicoeff{j}{\tvec} \left\{ \left(\prod_{s =3}^{2\quadnum-1} \left[(\param - \parammode)_{[i_1\dots i_s]!} \partial^{i_1 \dots i_s  } \llandp(\parammode)\right]^{t_s} \right) \left[(\param - \parammode)_{[i_1\dots i_{2\quadnum}]!} \partial^{i_1 \dots i_{2\quadnum}} \llandp(\parammodemid^\param) \right]^{t_{2\quadnum}} \right\}.
\] 

For each $s \in \{3,\dots,2\quadnum-1\}$, note that $(\param - \parammode)_{[i_1\dots i_s]!} \partial^{i_1 \dots i_s  } \llandp(\parammode)$ is a polynomial in $\param$ of total order $s$. Thus, $\multitaylorfunc(\param)^j$ is actually equal to
\*[
	&\sum_{\tvec \in \smalltset{j}} \multicoeff{j}{\tvec} \left\{ \left(\prod_{s =3}^{2\quadnum-1} \left[(\param - \parammode)_{[i_1\dots i_s]!} \partial^{i_1 \dots i_s  } \llandp(\parammode)\right]^{t_s} \right) \left[(\param - \parammode)_{[i_1\dots i_{2\quadnum}]!} \partial^{i_1 \dots i_{2\quadnum}} \llandp(\parammodemid^\param) \right]^{t_{2\quadnum}} \right\} \\
	&\qquad+ \sum_{\tvec \in \bigtset{j}} \multicoeff{j}{\tvec} \left\{ \left(\prod_{s =3}^{2\quadnum-1} \left[(\param - \parammode)_{[i_1\dots i_s]!} \partial^{i_1 \dots i_s  } \llandp(\parammode)\right]^{t_s} \right) \left[(\param - \parammode)_{[i_1\dots i_{2\quadnum}]!} \partial^{i_1 \dots i_{2\quadnum}} \llandp(\parammodemid^\param) \right]^{t_{2\quadnum}} \right\}.
\]

In summary, we have broken $\multitaylorfunc(\param)^j$ up into two cases: the terms in the first sum are polynomials in $\param$ of total order at most $2\quadnum-1$, while all polynomials in $\param$ contained in the terms of the second sum have degree at least $2\quadnum$.
Importantly, for any $j$, all $\tvec \in \smalltset{j}$ satisfy $t_{2\quadnum}=0$ necessarily, which means there is no dependence of $\parammodemid^\param$ in these polynomials.
Since $\quadrule{\quadpointset}{\weight}$ satisfies $\property{\quadnum}{\paramdim}$,
\*[
	&\hspace{-2em}\sum\limits_{\quadpointvec\in\quadpointset} \weight(\quadpointvec) \normaldens{\quadpointvec}{0}{\identmat_\paramdim} \exp\left\{\multitaylorfunc\Big(\logpostchol \, \quadpointvec + \parammode\Big)\right\} \\
	&= 
	\sum_{j=0}^\quadnumadj \int_{\paramspace} \normaldens{\param}{0}{\identmat_\paramdim}
	\sum_{\tvec \in \smalltset{j}} \multicoeff{j}{\tvec}
	\prod_{s =3}^{2\quadnum-1} \left[(\logpostchol \, \param)_{[i_1\dots i_s]!} \partial^{i_1 \dots i_s  } \llandp(\parammode)\right]^{t_s} \dee \param \\
	&\qquad +
	\sum_{j=1}^\quadnumadj \sum\limits_{\quadpointvec\in\quadpointset} \weight(\quadpointvec) \normaldens{\quadpointvec}{0}{\identmat_\paramdim}
	\sum_{\tvec \in \bigtset{j}} \multicoeff{j}{\tvec}
	\Bigg\{ \left(\prod_{s =3}^{2\quadnum-1} \left[(\logpostchol \, \quadpointvec)_{[i_1\dots i_s]!} \partial^{i_1 \dots i_s  } \llandp(\parammode)\right]^{t_s} \right) \\
	&\qquad\qquad\qquad\qquad\qquad\qquad\qquad\qquad\qquad\qquad
	\left[(\logpostchol \, \quadpointvec)_{[i_1\dots i_{2\quadnum}]!} \partial^{i_1 \dots i_{2\quadnum}} \llandp(\parammodemid^{\logpostchol \, \quadpointvec + \parammode}) \right]^{t_{2\quadnum}} \Bigg\} \\
	&\qquad + 
	\sum\limits_{\quadpointvec\in\quadpointset} \weight(\quadpointvec) \normaldens{\quadpointvec}{0}{\identmat_\paramdim} \frac{\multitaylorfunc(\logpostchol \, \quadpointvec + \parammode)^{(\quadnumadj + 1)}}{(\quadnumadj + 1)!}  \\
	&\qquad + 
	\sum\limits_{\quadpointvec\in\quadpointset} \weight(\quadpointvec) \normaldens{\quadpointvec}{0}{\identmat_\paramdim} \frac{\multitaylorfunc(\logpostchol \, \quadpointvec + \parammode)^{(\quadnumadj + 2)}}{(\quadnumadj + 2)!} \exp\{ \midmultitaylorfunc(\logpostchol \, \quadpointvec + \parammode) \}.
\]

Note that we have deliberately treated the $\kappa+1$ and $\kappa+2$ powers of $\multitaylorfunc(\param)$ separately, which is crucial to obtain the correct rate and to handle the remainder term $\exp\{\midmultitaylorfunc(\logpostchol \, \quadpointvec + \parammode)\}$.
Applying a change of variable, splitting up $\paramspace$, and recalling that the quadrature process applied to polynomials of total order $2\quadnum - 1$ is exactly the same as the integral of this polynomial then gives
\[\label{eqn:sum-high-low-split}
	&\sum\limits_{\quadpointvec\in\quadpointset} \weight(\quadpointvec) \normaldens{\quadpointvec}{0}{\identmat_\paramdim} \exp\left\{\multitaylorfunc\Big(\logpostchol \, \quadpointvec + \parammode\Big)\right\} \\
	&= 
	\sum_{j=0}^\quadnumadj \int_{\ball{\paramdim}{\parammode}{\shrinkingrad}} \normaldens{\param}{\parammode}{[\logposthess(\parammode)]^{-1}}
	\sum_{\tvec \in \smalltset{j}} \multicoeff{j}{\tvec}
	\prod_{s =3}^{2\quadnum-1} \left[(\param - \parammode)_{[i_1\dots i_s]!} \partial^{i_1 \dots i_s  } \llandp(\parammode)\right]^{t_s} \dee \param \\
	&\qquad +
	\sum_{j=0}^\quadnumadj \int_{\ballc{\paramdim}{\parammode}{\shrinkingrad}} \normaldens{\param}{\parammode}{[\logposthess(\parammode)]^{-1}}
	\sum_{\tvec \in \smalltset{j}} \multicoeff{j}{\tvec}
	\prod_{s =3}^{2\quadnum-1} \left[(\param - \parammode)_{[i_1\dots i_s]!} \partial^{i_1 \dots i_s  } \llandp(\parammode)\right]^{t_s} \dee \param \\
	&\qquad +
	\sum_{j=1}^\quadnumadj \sum\limits_{\quadpointvec\in\quadpointset} \weight(\quadpointvec) \normaldens{\quadpointvec}{0}{\identmat_\paramdim}
	\sum_{\tvec \in \bigtset{j}} \multicoeff{j}{\tvec}
	\Bigg\{ \left(\prod_{s =3}^{2\quadnum-1} \left[(\logpostchol \, \quadpointvec)_{[i_1\dots i_s]!} \partial^{i_1 \dots i_s  } \llandp(\parammode)\right]^{t_s} \right) \\
	&\qquad\qquad\qquad\qquad\qquad\qquad\qquad\qquad\qquad\qquad
	\left[(\logpostchol \, \quadpointvec)_{[i_1\dots i_{2\quadnum}]!} \partial^{i_1 \dots i_{2\quadnum}} \llandp(\parammodemid^{\logpostchol \, \quadpointvec + \parammode}) \right]^{t_{2\quadnum}} \Bigg\} \\
	&\qquad + 
	\sum\limits_{\quadpointvec\in\quadpointset} \weight(\quadpointvec) \normaldens{\quadpointvec}{0}{\identmat_\paramdim} \frac{\multitaylorfunc(\logpostchol \, \quadpointvec + \parammode)^{(\quadnumadj + 1 )}}{(\quadnumadj + 1)!}  \\
	&\qquad + 
	\sum\limits_{\quadpointvec\in\quadpointset} \weight(\quadpointvec) \normaldens{\quadpointvec}{0}{\identmat_\paramdim} \frac{\multitaylorfunc(\logpostchol \, \quadpointvec + \parammode)^{(\quadnumadj + 2)}}{(\quadnumadj + 2)!} \exp\{ \midmultitaylorfunc(\logpostchol \, \quadpointvec + \parammode) \}.
\]

Similarly,
\[\label{eqn:int-high-low-split}
	&\int_{\ball{\paramdim}{\parammode}{\shrinkingrad}}
	\normaldens{\param}{\parammode}{[\logposthess(\parammode)]^{-1}} 
	\exp\left\{\multitaylorfunc(\param) \right\}\dee\param \\
	&= 
	\sum_{j=0}^\quadnumadj \int_{\ball{\paramdim}{\parammode}{\shrinkingrad}} \normaldens{\param}{\parammode}{[\logposthess(\parammode)]^{-1}}
	\sum_{\tvec \in \smalltset{j}} \multicoeff{j}{\tvec}
	\prod_{s =3}^{2\quadnum-1} \left[(\param - \parammode)_{[i_1\dots i_s]!} \partial^{i_1 \dots i_s  } \llandp(\parammode)\right]^{t_s} \dee \param \\
	&\qquad +
	\sum_{j=1}^\quadnumadj \int_{\ball{\paramdim}{\parammode}{\shrinkingrad}} \normaldens{\param}{\parammode}{[\logposthess(\parammode)]^{-1}}
	\sum_{\tvec \in \bigtset{j}} \multicoeff{j}{\tvec}
	\Bigg\{ \left(\prod_{s =3}^{2\quadnum-1} \left[(\param - \parammode)_{[i_1\dots i_s]!} \partial^{i_1 \dots i_s  } \llandp(\parammode)\right]^{t_s} \right) \\
	&\qquad\qquad\qquad\qquad\qquad\qquad\qquad\qquad\qquad\qquad\qquad\qquad
	\left[(\param - \parammode)_{[i_1\dots i_{2\quadnum}]!} \partial^{i_1 \dots i_{2\quadnum}} \llandp(\parammodemid^{\param}) \right]^{t_{2\quadnum}} \Bigg\} \dee \param \\
	&\qquad +
	\int_{\ball{\paramdim}{\parammode}{\shrinkingrad}} \normaldens{\param}{\parammode}{[\logposthess(\parammode)]^{-1}} 
	\frac{\multitaylorfunc(\param)^{(\quadnumadj + 1)}}{(\quadnumadj + 1)!} \dee \param\\
	&\qquad +
	\int_{\ball{\paramdim}{\parammode}{\shrinkingrad}} \normaldens{\param}{\parammode}{[\logposthess(\parammode)]^{-1}} 
	\frac{\multitaylorfunc(\param)^{(\quadnumadj + 2)}}{(\quadnumadj + 2)!} \exp\{ \midmultitaylorfunc(\param) \} \dee \param.
\]

Substituting \cref{eqn:sum-high-low-split,eqn:int-high-low-split} into the absolute difference in \cref{eqn:main-lemma_polynomial-combo} and applying triangle inequality gives
\*[
	&\Bigg\lvert
	\int_{\ball{\paramdim}{\parammode}{\shrinkingrad}}
	\normaldens{\param}{\parammode}{[\logposthess(\parammode)]^{-1}} 
	\exp\left\{\multitaylorfunc(\param) \right\}\dee\param \\
	&\qquad\qquad\qquad - 
	\sum\limits_{\quadpointvec\in\quadpointset} \weight(\quadpointvec) \normaldens{\quadpointvec}{0}{\identmat_\paramdim} \exp\left\{\multitaylorfunc\Big(\logpostchol \, \quadpointvec + \parammode\Big)\right\}
	\Bigg \rvert \\
	&\leq
	\sum_{j=1}^\quadnumadj \int_{\ball{\paramdim}{\parammode}{\shrinkingrad}} \normaldens{\param}{\parammode}{[\logposthess(\parammode)]^{-1}}
	\sum_{\tvec \in \bigtset{j}} \multicoeff{j}{\tvec}
	\Bigg\{ \left(\prod_{s =3}^{2\quadnum-1} \absbig{(\param - \parammode)_{[i_1\dots i_s]!} \partial^{i_1 \dots i_s  } \llandp(\parammode)}^{t_s} \right) \\
	&\qquad\qquad\qquad\qquad\qquad\qquad\qquad\qquad\qquad\qquad\qquad
	\absbig{(\param - \parammode)_{[i_1\dots i_{2\quadnum}]!} \partial^{i_1 \dots i_{2\quadnum}} \llandp(\parammodemid^{\param})}^{t_{2\quadnum}} \Bigg\} \dee \param \\
	&\qquad +
	\absbig{\int_{\ball{\paramdim}{\parammode}{\shrinkingrad}} \normaldens{\param}{\parammode}{[\logposthess(\parammode)]^{-1}} 
	\frac{\multitaylorfunc(\param)^{(\quadnumadj + 1)}}{(\quadnumadj + 1)!}  \dee \param} \\
	&\qquad +
	\int_{\ball{\paramdim}{\parammode}{\shrinkingrad}} \normaldens{\param}{\parammode}{[\logposthess(\parammode)]^{-1}} 
	\absbig{\frac{\multitaylorfunc(\param)^{(\quadnumadj + 2)}}{(\quadnumadj + 2)!}} \exp\{ \midmultitaylorfunc(\param) \} \dee \param \\
	&\qquad +
	\sum_{j=0}^\quadnumadj \int_{\ballc{\paramdim}{\parammode}{\shrinkingrad}} \normaldens{\param}{\parammode}{[\logposthess(\parammode)]^{-1}}
	\sum_{\tvec \in \smalltset{j}} \multicoeff{j}{\tvec}
	\prod_{s =3}^{2\quadnum-1} \absbig{(\param - \parammode)_{[i_1\dots i_s]!} \partial^{i_1 \dots i_s  } \llandp(\parammode)}^{t_s} \dee \param \\
	&\qquad +
	\sum_{j=1}^\quadnumadj \sum\limits_{\quadpointvec\in\quadpointset} \weight(\quadpointvec) \normaldens{\quadpointvec}{0}{\identmat_\paramdim}
	\sum_{\tvec \in \bigtset{j}} \multicoeff{j}{\tvec}
	\Bigg\{ \left(\prod_{s =3}^{2\quadnum-1} \absbig{(\logpostchol \, \quadpointvec)_{[i_1\dots i_s]!} \partial^{i_1 \dots i_s  } \llandp(\parammode)}^{t_s} \right) \\
	&\qquad\qquad\qquad\qquad\qquad\qquad\qquad\qquad\qquad\qquad
	\absbig{(\logpostchol \, \quadpointvec)_{[i_1\dots i_{2\quadnum}]!} \partial^{i_1 \dots i_{2\quadnum}} \llandp(\parammodemid^{\logpostchol \, \quadpointvec + \parammode})}^{t_{2\quadnum}} \Bigg\} \\
	&\qquad + 
	\absbig{\sum\limits_{\quadpointvec\in\quadpointset} \weight(\quadpointvec) \normaldens{\quadpointvec}{0}{\identmat_\paramdim} 
	\frac{\multitaylorfunc(\logpostchol \, \quadpointvec + \parammode)^{(\quadnumadj + 1)}}{(\quadnumadj + 1)!}}  \\
	&\qquad + 
	\sum\limits_{\quadpointvec\in\quadpointset} \weight(\quadpointvec) \normaldens{\quadpointvec}{0}{\identmat_\paramdim} 
	\absbig{\frac{\multitaylorfunc(\logpostchol \, \quadpointvec + \parammode)^{(\quadnumadj + 1)}}{(\quadnumadj + 1)!}} 
	\exp\{ \midmultitaylorfunc(\logpostchol \, \quadpointvec + \parammode) \}.
\]

Next, observe that $\midmultitaylorfunc(\param) \leq \abssmall{\multitaylorfunc(\param)}$ and $\midmultitaylorfunc(\logpostchol \, \quadpointvec + \parammode) \leq \lvert \multitaylorfunc(\logpostchol \, \quadpointvec + \parammode)\rvert $. 
Additionally, by definition of the radius for which $\derivboundmode$ is a bound on the derivatives, it holds that $\abssmall{\partial^\alpha \llandp(\parammode)} \leq \derivboundmode$, $\abssmall{\partial^\alpha \llandp(\parammodemid^{\logpostchol \, \quadpointvec + \parammode})} \leq \derivboundmode$, and $\abssmall{\partial^\alpha \llandp(\parammodemid^{\param})} \leq \derivboundmode$ for any $\alpha$, $\quadpointvec \in \quadpointset$, and $\param \in \ball{\paramdim}{\parammode}{\shrinkingrad}$. Thus, we have
\[\label{eqn:post-cancel-all-terms}
	&\Bigg\lvert
	\int_{\ball{\paramdim}{\parammode}{\shrinkingrad}}
	\normaldens{\param}{\parammode}{[\logposthess(\parammode)]^{-1}} 
	\exp\left\{\multitaylorfunc(\param) \right\}\dee\param \\
	&\qquad\qquad\qquad - 
	\sum\limits_{\quadpointvec\in\quadpointset} \weight(\quadpointvec) \normaldens{\quadpointvec}{0}{\identmat_\paramdim} \exp\left\{\multitaylorfunc\Big(\logpostchol \, \quadpointvec + \parammode\Big)\right\}
	\Bigg \rvert \\
	&\leq
	\sum_{j=1}^\quadnumadj (\derivboundmode)^{j} \int_{\ball{\paramdim}{\parammode}{\shrinkingrad}} \normaldens{\param}{\parammode}{[\logposthess(\parammode)]^{-1}}
	\sum_{\tvec \in \bigtset{j}} \multicoeff{j}{\tvec} 
	\prod_{s =3}^{2\quadnum} \absbig{\sum_{i_1,\dots,i_s \in [\paramdim]} \prod_{\iota=1}^s (\param-\parammode)_{i_\iota}}^{t_s} \dee \param  \\
	&\qquad +
	\absbig{\int_{\ball{\paramdim}{\parammode}{\shrinkingrad}} \normaldens{\param}{\parammode}{[\logposthess(\parammode)]^{-1}} 
	\frac{\multitaylorfunc(\param)^{(\quadnumadj + 1)}}{(\quadnumadj + 1)!} \dee \param}  \\
	&\qquad +
	\int_{\ball{\paramdim}{\parammode}{\shrinkingrad}} \normaldens{\param}{\parammode}{[\logposthess(\parammode)]^{-1}} 
	\absbig{\frac{\multitaylorfunc(\param)^{(\quadnumadj + 2)}}{(\quadnumadj + 2)!}}
	\exp\left\{\absbig{\multitaylorfunc(\param)} \right\} \dee \param\\
	&\qquad +
	\sum_{j=0}^\quadnumadj (\derivboundmode)^{j} \int_{\ballc{\paramdim}{\parammode}{\shrinkingrad}} \normaldens{\param}{\parammode}{[\logposthess(\parammode)]^{-1}}
	\sum_{\tvec \in \smalltset{j}} \multicoeff{j}{\tvec}
	\prod_{s =3}^{2\quadnum-1} \absbig{\sum_{i_1,\dots,i_s \in [\paramdim]} \prod_{\iota=1}^s (\param-\parammode)_{i_\iota}}^{t_s} \dee \param \\
	&\qquad +
	\sum_{j=1}^\quadnumadj (\derivboundmode)^{j} \sum\limits_{\quadpointvec\in\quadpointset} \weight(\quadpointvec) \normaldens{\quadpointvec}{0}{\identmat_\paramdim}
	\sum_{\tvec \in \bigtset{j}} \multicoeff{j}{\tvec}
	\prod_{s =3}^{2\quadnum} \absbig{\sum_{i_1,\dots,i_s \in [\paramdim]} \prod_{\iota=1}^s (\logpostchol \, \quadpointvec)_{i_\iota}}^{t_s} \\
	&\qquad + 
	\absbig{\sum\limits_{\quadpointvec\in\quadpointset} \weight(\quadpointvec) \normaldens{\quadpointvec}{0}{\identmat_\paramdim} 
	\frac{\multitaylorfunc(\logpostchol \, \quadpointvec + \parammode)^{(\quadnumadj + 1)}}{(\quadnumadj + 1)!}} \\
	&\qquad + 
	\sum\limits_{\quadpointvec\in\quadpointset} \weight(\quadpointvec) \normaldens{\quadpointvec}{0}{\identmat_\paramdim} 
	\absbig{\frac{\multitaylorfunc(\logpostchol \, \quadpointvec + \parammode)^{(\quadnumadj + 2)}}{(\quadnumadj + 2)!}} \exp\left\{\absbig{\multitaylorfunc(\logpostchol \, \quadpointvec + \parammode)} \right\}. 
\]

Our next step is to simplify the dependence on $\logposthess(\parammode)$ for some terms. In particular, for every $\param \in \paramspace$,
\*[
	\normaldens{\param}{\parammode}{[\logposthess(\parammode)]^{-1}}
	&= \frac{1}{\abssmall{\logpostchol}(2\pi)^{\paramdim/2}}
	\exp\left\{-\frac{1}{2}(\param-\parammode)\Tr \logposthess(\parammode) (\param-\parammode) \right\} \\
	&\leq \frac{[\eigen_{1}(\logposthess(\parammode))]^{\paramdim/2}}{(2\pi)^{\paramdim/2}}
	\exp\left\{-\frac{1}{2} \eigen_{\paramdim}(\logposthess(\parammode)) (\param-\parammode)\Tr (\param-\parammode) \right\} \\
	&= \frac{[\eigen_{1}(\logposthess(\parammode))]^{\paramdim/2}}{[\eigen_{\paramdim}(\logposthess(\parammode))]^{\paramdim/2}} \normaldens{\param}{\parammode}{[\eigen_{\paramdim}(\logposthess(\parammode))]^{-1} \identmat_\paramdim}.
\]

Finally,
for each $j \leq \quadnumadj$, $\tvec$ such that $\sum_{s=3}^{2\quadnum} t_s = j$, and $\paramdum\in\paramspace$, it holds that
\*[
	\prod_{s=3}^{2\quadnum} \left(\sum_{i_1,\dots,i_s \in [\paramdim]} \prod_{\iota=1}^s \absbig{(\paramdum)_{i_\iota}}\right)^{t_s}
	\leq
	\sum_{i_1,\dots,i_{\tsum} \in [\paramdim]} \prod_{\iota=1}^{\tsum} \absbig{(\paramdum)_{i_\iota}}
	\leq \sum_{i_1,\dots,i_{\tsum} \in [\paramdim]} \frac{1}{\tsum} \sum_{\iota=1}^{\tsum} \absbig{(\paramdum)_{i_\iota}}^{\tsum},
\]
where the first inequality is because all permutations that show up on the LHS must appear in the sum on the RHS by definition, and the second inequality is by the AM-GM inequality. We will further control $1/\tsum$ using that if $\tvec \in \bigtset{j}$ then $\tsum \geq 2\quadnum$, and $\tsum \leq 2\quadnum j \leq 2\quadnum \quadnumadj$ for all $\tvec$.

Applying this simplification of both $\normaldens{\param}{\parammode}{[\logposthess(\parammode)]^{-1}}$ and the polynomial terms, we get, informally,
\[\label{eqn:post-informal-all-terms}
	&\Bigg\lvert
	\int_{\ball{\paramdim}{\parammode}{\shrinkingrad}}
	\normaldens{\param}{\parammode}{[\logposthess(\parammode)]^{-1}} 
	\exp\left\{\multitaylorfunc(\param) \right\}\dee\param
	- 
	\sum\limits_{\quadpointvec\in\quadpointset} \weight(\quadpointvec) \normaldens{\quadpointvec}{0}{\identmat_\paramdim} \exp\left\{\multitaylorfunc\Big(\logpostchol \, \quadpointvec + \parammode\Big)\right\}
	\Bigg \rvert \\
	&\leq
	\text{ High-degree polynomial remainder} 
	+ \text{ Low-degree polynomial remainder}  \\
	&\qquad + \text{ True posterior remainder} 
	+ \text{ Approximate posterior remainder}.
\]
Each of these terms are precisely quantified as follows, and broken up into Terms 1 through 7.

\*[
	&\text{ High-degree polynomial remainder:} \\
	&=
	\left(\frac{\hessbigmode}{\hesssmallmode}\right)^{\paramdim/2} \sum_{j=1}^\quadnumadj (\derivboundmode)^{j} 
	\sum_{\tvec \in \bigtset{j}} \frac{\multicoeff{j}{\tvec} \,  }{2\quadnum} \sum_{i_1,\dots,i_{\tsum} \in [\paramdim]} \sum_{\iota=1}^{\tsum} \\
	&\qquad\underbrace{\qquad\times
	\int_{\ball{\paramdim}{\parammode}{\shrinkingrad}} 
	\normaldens{\param}{\parammode}{[\eigen_{\paramdim}(\logposthess(\parammode))]^{-1} \identmat_\paramdim}
	\absbig{(\param - \parammode)_{i_\iota}}^\tsum \, \dee \param}_{\text{Term 1}} \\
	&\qquad +
	\underbrace{\sum_{j=1}^\quadnumadj (\derivboundmode)^{j} 
	\sum_{\tvec \in \bigtset{j}} \frac{\multicoeff{j}{\tvec} \,  }{2\quadnum} \sum_{i_1,\dots,i_{\tsum} \in [\paramdim]} \sum_{\iota=1}^{\tsum}
	\sum\limits_{\quadpointvec\in\quadpointset} \weight(\quadpointvec) \normaldens{\quadpointvec}{0}{\identmat_\paramdim}
	\absbig{(\logpostchol \, \quadpointvec)_{i_\iota}}^\tsum}_{\text{Term 2}}.
\]
These terms involve terms that are the higher order polynomial terms that are cancelled out by the quadrature process.

\*[
	&\text{ Low-degree polynomial remainder:} \\
	&=
	\left(\frac{\hessbigmode}{\hesssmallmode}\right)^{\paramdim/2}
	\sum_{j=0}^\quadnumadj (\derivboundmode)^{j} 
	\sum_{\tvec \in \smalltset{j}} \frac{\multicoeff{j}{\tvec} \,  }{\tsum} \sum_{i_1,\dots,i_{\tsum} \in [\paramdim]} \sum_{\iota=1}^{\tsum} \\
	&\qquad\underbrace{\qquad\times
	\int_{\ballc{\paramdim}{\parammode}{\shrinkingrad}} 
	\normaldens{\param}{\parammode}{[\eigen_{\paramdim}(\logposthess(\parammode))]^{-1} \identmat_\paramdim}
	\absbig{(\param - \parammode)_{i_\iota}}^\tsum \, \dee \param}_{\text{Term 3}}.
\]
These terms involve the tails integrals of the lower order polynomial terms that are not cancelled out by the quadrature process due to our truncation argument.

\*[
	&\text{ True posterior remainder:} \\
	&=
	\underbrace{\left(\frac{\hessbigmode}{\hesssmallmode}\right)^{\paramdim/2} 
	\int_{\ball{\paramdim}{\parammode}{\shrinkingrad}} 
	\normaldens{\param}{\parammode}{[\eigen_{\paramdim}(\logposthess(\parammode))]^{-1} \identmat_\paramdim}
	\absbig{\multitaylorfunc(\param)}^{(\quadnumadj + 2)} 
	\exp\left\{\absbig{\multitaylorfunc(\param)} \right\} \dee \param}_{\text{Term 4}} \\
	&\qquad +
	\underbrace{ \absbig{
	\int_{\ball{\paramdim}{\parammode}{\shrinkingrad}} 
	\normaldens{\param}{\parammode}{[\logposthess(\parammode)]^{-1}}
	\multitaylorfunc(\param)^{(\quadnumadj + 1)}
	 \dee \param}}_{\text{Term 5}}.
\]
These correspond to the integral of higher order terms that are not cancelled out by the quadrature process, but also have some peculiarities that we need to exploit or address within the proof. Term 4 contains an exponential term that needs to be bounded, and the behaviour of Term 5 changes depending on the value of $\quadnumadj$. 

\*[
	&\text{ Approximate posterior remainder:} \\
	&=
	\underbrace{\sum\limits_{\quadpointvec\in\quadpointset} \weight(\quadpointvec) \normaldens{\quadpointvec}{0}{\identmat_\paramdim} 
	\absbig{\multitaylorfunc(\logpostchol \, \quadpointvec + \parammode)}^{(\quadnumadj + 2)}
	\exp\left\{\absbig{\multitaylorfunc(\logpostchol \, \quadpointvec + \parammode)} \right\}}_{\text{Term 6}}\\
	&\qquad + 
	\underbrace{\absbig{\sum\limits_{\quadpointvec\in\quadpointset} \weight(\quadpointvec) \normaldens{\quadpointvec}{0}{\identmat_\paramdim} 
	\frac{\multitaylorfunc(\logpostchol \, \quadpointvec + \parammode)^{(\quadnumadj + 1)}}{(\quadnumadj + 1)!}}}_{\text{Term 7}}.
\]
These correspond to the numerical summation of the higher order terms that are not cancelled by the quadrature process, but also have some peculiarities that we need to exploit or address within the proof, similar to the ``true posterior terms". Term 6 contains an exponential term that needs to be bounded, and the behaviour of Term 7 changes depending on the value of $\quadnumadj$. We now handle each of these terms separately. 
\subsection{Bounding Term 1 of \cref{eqn:post-informal-all-terms}}\label{sec:term1}

For any $i \in [\paramdim]$, $j \in [\quadnumadj]$, and $\tvec \in \bigtset{j}$,
\*[
	\int_{\ball{\paramdim}{\parammode}{\shrinkingrad}} 
	\normaldens{\param}{\parammode}{[\eigen_{\paramdim}(\logposthess(\parammode))]^{-1} \identmat_\paramdim}
	\abssmall{(\param - \parammode)_i}^\tsum \, \dee \param
	&\leq
	\int_{0}^{\shrinkingrad} 
	\normaldens{\paramidx}{0}{[\eigen_{\paramdim}(\logposthess(\parammode))]^{-1}}
	\abssmall{\paramidx}^\tsum \, \dee \paramidx.
\]
This can be bounded by standard results on the moments of Gaussians. In particular, Eq.\ (18) of \citet{winkelbauer12gaussian} gives
\*[
	&\hspace{-2em}\int_{\ball{\paramdim}{\parammode}{\shrinkingrad}} 
	\normaldens{\param}{\parammode}{[\eigen_{\paramdim}(\logposthess(\parammode))]^{-1} \identmat_\paramdim}
	\abssmall{(\param - \parammode)_i}^\tsum \, \dee \param \\
	&\leq 2^{\tsum/2} \, \Gamma\left(\frac{\tsum+1}{2} \right) [\eigen_{\paramdim}(\logposthess(\parammode))]^{-\frac{1}{2} \tsum} \\
	&= 2^{\tsum/2} \, \Gamma\left(\frac{\tsum+1}{2} \right) (\hesssmallmode \, n)^{-\tsum/2} \\
	&\leq 2^{\quadnum \quadnumadj} \, \Gamma\left(\frac{2\quadnum\quadnumadj+1}{2} \right) (\hesssmallmode \, n)^{-\tsum/2}
\]

\subsection{Bounding Term 2 of \cref{eqn:post-informal-all-terms}}\label{sec:term2}

For any $i \in [\paramdim]$, $j \in [\quadnumadj]$, $\tvec \in \bigtset{j}$, and $\quadpointvec\in\quadpointset$,
\*[
	\absbig{(\logpostchol \, \quadpointvec)_i}^{\tsum}
	\leq \left[\norm{\logpostchol}_{\mathrm{Op}} \norm{\quadpointvec}_2\right]^{\tsum}
	\leq \Big[[\eigen_{\paramdim}(\logposthess(\parammode))]^{-1/2} \norm{\quadpointvec}_2\Big]^{\tsum}
	\leq (\hesssmallmode \, n)^{-\tsum/2} \quadpointbound^{\tsum}.
\]
Thus, since $\sup_{\quadpointvec} \normaldens{\quadpointvec}{0}{\identmat_\paramdim} \leq (2\pi)^{-\paramdim/2}$ and $\abssmall{\quadpointset} = \quadnum^\paramdim$,
\*[
	\sum\limits_{\quadpointvec\in\quadpointset} \weight(\quadpointvec) \normaldens{\quadpointvec}{0}{\identmat_\paramdim}
	\absbig{(\logpostchol \, \quadpointvec)_i}^\tsum
	\leq (\hesssmallmode \, n)^{-\tsum/2} \quadnum^\paramdim (2\pi)^{-\paramdim/2} \quadpointbound^{2\quadnum\quadnumadj}
\]
\subsection{Bounding Term 3 of \cref{eqn:post-informal-all-terms}}\label{sec:term3}

For any $i \in [\paramdim]$, $j \in \{0\} \cup [\quadnumadj]$, and $\tvec \in \smalltset{j}$,
\*[
	&\hspace{-2em}\int_{\ballc{\paramdim}{\parammode}{\shrinkingrad}} 
	\normaldens{\param}{\parammode}{[\eigen_{\paramdim}(\logposthess(\parammode))]^{-1} \identmat_\paramdim}
	\abssmall{(\param - \parammode)_i}^\tsum \, \dee \param \\
	&= 
	\int_{\shrinkingrad}^\infty 
	\normaldens{\paramidx}{0}{[\eigen_{\paramdim}(\logposthess(\parammode))]^{-1}}
	\abssmall{\paramidx}^\tsum \, \dee \paramidx \\
	&= 
	[\eigen_{\paramdim}(\logposthess(\parammode))]^{-1/2}
	\int_{\shrinkingrad\sqrt{\eigen_{\paramdim}(\logposthess(\parammode))}}^\infty 
	\frac{1}{\sqrt{2\pi}}
	e^{-v^2/2}
	\abssmall{v}^\tsum
	\dee v,
\]
where the last step uses the transformation $v = \paramidx  \sqrt{\eigen_{\paramdim}(\logposthess(\parammode))}$.

Then, using the further transformation $x = v^2$,
\*[
	&\hspace{-2em}\int_{\shrinkingrad\sqrt{\eigen_{\paramdim}(\logposthess(\parammode))}}^\infty 
	\frac{1}{\sqrt{2\pi}}
	e^{-v^2/2}
	\abssmall{v}^\tsum
	\dee v \\
	&= \int_{\shrinkingconst\sqrt{\log(n) \hesssmallmode}}^\infty 
	\frac{1}{\sqrt{2\pi}}
	e^{-v^2/2}
	(v^2)^{\tsum/2}
	\dee v \\
	&= \frac{1}{2} \int_{\shrinkingconst^2 \log(n) \hesssmallmode }^\infty
	\frac{x^{-1/2}}{\sqrt{2\pi}}e^{-x/2}
	x^{\tsum/2} \dee x \\
	&= \frac{\Gamma\left(\frac{\tsum+1}{2}\right) 2^{\frac{\tsum+1}{2}} }{2\sqrt{2\pi}}
	\int_{\shrinkingconst^2 \log(n) \hesssmallmode }^\infty
	\frac{1}{\Gamma\left(\frac{\tsum+1}{2}\right) 2^{\frac{\tsum+1}{2}}} x^{\frac{\tsum+1}{2} - 1} e^{-x/2} \dee x.
\]
This integral is just the tail of a chi-square distribution, which we bound using Lemma~3 of \cite{Fan}. In particular, for any $\nu \in \Nats$,
\*[
	\PP\left[ \chi_{\nu}^2 \geq \shrinkingconst^2 \log(n) \hesssmallmode \right] 
	= \PP\left[ \frac{\chi_{\nu}^2}{\nu} > 1 + \zeta_n \right] 
	\leq \exp\left\{\nu(\log(1 + \zeta_n) - \zeta_n)/2 \right\}, 
\]
where $\zeta_n = \shrinkingconst^2 \log(n) \hesssmallmode / \nu - 1$.
Substituting in $\nu = \tsum + 1$, observe that
\*[
	&\hspace{-2em}\exp\left\{\nu(\log(1 + \zeta_n) - \zeta_n)/2 \right\} \\
	&= \exp\left\{\frac{(\tsum+1)}{2}\left[\log\left(\frac{\shrinkingconst^2 \log(n) \hesssmallmode}{\tsum+1}\right) - \frac{\shrinkingconst^2 \log(n) \hesssmallmode}{\tsum+1} + 1\right] \right\} \\
	&\leq \exp\left\{\frac{(\tsum+1)}{2}\left[- \frac{\shrinkingconst^2 \log(n) \hesssmallmode}{2(\tsum+1)} + 1\right] \right\} \\
	&= e^{(\tsum+1)/2} n^{-\shrinkingconst^2 \hesssmallmode/4},
\]
where we have used that $\log(x) < x/2$ for all $x>0$.

Thus,
\*[
	&\hspace{-2em}\int_{\ballc{\paramdim}{\parammode}{\shrinkingrad}} 
	\normaldens{\param}{\parammode}{[\eigen_{\paramdim}(\logposthess(\parammode))]^{-1} \identmat_\paramdim}
	\abssmall{(\param - \parammode)_i}^\tsum \, \dee \param \\
	&\leq 
	[\eigen_{\paramdim}(\logposthess(\parammode))]^{-1/2}
	\frac{\Gamma\left(\frac{\tsum+1}{2}\right) 2^{\frac{\tsum+1}{2}} }{2\sqrt{2\pi}}
	e^{(d+1)/2} n^{-\shrinkingconst^2 \hesssmallmode/4} \\
	&\leq
	\frac{\Gamma\left(\frac{2\quadnum+1}{2}\right) (2e)^{\frac{2\quadnum+1}{2}} }{2\sqrt{2\pi \hesssmallmode}}
	n^{-\frac{\shrinkingconst^2 \hesssmallmode + 2}{4}}.
\]

\subsection{Bounding Term 4 of \cref{eqn:post-informal-all-terms}}\label{sec:term4}

First, we can control $\exp\{\abssmall{\multitaylorfunc(\param)}\}$ by observing that for any $\param \in \ball{\paramdim}{\parammode}{\shrinkingrad}$, the AM-GM inequality gives
\*[
	\absbig{\multitaylorfunc(\param)}
	&\leq
	\sum_{s = 3}^{2\quadnum - 1} \absbig{(\param - \parammode)_{[i_1\dots i_s]!} \partial^{i_1 \dots i_s} \llandp(\parammode)} +
	\absbig{(\param - \parammode)_{[i_1\dots i_{2\quadnum}]!}  \partial^{i_1\dots i_{2\quadnum}} \llandp(\parammodemid^\param)} \\
	&\leq
	\derivboundmode
	\sum_{s=3}^{2\quadnum} 
	\sum_{i_1,\dots,i_s \in [\paramdim]}
	\prod_{\iota = 1}^s \absbig{(\param - \parammode)_{i_\iota}} \\
	&\leq 
	\derivboundmode
	\sum_{s=3}^{2\quadnum} 
	\sum_{i_1,\dots,i_s \in [\paramdim]}
	\frac{1}{s}
	\sum_{\iota = 1}^s \absbig{(\param - \parammode)_{i_\iota}}^s \\
	&\leq 
	\derivboundmode \,
	\paramdim^{2\quadnum}
	\sum_{s=3}^{2\quadnum} 
	\shrinkingrad^s \\
	&= 
	\derivboundmode \,
	\paramdim^{2\quadnum}
	\sum_{s=3}^{2\quadnum} 
	\shrinkingconst^s \Big(\frac{\log(n)}{n} \Big)^{s/2} \\
	&\leq  (2\quadnum) \, \paramdim^{2\quadnum} \max\{1, \shrinkingconst^{2\quadnum}\} \derivboundmode \Big(\frac{\log(n)}{n} \Big)^{3/2}.
\]

Then, for any $\param\in\ball{\paramdim}{\parammode}{\shrinkingrad}$, $\abssmall{\multitaylorfunc(\param)}^{(\quadnumadj + 2)}$ can be written as
\*[
	&\sum_{\tvec \in \smalltset{\quadnumadj+2}} \multicoeff{\quadnumadj+2}{\tvec} \left\{ \left(\prod_{s =3}^{2\quadnum-1} \absbig{(\param - \parammode)_{[i_1\dots i_s]!} \partial^{i_1 \dots i_s  } \llandp(\parammode)}^{t_s} \right) \absbig{(\param - \parammode)_{[i_1\dots i_{2\quadnum}]!} \partial^{i_1 \dots i_{2\quadnum}} \llandp(\parammodemid^\param)}^{t_{2\quadnum}} \right\} \\
	&+ \sum_{\tvec \in \bigtset{\quadnumadj+2}} \multicoeff{\quadnumadj+2}{\tvec} \left\{ \left(\prod_{s =3}^{2\quadnum-1} \absbig{(\param - \parammode)_{[i_1\dots i_s]!} \partial^{i_1 \dots i_s  } \llandp(\parammode)}^{t_s} \right) \absbig{(\param - \parammode)_{[i_1\dots i_{2\quadnum}]!} \partial^{i_1 \dots i_{2\quadnum}} \llandp(\parammodemid^\param)}^{t_{2\quadnum}} \right\}.
\]

By definition $\sum_{s=3}^{2\quadnum} t_s = \quadnumadj+2$, so it holds that $\sum_{s=3}^{2\quadnum} s t_s \geq 3(\quadnumadj+2)$. Thus, $\abssmall{\multitaylorfunc(\param)}^{(\quadnumadj + 2)}$ can further be written as
\*[
	\sum_{\tvec \in \customtset{\quadnumadj+2}{3(\quadnumadj+2)}} \multicoeff{\quadnumadj+2}{\tvec} \left\{ \left(\prod_{s =3}^{2\quadnum-1} \absbig{(\param - \parammode)_{[i_1\dots i_s]!} \partial^{i_1 \dots i_s  } \llandp(\parammode)}^{t_s} \right) \absbig{(\param - \parammode)_{[i_1\dots i_{2\quadnum}]!} \partial^{i_1 \dots i_{2\quadnum}} \llandp(\parammodemid^\param)}^{t_{2\quadnum}} \right\}.
\]

As before, since $\param\in\ball{\paramdim}{\parammode}{\shrinkingrad}$ all derivatives of $\llandp$ are bounded by $\derivboundmode$. Similarly, the resulting polynomial in $\param$ can be bounded again using the AM-GM inequality, so
\*[
	&\int_{\ball{\paramdim}{\parammode}{\shrinkingrad}} 
	\normaldens{\param}{\parammode}{[\eigen_{\paramdim}(\logposthess(\parammode))]^{-1} \identmat_\paramdim}
	\absbig{\multitaylorfunc(\param)}^{(\quadnumadj + 2)} \dee \param \\
	&\leq (\derivboundmode)^{\quadnumadj+2}
	\sum_{\tvec \in \customtset{\quadnumadj+2}{3(\quadnumadj+2)}} 
	\frac{\multicoeff{\quadnumadj+2}{\tvec}}{3(\quadnumadj+2)} \sum_{i_1,\dots,i_{\tsum} \in [\paramdim]} \sum_{\iota=1}^{\tsum}
	\int_{\ball{\paramdim}{\parammode}{\shrinkingrad}} 
	\normaldens{\param}{\parammode}{[\eigen_{\paramdim}(\logposthess(\parammode))]^{-1} \identmat_\paramdim}
	\abssmall{(\param - \parammode)_i}^\tsum \, \dee \param \\
	&\leq
	(\derivboundmode)^{\quadnumadj+2}
	\sum_{\tvec \in \customtset{\quadnumadj+2}{3(\quadnumadj+2)}}
	\frac{\multicoeff{\quadnumadj+2}{\tvec}}{3(\quadnumadj+2)} \sum_{i_1,\dots,i_{\tsum} \in [\paramdim]} \sum_{\iota=1}^{\tsum}
	2^{\tsum/2} \, \Gamma\left(\frac{\tsum+1}{2} \right) (\hesssmallmode \, n)^{-\tsum/2} \\
	&\leq
	(\derivboundmode)^{\quadnumadj+2}
	\sum_{\tvec \in \customtset{\quadnumadj+2}{3(\quadnumadj+2)}}
	\frac{\multicoeff{\quadnumadj+2}{\tvec}}{3(\quadnumadj+2)} 
	\paramdim^{2\quadnum\quadnumadj} (2\quadnum\quadnumadj)
	2^{\quadnum\quadnumadj} \, \Gamma\left(\frac{2\quadnum\quadnumadj+1}{2} \right) (\hesssmallmode \, n)^{-\tsum/2},
\]
where we have again used Eq.\ (18) of \citet{winkelbauer12gaussian}.

\subsection{Bounding Term 5 of \cref{eqn:post-informal-all-terms}}\label{sec:term5}
In the case that $2\quadnum \neq 2 \pmod 3 $, this term can be treated in the same manner as in \cref{sec:term4}. In particular,
\*[
	&\hspace{-2em}\absbig{
	\int_{\ball{\paramdim}{\parammode}{\shrinkingrad}} 
	\normaldens{\param}{\parammode}{[\logposthess(\parammode)]^{-1}}
	\multitaylorfunc(\param)^{(\quadnumadj + 1)}
	 \dee \param} \\
	&\leq \left(\frac{\hessbigmode}{\hesssmallmode}\right)^{\paramdim/2}
	(\derivboundmode)^{\quadnumadj+1}
	\sum_{\tvec \in \customtset{\quadnumadj+1}{3(\quadnumadj+1)}}
	\frac{\multicoeff{\quadnumadj+1}{\tvec}}{3(\quadnumadj+1)} 
	\paramdim^{2\quadnum\quadnumadj} (2\quadnum\quadnumadj)
	2^{\quadnum\quadnumadj} \, \Gamma\left(\frac{2\quadnum\quadnumadj+1}{2} \right) (\hesssmallmode \, n)^{-\tsum/2}.
\]
However, in the case that $2\quadnum = 2 \pmod 3$, $3(\quadnumadj +1 ) = 2\quadnum + 1$, and we write $\multitaylorfunc(\param)^{(\quadnumadj + 1)}$ as
\[\label{eq:term5split}
	&\hspace{-2em}\sum_{\tvec \in \customtsetequal{\quadnumadj+1}{3(\quadnumadj+1)}} \multicoeff{\quadnumadj+1}{\tvec} \left\{ \left(\prod_{s =3}^{2\quadnum-1} \left((\param - \parammode)_{[i_1\dots i_s]!} \partial^{i_1 \dots i_s  } \llandp(\parammode)\right)^{t_s} \right) \left((\param - \parammode)_{[i_1\dots i_{2\quadnum}]!} \partial^{i_1 \dots i_{2\quadnum}} \llandp(\parammodemid^\param)\right)^{t_{2\quadnum}} \right\}\\
	&\hspace{-1em}+ \sum_{\tvec \in \customtset{\quadnumadj+1}{3(\quadnumadj+1)+1}} \multicoeff{\quadnumadj+1}{\tvec} \left\{ \left(\prod_{s =3}^{2\quadnum-1} \left((\param - \parammode)_{[i_1\dots i_s]!} \partial^{i_1 \dots i_s  } \llandp(\parammode)\right)^{t_s} \right) \left((\param - \parammode)_{[i_1\dots i_{2\quadnum}]!} \partial^{i_1 \dots i_{2\quadnum}} \llandp(\parammodemid^\param)\right)^{t_{2\quadnum}} \right\}.
\]
For the first of these terms, note that since $2\quadnum = 3(\quadnumadj+1)-1 > \quadnumadj + 1$, $t_{2\quadnum} = 0$. Further, by the symmetry of the multivariate normal distribution and the fact that $3(\quadnumadj+1)$ is odd, 
\*[
	\int_{\ball{\paramdim}{\parammode}{\shrinkingrad}}  \normaldens{\param}{\parammode}{[\logposthess(\parammode)]^{-1}} \sum_{\tvec \in \customtsetequal{\quadnumadj+1}{3(\quadnumadj+1)}} \multicoeff{\quadnumadj+1}{\tvec} \left(\prod_{s =3}^{2\quadnum-1} \left((\param - \parammode)_{[i_1\dots i_s]!} \partial^{i_1 \dots i_s  } \llandp(\parammode)\right)^{t_s} \right)  \dee \param = 0.
\]

Thus, in this case the magnitude of Term 5 can be instead bounded by the second term in \cref{eq:term5split}, giving 
\*[
	&\hspace{-2em}\absbig{
	\int_{\ball{\paramdim}{\parammode}{\shrinkingrad}} 
	\normaldens{\param}{\parammode}{[\logposthess(\parammode)]^{-1}}
	\multitaylorfunc(\param)^{(\quadnumadj + 1)}
	 \dee \param} \\
	&\leq \left(\frac{\hessbigmode}{\hesssmallmode}\right)^{\paramdim/2}
	(\derivboundmode)^{\quadnumadj+1}
	\sum_{\tvec \in \customtset{\quadnumadj+1}{3(\quadnumadj+1) + 1}}
	\frac{\multicoeff{\quadnumadj+1}{\tvec}}{3(\quadnumadj+1) + 1} 
	\paramdim^{2\quadnum\quadnumadj} (2\quadnum\quadnumadj)
	2^{\quadnum\quadnumadj} \, \Gamma\left(\frac{2\quadnum\quadnumadj+1}{2} \right) (\hesssmallmode \, n)^{-\tsum/2},
\]
where we have once again applied the same argument as in \cref{sec:term4}.

\subsection{Bounding Term 6 of \cref{eqn:post-informal-all-terms}}\label{sec:term6}

By the same argument for bounding $\exp\{\abssmall{\multitaylorfunc(\param)}\}$ that we used in \cref{sec:term4}, for all
$\quadpointvec \in \quadpointset$ it holds that
\*[
	\absbig{\multitaylorfunc(\logpostchol \, \quadpointvec + \parammode)}
	&\leq 
	\derivboundmode
	\sum_{s=3}^{2\quadnum} 
	\sum_{i_1,\dots,i_s \in [\paramdim]}
	\frac{1}{s}
	\sum_{\iota = 1}^s \absbig{(\logpostchol \, \quadpointvec)_{i_\iota}}^s \\
	&\leq 
	\derivboundmode \,
	\paramdim^{2\quadnum}
	\sum_{s=3}^{2\quadnum} 
	(\hesssmallmode \, n)^{-s/2} \quadpointbound^s  \\
	&\leq  (2\quadnum) \, \paramdim^{2\quadnum} \max\{1, \quadpointbound^{2\quadnum}\} \derivboundmode \max\left\{(\hesssmallmode \, n)^{-3/2}, (\hesssmallmode \, n)^{-\quadnum}\right\}.
\]

Similarly, by the same argument for bounding $\abssmall{\multitaylorfunc(\param)}^{\quadnumadj+2}$ that we used in \cref{sec:term4},
\*[
	&\hspace{-2em}\sum\limits_{\quadpointvec\in\quadpointset} \weight(\quadpointvec) \normaldens{\quadpointvec}{0}{\identmat_\paramdim} 
	\absbig{\multitaylorfunc(\logpostchol \, \quadpointvec + \parammode)}^{(\quadnumadj + 2)} \\
	&\leq
	(\derivboundmode)^{\quadnumadj+2}
	\sum_{\tvec \in \customtset{\quadnumadj+2}{3(\quadnumadj+2)}} 
	\frac{\multicoeff{\quadnumadj+2}{\tvec}}{3(\quadnumadj+2)} \sum_{i_1,\dots,i_{\tsum} \in [\paramdim]} \sum_{\iota=1}^{\tsum}
	\sum\limits_{\quadpointvec\in\quadpointset} \weight(\quadpointvec) \normaldens{\quadpointvec}{0}{\identmat_\paramdim}
	\absbig{(\logpostchol \, \quadpointvec)_{i_\iota}}^\tsum \\
	&\leq
	(\derivboundmode)^{\quadnumadj+2}
	\sum_{\tvec \in \customtset{\quadnumadj+2}{3(\quadnumadj+2)}}
	\frac{\multicoeff{\quadnumadj+2}{\tvec}}{3(\quadnumadj+2)} 
	\paramdim^{2\quadnum\quadnumadj} (2\quadnum\quadnumadj)
	(\hesssmallmode \, n)^{-\tsum/2} \quadnum^\paramdim (2\pi)^{-\paramdim/2} \quadpointbound^{2\quadnum\quadnumadj},
\]
where the last step uses the bound of \cref{sec:term2}.

\subsection{Bounding Term 7 of \cref{eqn:post-informal-all-terms}}\label{sec:term7}
We handle this term using a similar logical argument to \cref{sec:term5}.
In the case that $2\quadnum \neq 2 \pmod 3$, this term can be treated in the same manner as \cref{sec:term6}, giving 
\*[ 
	&\hspace{-2em}\absbig{\sum\limits_{\quadpointvec\in\quadpointset} \weight(\quadpointvec) \normaldens{\quadpointvec}{0}{\identmat_\paramdim} 
	\frac{\multitaylorfunc(\logpostchol \, \quadpointvec + \parammode)^{(\quadnumadj + 1)}}{(\quadnumadj + 1)!}} \\
	&\leq  (\derivboundmode)^{\quadnumadj+1}
	\sum_{\tvec \in \customtset{\quadnumadj+1}{3(\quadnumadj+1)}}
	\frac{\multicoeff{\quadnumadj+1}{\tvec}}{3(\quadnumadj+1)} 
	\paramdim^{2\quadnum\quadnumadj} (2\quadnum\quadnumadj)
	(\hesssmallmode \, n)^{-\tsum/2} \quadnum^\paramdim (2\pi)^{-\paramdim/2} \quadpointbound^{2\quadnum\quadnumadj}.
\]

However, in the case that $2\quadnum = 2 \pmod 3$, we split $\multitaylorfunc(\logpostchol \, \quadpointvec + \parammode)^{(\quadnumadj + 1)}$ in the same manner as in \cref{eq:term5split}. 
Namely, since $\quadrule{\quadpointset}{\weight}$ is symmetric, $\normaldens{\quadpointvec}{0}{\identmat_\paramdim}$ is symmetric around zero, and $3(\quadnumadj+1)$ is odd,
\*[
\sum\limits_{\quadpointvec\in\quadpointset} \weight(\quadpointvec) \normaldens{\quadpointvec}{0}{\identmat_\paramdim}  \sum_{\tvec \in \customtsetequal{\quadnumadj+1}{3(\quadnumadj+1)}} \multicoeff{\quadnumadj+1}{\tvec} \left(\prod_{s =3}^{2\quadnum-1} \left((\logpostchol \, \quadpointvec)_{[i_1\dots i_s]!} \partial^{i_1 \dots i_s  } \llandp(\parammode)\right)^{t_s} \right) = 0.
\]
That is, in this case 
\*[
	&\hspace{-2em}\absbig{\sum\limits_{\quadpointvec\in\quadpointset} \weight(\quadpointvec) \normaldens{\quadpointvec}{0}{\identmat_\paramdim} 
	\frac{\multitaylorfunc(\logpostchol \, \quadpointvec + \parammode)^{(\quadnumadj + 1)}}{(\quadnumadj + 1)!}} \\
	&\leq  (\derivboundmode)^{\quadnumadj+1}
	\sum_{\tvec \in \customtset{\quadnumadj+1}{3(\quadnumadj+1)+1}}
	\frac{\multicoeff{\quadnumadj+1}{\tvec}}{3(\quadnumadj+1) + 1} 
	\paramdim^{2\quadnum\quadnumadj} (2\quadnum\quadnumadj)
	(\hesssmallmode \, n)^{-\tsum/2} \quadnum^\paramdim (2\pi)^{-\paramdim/2} \quadpointbound^{2\quadnum\quadnumadj}. 
\]

\subsection{Combining the Bounds on \cref{eqn:post-informal-all-terms}}

We now combine the results of \cref{sec:term1,sec:term2,sec:term3,sec:term4,sec:term5,sec:term6,sec:term7} and apply them to \cref{eqn:post-informal-all-terms}. 
We group the terms slightly more compactly than in \cref{eqn:post-informal-all-terms} to summarize this, grouping them by 
bounds on the summation or integration of high-degree polynomials,
bounds on the tails of low-degree polynomials, 
and bounds on the summation or integration of high-degree polynomials involving a remainder of the expansion of the exponential function. Consider,
\*[
	&\text{ Term 1 + Term 2 + Term 5 + Term 7} \\
	&=
	\left(\frac{\hessbigmode}{\hesssmallmode}\right)^{\paramdim/2} \sum_{j=1}^\quadnumadj (\derivboundmode)^{j} 
	\sum_{\tvec \in \bigtset{j}} \frac{\multicoeff{j}{\tvec} \,  }{2\quadnum} \paramdim^{2\quadnum\quadnumadj} (2\quadnum\quadnumadj)
	2^{\quadnum \quadnumadj} \, \Gamma\left(\frac{2\quadnum\quadnumadj+1}{2} \right) (\hesssmallmode \, n)^{-\tsum/2} \\
	&\qquad+ 
	\left(\frac{\hessbigmode}{\hesssmallmode}\right)^{\paramdim/2}
	(\derivboundmode)^{\quadnumadj+1}
	\sum_{\tvec \in \customtset{\quadnumadj+1}{3(\quadnumadj+1) + \ind\{2\quadnum = 2 \!\!\!\! \pmod 3\}}}
	\frac{\multicoeff{\quadnumadj+1}{\tvec}}{3(\quadnumadj +1)} 
	\paramdim^{2\quadnum\quadnumadj} (2\quadnum\quadnumadj)
	2^{\quadnum\quadnumadj} \, \Gamma\left(\frac{2\quadnum\quadnumadj+1}{2} \right) (\hesssmallmode \, n)^{-\tsum/2}\\
	&\qquad +
	\sum_{j=1}^\quadnumadj (\derivboundmode)^{j} 
	\sum_{\tvec \in \bigtset{j}} \frac{\multicoeff{j}{\tvec} \,  }{2\quadnum} 
	\paramdim^{2\quadnum\quadnumadj} (2\quadnum\quadnumadj)
	(\hesssmallmode \, n)^{-\tsum/2} \quadnum^\paramdim (2\pi)^{-\paramdim/2} \quadpointbound^{2\quadnum\quadnumadj}\\
	&\qquad+ 
	(\derivboundmode)^{\quadnumadj+1}
	\sum_{\tvec \in \customtset{\quadnumadj+1}{3(\quadnumadj+1) + \ind\{2\quadnum = 2 \!\!\!\! \pmod 3\}}}
	\frac{\multicoeff{\quadnumadj+1}{\tvec}}{3(\quadnumadj+1)} 
	\paramdim^{2\quadnum\quadnumadj} (2\quadnum\quadnumadj)
	(\hesssmallmode \, n)^{-\tsum/2} \quadnum^\paramdim (2\pi)^{-\paramdim/2} \quadpointbound^{2\quadnum\quadnumadj}, \]
which can be upper bounded by:
\[\label{eqn:term_high}
	&\hspace{-1em}\leq \paramdim^{2\quadnum\quadnumadj} \,\quadnumadj
	\left[
	\left(\frac{\hessbigmode}{\hesssmallmode}\right)^{\paramdim/2}2^{\quadnum \quadnumadj} \, \Gamma\left(\frac{2\quadnum\quadnumadj+1}{2} \right) 
	+
	\quadnum^\paramdim (2\pi)^{-\paramdim/2} \quadpointbound^{2\quadnum\quadnumadj}
	\right] 
	\\
	&\times
	\bigg[
	\sum_{j=1}^{\quadnumadj}
	(\derivboundmode)^{j} 
	\sum_{\tvec \in \bigtset{j}} \multicoeff{j}{\tvec}  \,
	(\hesssmallmode \, n)^{-\tsum/2}
	+
	(\derivboundmode)^{\quadnumadj + 1} 
	\sum_{\tvec \in \customtset{\quadnumadj+1}{3(\quadnumadj +1) + \ind\{2\quadnum = 2 \!\!\!\! \pmod 3\}} } \!\!\!\!\!\!\!\!\!\!\! \multicoeff{\quadnumadj+1}{\tvec}  \,
	(\hesssmallmode \, n)^{-\tsum/2}
	\bigg].	
\]

Additionally
\[\label{eqn:term_tails}
	\text{ Term 3}
	&= \left(\frac{\hessbigmode}{\hesssmallmode}\right)^{\paramdim/2}
	\sum_{j=0}^\quadnumadj (\derivboundmode)^{j} 
	\sum_{\tvec \in \smalltset{j}} \multicoeff{j}{\tvec}
	\paramdim^{2\quadnum\quadnumadj} (2\quadnum\quadnumadj)
	\frac{\Gamma\left(\frac{2\quadnum+1}{2}\right) (2e)^{\frac{2\quadnum+1}{2}} }{2\sqrt{2\pi \hesssmallmode}}
	n^{-\frac{\shrinkingconst^2 \hesssmallmode + 2}{4}} .
\]

Finally,
\*[
    &\text{ Term 4 + Term 6}\\
    &= 
	\left(\frac{\hessbigmode}{\hesssmallmode}\right)^{\paramdim/2}
	\exp\left\{(2\quadnum) \, \paramdim^{2\quadnum} \max\{1, \shrinkingconst^{2\quadnum}\} \derivboundmode \Big(\frac{\log(n)}{n} \Big)^{3/2} \right\} \\
	&\qquad\qquad \times (\derivboundmode)^{\quadnumadj+2}
	\sum_{\tvec \in \customtset{\quadnumadj+2}{3(\quadnumadj+2)}}
	\frac{\multicoeff{\quadnumadj+2}{\tvec}}{3(\quadnumadj +2)} 
	\paramdim^{2\quadnum\quadnumadj} (2\quadnum\quadnumadj)
	2^{\quadnum\quadnumadj} \, \Gamma\left(\frac{2\quadnum\quadnumadj+1}{2} \right) (\hesssmallmode \, n)^{-\tsum/2} \\
	&\qquad + 
	\exp\left\{(2\quadnum) \, \paramdim^{2\quadnum} \max\{1, \quadpointbound^{2\quadnum}\} \derivboundmode \max\left\{(\hesssmallmode \, n)^{-3/2}, (\hesssmallmode \, n)^{-\quadnum}\right\}\right\} \\
	&\qquad\qquad \times (\derivboundmode)^{\quadnumadj+2}
	\sum_{\tvec \in \customtset{\quadnumadj+2}{3(\quadnumadj+2)}}
	\frac{\multicoeff{\quadnumadj+2}{\tvec}}{3(\quadnumadj+2)} 
	\paramdim^{2\quadnum\quadnumadj} (2\quadnum\quadnumadj)
	(\hesssmallmode \, n)^{-\tsum/2} \quadnum^\paramdim (2\pi)^{-\paramdim/2} \quadpointbound^{2\quadnum\quadnumadj}, 
	\]
	which can be upper bounded by:
\[ \label{eqn:term_highexp}
	&\leq  \Bigg[
	\left(\frac{\hessbigmode}{\hesssmallmode}\right)^{\paramdim/2}2^{\quadnum \quadnumadj} \, \Gamma\left(\frac{2\quadnum\quadnumadj+1}{2} \right) 
	\exp\left\{(2\quadnum) \, \paramdim^{2\quadnum} \max\{1, \shrinkingconst^{2\quadnum}\} \derivboundmode \Big(\frac{\log(n)}{n} \Big)^{3/2} \right\} \\
	&\qquad+
	\quadnum^\paramdim (2\pi)^{-\paramdim/2} \quadpointbound^{2\quadnum\quadnumadj}
	\exp\left\{(2\quadnum) \, \paramdim^{2\quadnum} \max\{1, \quadpointbound^{2\quadnum}\} \derivboundmode \max\left\{(\hesssmallmode \, n)^{-3/2}, (\hesssmallmode \, n)^{-\quadnum}\right\}\right\}
	\Bigg] \\
	&\quad\times(\derivboundmode)^{\quadnumadj+2} 
	\sum_{\tvec \in \customtset{\quadnumadj+2}{3(\quadnumadj+2)}} \multicoeff{\quadnumadj+2}{\tvec}  \,
	\paramdim^{2\quadnum\quadnumadj} \,\quadnumadj
	(\hesssmallmode \, n)^{-\tsum/2}.
\]

Combining \cref{eqn:term_high,eqn:term_tails,eqn:term_highexp}, we obtain the final inequality

\*[
	&\Bigg\lvert
	\int_{\ball{\paramdim}{\parammode}{\shrinkingrad}}
	\normaldens{\param}{\parammode}{[\logposthess(\parammode)]^{-1}} 
	\exp\left\{\multitaylorfunc(\param) \right\}\dee\param
	- 
	\sum\limits_{\quadpointvec\in\quadpointset} \weight(\quadpointvec) \normaldens{\quadpointvec}{0}{\identmat_\paramdim} \exp\left\{\multitaylorfunc\Big(\logpostchol \, \quadpointvec + \parammode\Big)\right\}
	\Bigg \rvert \\
	&\leq
	\paramdim^{2\quadnum\quadnumadj} \,\quadnumadj
	\left[
	\left(\frac{\hessbigmode}{\hesssmallmode}\right)^{\paramdim/2}2^{\quadnum \quadnumadj} \, \Gamma\left(\frac{2\quadnum\quadnumadj+1}{2} \right) 
	+
	\quadnum^\paramdim (2\pi)^{-\paramdim/2} \quadpointbound^{2\quadnum\quadnumadj}
	\right] 
	\\
	&\qquad\times
	\left[
	\sum_{j=1}^{\quadnumadj}
	(\derivboundmode)^{j} 
	\sum_{\tvec \in \bigtset{j}} \multicoeff{j}{\tvec}  \,
	(\hesssmallmode \, n)^{-\tsum/2}
	+
	(\derivboundmode)^{\quadnumadj + 1} 
	\sum_{\tvec \in \customtset{\quadnumadj+1}{3(\quadnumadj +1) + \ind\{2\quadnum = 2 \!\!\!\! \pmod 3\}} } \multicoeff{\quadnumadj+1}{\tvec}  \,
	(\hesssmallmode \, n)^{-\tsum/2}
	\right]	
	\\
	&\quad +
	\left(\frac{\hessbigmode}{\hesssmallmode}\right)^{\paramdim/2}
	\sum_{j=0}^\quadnumadj (\derivboundmode)^{j} 
	\sum_{\tvec \in \smalltset{j}} \multicoeff{j}{\tvec}
	\paramdim^{2\quadnum\quadnumadj} (2\quadnum\quadnumadj)
	\frac{\Gamma\left(\frac{2\quadnum+1}{2}\right) (2e)^{\frac{2\quadnum+1}{2}} }{2\sqrt{2\pi \hesssmallmode}}
	n^{-\frac{\shrinkingconst^2 \hesssmallmode + 2}{4}} \\
	&\quad +
	\Bigg[
	\left(\frac{\hessbigmode}{\hesssmallmode}\right)^{\paramdim/2}2^{\quadnum \quadnumadj} \, \Gamma\left(\frac{2\quadnum\quadnumadj+1}{2} \right) 
	\exp\left\{(2\quadnum) \, \paramdim^{2\quadnum} \max\{1, \shrinkingconst^{2\quadnum}\} \derivboundmode \Big(\frac{\log(n)}{n} \Big)^{3/2} \right\} \\
	&\qquad\qquad+
	\quadnum^\paramdim (2\pi)^{-\paramdim/2} \quadpointbound^{2\quadnum\quadnumadj}
	\exp\left\{(2\quadnum) \, \paramdim^{2\quadnum} \max\{1, \quadpointbound^{2\quadnum}\} \derivboundmode \max\left\{(\hesssmallmode \, n)^{-3/2}, (\hesssmallmode \, n)^{-\quadnum}\right\}\right\}
	\Bigg] \\
	&\qquad\times(\derivboundmode)^{\quadnumadj+2} 
	\sum_{\tvec \in \customtset{\quadnumadj+2}{3(\quadnumadj+2)}} \multicoeff{\quadnumadj+2}{\tvec}  \,
	\paramdim^{2\quadnum\quadnumadj} \,\quadnumadj
	(\hesssmallmode \, n)^{-\tsum/2}
	.
\]

The statement of the lemma then follows by taking the worst-case choices of the sum indices for the constants, which all depend combinatorially on only $\paramdim$ and $\quadnum$. Additionally, the statement of the lemma includes the dependence on $\abssmall{\logpostchol} \leq [\eigen_{\paramdim}(\logposthess(\parammode))]^{-\paramdim/2} = (\hesssmallmode \, n)^{-\paramdim/2}$ that was dropped after \cref{eqn:main-lemma_polynomial-combo}.

\subsection{Laplace Approximation Proof ($\quadnum=1$)}\label{sec:laplace_proof}

When $\quadnum=1$ the argument follows nearly identically. The primary difference is that for \cref{eqn:initial-expansion} we use a fourth-order rather than second-order initial Taylor expansion. The proof then follows through as though $\quadnum=2$ for the expansions, although the terms of \cref{eqn:post-informal-all-terms} are slightly different. By definition, $\quadnumadj = 0$, so Terms 1 and 2 no longer appear. Further, the only valid $\tvec$ for Term 3 is all zeros since $j=0$, and the empty sum cancels with $1/\tsum$, leaving only the tail of a multivariate Normal. The bound used to control this in \cref{sec:term3} then still applies, recalling that while we have expanded as though $\quadnum=2$, we still actually have $\quadnum=1$. Terms 4, 5, 6, and 7 are treated in exactly the same way.

\manualendproof

%% file: sections/proof-A-summaries.tex
\section{Proofs of Convergence Rates for\jasa{\\} Approximate Posterior Summaries}\label{sec:helpers}

\subsection{Proofs for Exact Integration of Approximate Posterior}\label{sec:cor-app-proof}

As mentioned in \cref{sec:convergence}, in the idealized situation where one can exactly integrate the approximate posterior, the convergence rate of \cref{thm:mainresult} is preserved without additional assumptions. We now describe summary statistics of interest and prove this is the case. Of particular interest is credible set coverage and quantiles, since these require integration over a subset of the parameter space, and consequently the results of \cref{sec:convergence-summary} do not apply. For details on how these quantities are computed in practice (for which it remains open to prove convergence rates), see \cref{sec:computational}.

First, letting $\Borel{\paramdim}$ denote the Borel sets on $\Reals^\paramdim$, a \emph{credible function} is any function $\credfunc: \Reals^{n \times\datadim} \times [0,1] \to \Borel{\paramdim}$ such that for all datasets $\data$ and $\alpha\in[0,1]$,
\*[
	\int_{\credfunc(\data,\alpha)}\post{\param}\dee\param = \alpha.
\]
When $\data$ is clear, we denote $\credfunc(\data,\alpha)$ by $\credfunc_n(\alpha)$, and call the output an \emph{$\alpha$-credible set}.
Beyond generic credible sets, we are also interested in the \emph{marginal posterior quantiles}. 
For $j \in [\paramdim] = \{1,\dots,\paramdim\}$, the \emph{marginal posterior} evaluated at $\paramidx \in \Reals$ is
\*[
	\margpost{j}{\paramidx} = \int_{ \Reals^{\paramdim-1}} \post{\paramidx_{1}, \dots, \paramidx_{j -1}, \paramidx, \paramidx_{j+1}, \dots, \paramidx_{\paramdim}} \dee\param_{-j},
\]
and we denote its CDF by $\cdfpost{j}(x)$.
Further, the \emph{pseudo-inverse} of the marginal posterior CDF is defined in the usual way for $y \in [0,1]$ by
\*[
	\cdfpostinv{j}(y) = \inf \big\{x \in \Reals \setdelim \cdfpost{j}(x) \geq y \big\}.
\]
Then, for any $j \in [\paramdim]$ and $\alpha\in[0,1]$, the marginal posterior quantile is $\quant{j}{\alpha} = \cdfpostinv{j}(\alpha)$.
The posterior median $\quant{j}{0.5}$ is often used as a point estimate of $\paramtrueidx_{j}$, and the $95\%$ posterior credible interval $(\quant{j}{0.025},\quant{j}{0.975})$
is often used as a measure of uncertainty for this estimate. In addition to credible intervals, it is of interest to compute posterior moments, defined for any measurable function $\meanfunc:\paramspace \to \Reals$ by $\EE[\meanfunc(\param) \setdelim \data] = \int_{\paramspace} \meanfunc(\param) \post{\param} \dee\param$.

Using the accent $\check{\square}$ to denote exact integration of any approximate posterior $\approxpost{\param}$, an \emph{approximate credible function} $\intapproxcredfunc$ is analogous to a credible function that satisfies
\*[
	\int_{\intapproxcredfunc(\data,\alpha)}\approxpost{\param}\dee\param = \alpha.
\]
Similarly, 
\*[
	\intapproxmargpost{j}{\paramidx} = \int_{ \Reals^{\paramdim-1}} \approxpost{\paramidx_{1}, \dots, \paramidx_{j -1}, \paramidx, \paramidx_{j+1}, \dots, \paramidx_{\paramdim}} \dee\param_{-j},
\]
and $\cdfintapprox{j}(\cdot)$, $\cdfintapproxinv{j}(\cdot)$, and $\intapproxquant{j}{\cdot}$ are defined using $\intapproxmargpost{j}{\cdot}$.
Finally, define $\intapproxEE[\meanfunc(\param) \setdelim \data] = \int_{\paramspace} \meanfunc(\param) \approxpost{\param} \dee\param$.

\begin{restatable}{corollary}{Applications}
\label{cor:applications}
If $\approxpost{\param}$ satisfies the conditions of \cref{thm:mainresult},\\
{\upshape i)} 
For any $\intapproxcredfunc$,
\*[
	\lim_{n \to \infty}
	\datatrueprobn
	\bigg(
	\sup_{\alpha \in [0,1]}
	\absbig{\int_{\intapproxcredfunc_n(\alpha)}\post{\param}\dee\param - \alpha}
	\leq \constt \, n^{-\lfloor \frac{\quadnum+2}{3} \rfloor}
	\bigg)
	= 1.
\]
{\upshape ii)}
For all $\cdfderivbound>0$,
\*[
	\lim_{n \to \infty}
	\datatrueprobn
	\bigg(
	\sup_{j\in[\paramdim],\alpha\in[0,1]}\bigg\{\absbig{\intapproxquant{j}{\alpha} - \quant{j}{\alpha}} \sT \margpost{j}{\quant{j}{\alpha}} \geq \cdfderivbound\bigg\}
	\leq \frac{2 \constt}{\cdfderivbound} \, n^{-\lfloor \frac{\quadnum+2}{3} \rfloor}
	\bigg)
	= 1.
\]
{\upshape iii)}
For all measurable $\meanfunc:\paramspace \to \Reals$ with $\EE[\meanfunc(\param)\setdelim\data] < \infty$ a.s.,
\*[
	\lim_{n \to \infty}
	\datatrueprobn
	\bigg(
	\absbig{\frac{\EE[\meanfunc(\param) \setdelim \data]}{\intapproxEE[\meanfunc(\param) \setdelim \data]} - 1} \leq \constt \, \, n^{-\lfloor \frac{\quadnum+2}{3} \rfloor}
	\bigg)
	= 1. 
\]
\end{restatable}
\begin{remark}
\cref{cor:applications} ii) quantifies the notion that approximating quantiles in regions where the posterior cumulative distribution function is very flat is more difficult than regions where it is steep, which is seen in typical applications.
\end{remark}

\begin{proof}[Proof of \cref{cor:applications}]
{\upshape i)} 
For each $n$,
\*[
	&\hspace{-2em}
	\datatrueprobn
	\bigg(
	\sup_{\alpha \in [0,1]}
	\absbig{\int_{\intapproxcredfunc_n(\alpha)}\post{\param}\dee\param - \alpha}
	\leq \constt \, n^{-\lfloor \frac{\quadnum+2}{3} \rfloor}
	\bigg) \\
	&\geq 
	\datatrueprobn
	\bigg(
	\sup_{\alpha\in[0,1]} 
	\sup_{\paramdum\in\paramspace}
	\absbig{\frac{\post{\paramdum}}{\quadadaptapprox{\quadpointset}{\weight}{\paramdum}} \int_{\intapproxcredfunc_n(\alpha)} \quadadaptapprox{\quadpointset}{\weight}{\param} \dee \param - \alpha} \leq \constt \, n^{-\lfloor \frac{\quadnum+2}{3} \rfloor}
	\bigg) \\
	&=
	\datatrueprobn
	\bigg(
	\sup_{\alpha\in[0,1]} 
	\alpha
	\sup_{\paramdum\in\paramspace}
	\absbig{\frac{\post{\paramdum}}{\quadadaptapprox{\quadpointset}{\weight}{\paramdum}} - 1} \leq \constt \, n^{-\lfloor \frac{\quadnum+2}{3} \rfloor}
	\bigg) \\
	&=
	\datatrueprobn
	\bigg(
	\sup_{\paramdum\in\paramspace}
	\absbig{\frac{\post{\paramdum}}{\quadadaptapprox{\quadpointset}{\weight}{\paramdum}} - 1} \leq \constt \, n^{-\lfloor \frac{\quadnum+2}{3} \rfloor}
	\bigg).
\]
The results follows since the RHS tends to $1$ as $n \to \infty$ by applying \cref{lem:relative-helper} to \cref{cor:rateofothererrors}.\\

\noindent
{\upshape ii)} 
Fix $n$, and let $\boundofinterest = \constt \, n^{-\lfloor \frac{\quadnum+2}{3} \rfloor}$, where $\constt$ is the constant from \cref{thm:mainresult}. Suppose that $\norm{\post{\cdot} - \quadadaptapprox{\quadpointset}{\weight}{\cdot}}_{\mathrm{TV}} \leq \boundofinterest$. 
By smoothness conditions, $\frac{\dee}{\dee q}\cdfpost{j}(q) = \margpost{j}{q}$ for every $j \in [\paramdim]$.
Consider $j \in [\paramdim]$ and $\alpha\in[0,1]$ such that
\*[
	\margpost{j}{\quant{j}{\alpha}} \geq \cdfderivbound.
\]
This implies that
\*[
	\frac{\dee}{\dee y}\cdfpostinv{j}(y)\Big\rvert_{\alpha} \leq 1/\cdfderivbound,
\]
so
\*[
	\cdfpostinv{j}(\alpha) - \cdfpostinv{j}(\alpha - \boundofinterest)
	\leq \boundofinterest/\cdfderivbound.
\]

Choose an arbitrary $y$ (which we have just guaranteed exists) such that
$\cdfpostinv{j}(\alpha) - 2\boundofinterest/\cdfderivbound < y < \cdfpostinv{j}(\alpha - \boundofinterest)$.
This implies that $\cdfpost{j}(y) < \alpha - \boundofinterest$, which implies that $\cdfintapprox{j}(y) < \alpha$, and thus $\cdfintapproxinv{j}(\alpha) \geq y$. 
That is, $\cdfpostinv{j}(\alpha) - \cdfintapproxinv{j}(\alpha) \leq 2\boundofinterest/\cdfderivbound$.
The reverse direction follows by exactly the same logic, giving $\abssmall{\quant{j}{\alpha} - \intapproxquant{j}{\alpha}}
	= \abssmall{\cdfpostinv{j}(\alpha) - \cdfintapproxinv{j}(\alpha)}
	\leq 2 \boundofinterest/\cdfderivbound$.

Putting this together gives
\*[
	&\hspace{-2em}\lim_{n \to \infty}
	\datatrueprobn
	\bigg(
	\sup_{j\in[\paramdim],\alpha\in[0,1]}\bigg\{\absbig{\intapproxquant{j}{\alpha} - \quant{j}{\alpha}} \sT \margpost{j}{\quant{j}{\alpha}} \geq \cdfderivbound\bigg\}
	\leq 2 \boundofinterest/\cdfderivbound
	\bigg) \\
	&\geq
	\lim_{n \to \infty} \datatrueprobn\bigg(\norm{\post{\cdot} - \quadadaptapprox{\quadpointset}{\weight}{\cdot}}_{\mathrm{TV}} \leq \constt \, n^{-\lfloor \frac{\quadnum+2}{3} \rfloor}\bigg) \\
	&= 1.
\]

{\upshape iii)} 
Rearranging and then applying \cref{thm:mainresult} gives
\*[
	&\hspace{-2em}\lim_{n\to\infty}\datatrueprobn\bigg( \absbig{\intapproxEE[f(\param) \setdelim \data] - \EE[f(\param) \setdelim \data]} \leq \constt \, \EE[f(\param) \setdelim \data] \, n^{-\lfloor \frac{\quadnum+2}{3} \rfloor}\bigg) \\
	&= \lim_{n\to\infty}\datatrueprobn\left( \absbig{\int_{\paramspace} f(\param)\Big[\quadadaptapprox{\quadpointset}{\weight}{\param} - \post{\param}\Big] \dee\param} \leq \constt \, \EE[f(\param) \setdelim \data] \, n^{-\lfloor \frac{\quadnum+2}{3} \rfloor}\right) \\
	&= \lim_{n\to\infty}\datatrueprobn\left( \absbig{\int_{\paramspace} f(\param)\Big[\frac{\poststar{\param}}{\quadadaptmarg{\quadpointset}{\weight}} - \frac{\poststar{\param}}{\dist(\data)}\Big] \dee\param} \leq \constt \, \EE[f(\param) \setdelim \data] \, n^{-\lfloor \frac{\quadnum+2}{3} \rfloor}\right) \\
	&= \lim_{n\to\infty}\datatrueprobn\left( \absbig{\frac{1}{\quadadaptmarg{\quadpointset}{\weight}} - \frac{1}{\dist(\data)}} \, \dist(\data) \, \absbig{\int_{\paramspace} f(\param)\post{\param} \dee\param} \leq \constt \, \EE[f(\param) \setdelim \data] \, n^{-\lfloor \frac{\quadnum+2}{3} \rfloor}\right) \\
	&= \lim_{n\to\infty}\datatrueprobn\left( \absbig{\frac{\dist(\data)}{\quadadaptmarg{\quadpointset}{\weight}} - 1} \EE[f(\param) \setdelim \data] \leq \constt \, \EE[f(\param) \setdelim \data] \, n^{-\lfloor \frac{\quadnum+2}{3} \rfloor}\right) \\
	&= 1.
\]
The result then follows from \cref{lem:relative-helper}.

\end{proof}

\subsection{Proofs for Approximating Marginal Posterior Density}\label{sec:proof-pos-marg-compute}

Recall \cref{fact:pos-marg-compute}, which states that the convergence rate is preserved for marginal posterior approximations for values of $\margparamval$ that satisfy certain conditions.

\PosMargCompute*

We now provide sufficient conditions for \cref{fact:pos-marg-compute} to apply.
\begin{proposition}\label{fact:pos-marg-vals}
Letting $\parammode = (\margparammode, \nuisparammode)$ be the decomposition of the unrestricted mode, 
for all $\constt^\prime$ there exists $\constt$ such that under \cref{assn:kderiv,assn:hessian,assn:limsup,assn:consistency,assn:prior}
\*[
\lim_{n \to \infty}
	\datatrueprobn
	\bigg(
	\sup_{\margparamval \in \ball{\interestdim}{\margparammode}{\constt^\prime\, n^{-1/2}}} \absbig{\frac{\dist(\margparamval \setdelim \data)}{\approxdist(\margparamval \setdelim \data)} - 1} \leq \constt \, \, n^{-\lfloor \frac{\quadnum+2}{3} \rfloor}
	\bigg)
	= 1.
\]
\end{proposition}
\begin{remark}
\cref{fact:pos-marg-vals} states that the posterior marginal density can be accurately approximated in a $n^{-1/2}$-neighbourhood of the unrestricted posterior mode without any additional assumptions. By the proof of \cref{thm:mainresult}, the posterior is sufficiently small outside of this neighbourhood such that for large enough $n$ the marginal posterior density is well-approximated everywhere.
\end{remark}

\begin{proof}[Proof of \cref{fact:pos-marg-vals}]
For any $\margparamval$, let $\margdist{\margparamval}(\data) = \int \dist(\margparamval,\nuis,\data) \dee \nuis$ and
\*[
	\approxmargdist{\margparamval}(\data)
	= \abssmall{\logpostcholmarg} \sum_{\quadpointvec^\prime \in\quadpointset^\prime } \dist\left( (0, \logpostcholmarg \, \quadpointvec\Tr)\Tr  
	+ \parammodemarg, \data \right) \weight^\prime(\quadpointvec^\prime),
\]
using the quantities defined in \cref{sec:approximatebayesianinference}.
Recall that
\*[
	\dist(\margparamval \setdelim \data)
	= \frac{\margdist{\margparamval}(\data)}{\dist(\data)},
\]
so since the conditions of \cref{thm:mainresult} hold, it remains to show that
\*[
	\lim_{n \to \infty}
	\datatrueprobn
	\bigg(
	\sup_{\margparamval \in \ball{\interestdim}{\margparammode}{\constt^\prime\, n^{-1/2}}} \absbig{\frac{\margdist{\margparamval}(\data)}{\approxmargdist{\margparamval}(\data)} - 1} \leq \constt \, \, n^{-\lfloor \frac{\quadnum+2}{3} \rfloor}
	\bigg)
	= 1.	
\]
In particular, this amounts to a variant of \cref{thm:mainresult} that (a) applies to the constrained likelihood and (b) holds uniformly over a shrinking ball of values of $\margparamval$.

First, observe that for every $\margparamval$, \cref{LEM:NORMALIZING_CONSTANT_AS} will hold almost surely with the quantities adjusted appropriately to depend on $\approxmargdist{\margparamval}$, $\parammodemarg$, and $\logpostcholmarg$. 
We now focus on verifying that the other lemmas in the proof of \cref{thm:mainresult} can be appropriately applied 
in supremum over $\margparamval \in \ball{\interestdim}{\margparammode}{\constt^\prime\, n^{-1/2}}$. 

The key observation is that \cref{assn:kderiv,assn:hessian,assn:limsup,assn:prior} all hold uniformly in a fixed ball around $\paramtrue$, and \cref{assn:hessian,assn:limsup} imply that in limiting probability the unconstrained likelihood is strictly convex inside this ball and exponentially small outside of this ball respectively. By continuity of the likelihood, and hence continuity of $\parammodemarg$ as a function of $\margparamval$, these assumptions also imply
\*[
	\lim_{n\to\infty} \datatrueprobn\Big[ \{\parammodemarg: \margparamval\in \ball{\interestdim}{\margparammode}{\constt^\prime\, n^{-1/2}}\} \subseteq \ball{\paramdim}{\paramtrue}{\universalradius} \Big] = 1.
\]
This implies that the uniform analogues of \cref{assn:kderiv,assn:hessian,assn:limsup,assn:prior} all hold with respect to $\ball{\interestdim}{\margparammode}{\constt^\prime\, n^{-1/2}}$, and more specifically this implies the constants appearing in these bounds have no dependence on $\margparamval$. 

We now show that \cref{assn:consistency} holds with the constrained mode.
Following the second derivative notation and the proof of Lemma 1 in \citet{tang2020modified},
the derivative of $\parammodemarg$ is
\*[
\frac{\partial \parammodemarg}{\partial\margparamval} = [\logposthess^{\nuis\nuis}(\parammodemarg)]^{-1} \logposthess^{\margparam\nuis}(\parammodemarg),
\]
where
\*[
\logposthess(\param) = \begin{pmatrix}
\logposthess^{\margparam\margparam}(\param) & \logposthess^{\margparam\nuis}(\param) \\
\logposthess^{\nuis\margparam}(\param) & \logposthess^{\nuis\nuis}(\param)
\end{pmatrix}.
\]
By the uniform analogues of \cref{assn:kderiv,assn:hessian} on  $\ball{\interestdim}{\margparammode}{\constt^\prime\, n^{-1/2}}$, we have that the $\parammodemarg$ is a Lipschitz function in this ball when viewed as a function of $\margparamval$, which then implies the uniform version of \cref{assn:consistency}. Thus, the random coefficients appearing in the constrained variant of \cref{lem:normalizing_constant_as} can then be upper bounded uniformly using \cref{lem:constant_convergence}, since the uniform variants of \cref{assn:kderiv,assn:hessian,assn:limsup,assn:consistency,assn:prior}  have no dependence on the value of $\margparamval$.

To finish the proof, it remains to observe that the uniform analogues of \cref{lem:fixed_numer,lem:shrinking_numer} follow from the uniform analogues of the assumptions, and hence the uniform variants of \cref{lem:normalizing_constant_approx,lem:lowerbound_quad} hold.
\end{proof}

\subsection{Proofs for Approximating Marginal Posterior Expectation}\label{sec:proof-pos-mean-compute}

We require the following assumptions on $\meanfunc$ in order to prove convergence rates for marginal poster expectations computed using further applications of \AGHQ{}.

\begin{itemize}
\item[\upshape (M1)] There exists $\derivbound>0$ such that for all $\dumderivvec \subseteq \N^\paramdim$ with $0 \leq \abssmall{\dumderivvec} \leq \derivnum $ 
\*[
	\sup_{\param \in \ball{\paramdim}{\paramtrue}{\universalradius}} \abssmall{\partial^{\dumderivvec}\log\meanfunc(\param)} < \derivbound.
\]
\item[\upshape (M2)] There exists $- \infty < \hesssmallmean \leq \hessbigmean < \infty$ such that 
\*[
	\hesssmallmean \leq \inf_{\param \in \ball{\paramdim}{\paramtrue}{\universalradius}}\eigen_{\paramdim}(\partial^2[-\log \meanfunc(\param)]) \leq \sup_{\param \in \ball{\paramdim}{\paramtrue}{\universalradius}} \eigen_{1}(\partial^2[-\log \meanfunc(\param)])  \leq \hessbigmean.
\] 
\item[\upshape (M3)] There exist  $0 < \priorsmall < \priorbig < \infty$ such that
\*[
   \priorsmall \leq \inf_{\param \in \ball{\paramdim}{\paramtrue}{\universalradius}} \meanfunc(\param) \leq \sup_{\param \in \ball{\paramdim}{\paramtrue}{\universalradius}} \meanfunc(\param) \leq \priorbig.
\]
\end{itemize}

We restate \cref{fact:pos-mean-compute} here for convenience.

\PosMeanCompute*

\begin{proof}[Proof of \cref{fact:pos-mean-compute}]
We want to verify that the assumptions of \cref{sec:assumptions} hold for the new ``prior'' $\distfunc(\param)$. Under (M1), \cref{assn:kderiv} holds since by product rule,
\*[
	\absbig{\partial^{\dumderivvec}\log\distfunc(\data,\param)}
	&= \frac{\abssmall{\partial^{\dumderivvec}\log\distfunc(\data,\param)}}{\distfunc(\data,\param)} \\
	&= \frac{\abssmall{\partial^{\dumderivvec}\dist(\data,\param)}\meanfunc(\param) + \dist(\data,\param)\abssmall{\partial^{\dumderivvec}\meanfunc(\param)}}{\dist(\data,\param)\meanfunc(\param)} \\
	&= \abssmall{\partial^{\dumderivvec}\log\dist(\data,\param)} + \abssmall{\partial^{\dumderivvec}\log\meanfunc(\param)}.
\]
Under (M2), it holds almost surely that
\*[
	\hesssmallmean + \inf_{\param \in \ball{\paramdim}{\paramtrue}{\universalradius}}\eigen_{\paramdim}(\logposthess(\param))
	&\leq \inf_{\param \in \ball{\paramdim}{\paramtrue}{\universalradius}}\eigen_{\paramdim}(\partial^2[-\log \meanfunc(\param)]) + \inf_{\param \in \ball{\paramdim}{\paramtrue}{\universalradius}}\eigen_{\paramdim}(\logposthess(\param)) \\
	&\leq \inf_{\param \in \ball{\paramdim}{\paramtrue}{\universalradius}}\eigen_{\paramdim}(\logposthessfunc(\param)) \\
	&\leq \sup_{\param \in \ball{\paramdim}{\paramtrue}{\universalradius}}\eigen_{1}(\logposthessfunc(\param)) \\
	&\leq \sup_{\param \in \ball{\paramdim}{\paramtrue}{\universalradius}}\eigen_{1}(\partial^2[-\log \meanfunc(\param)]) + \sup_{\param \in \ball{\paramdim}{\paramtrue}{\universalradius}}\eigen_{1}(\logposthess(\param)) \\
	&\leq \hessbigmean + \sup_{\param \in \ball{\paramdim}{\paramtrue}{\universalradius}}\eigen_{1}(\logposthess(\param)).
\]
Thus, if \cref{assn:hessian} holds for the prior $\dist$, we also have 
\*[
	\lim_{n \to \infty} \datatrueprobn\Big[ n \hesssmall / 2 \leq \inf_{\param \in \ball{\paramdim}{\paramtrue}{\universalradius}}\eigen_{\paramdim}(\logposthessfunc(\param)) \leq \sup_{\param \in \ball{\paramdim}{\paramtrue}{\universalradius}} \eigen_{1}(\logposthessfunc(\param))  \leq 2 n \hessbig\Big] = 1,
\]
since for large enough $n$, $\hesssmallmean > - n \hesssmall / 2$ and $\hessbigmean < n \hessbig$. 
In particular, \cref{assn:hessian} holds for $\distfunc$ with the constants $\hesssmall / 2$ and $2 \hessbig$, and so $\logposthessfunc$ is locally positive definite at $\paramtrue$ (and hence $\logpostcholfunc$ exists) with probability tending to one.

\cref{assn:limsup} is implied by (M3) and the positivity of $\meanfunc$.
Finally, \cref{assn:consistency} is implied by (M3) and the usual consistency argument, while \cref{assn:prior} is trivially implied by (M3).

Having verified these conditions, we also note that for $x \in (0,1/2)$,
\*[
	\frac{1+x}{1-x} \leq 1+4x \quad
	\text{ and } \quad
	\frac{1-x}{1+x} \geq 1-2x.
\]
Thus, the statement follows from the following three facts:
\*[
	\lim_{n\to\infty}\datatrueprobn\left( \absbig{\frac{\int \distfunc(\param,\data)\dee\param}{\abssmall{\logpostcholfunc} \sum_{\quadpointvec\in\quadpointset} \distfunc(\logpostcholfunc \, \quadpointvec + \parammodefunc, \data) \weight(\quadpointvec)} - 1} \leq \constt \, n^{-\lfloor \frac{\quadnum+2}{3} \rfloor}\right) = 1,
\]
\*[
	\lim_{n\to\infty}\datatrueprobn\left( \absbig{\frac{\int \dist(\param,\data)\dee\param}{\abssmall{\logpostchol} \sum_{\quadpointvec\in\quadpointset} \dist(\logpostchol \, \quadpointvec + \parammode, \data) \weight(\quadpointvec)} - 1} \leq \constt \, n^{-\lfloor \frac{\quadnum+2}{3} \rfloor}\right) = 1,
\]
and
\*[
	\EE[\meanfunc(\param)\setdelim\data]
	= \frac{\int \distfunc(\param,\data)\dee\param}{\int \dist(\param,\data)\dee\param},
\]
as $0 < \constt \, n^{-\lfloor \frac{\quadnum+2}{3} \rfloor} < 1/2$ eventually for large $n$, leading to the desired limiting statement.
\end{proof}

The most relevant application of \cref{fact:pos-mean-compute} is to compute the marginal posterior moments, which we now show satisfies the conditions of \cref{fact:pos-mean-compute}. 
\begin{proposition}\label{fact:pos-mean-compute-verify}
For every $j\in[\paramdim]$ and $i \in \Nats$, if $\meanfunc(\param) = (\paramidx_j)^i$,
Assumptions (M1) through (M3) from \cref{sec:proof-pos-mean-compute} are satisfied for
$\meanfunc^{+}$ when $\paramtrueidx_j > 0$ and $\meanfunc^{-}$ when $\paramtrueidx_j < 0$.
\end{proposition}
\begin{remark}
For odd moments, when the parameter is negative the integral cannot be approximated using the techniques of this paper, but by posterior concentration this contribution to the integral is tending to zero and can be discarded (see \cref{app:computesummaries} for computational details).
\end{remark}

\begin{proof}[Proof of \cref{fact:pos-mean-compute-verify}]
Without loss of generality, suppose $\paramtrueidx_j > 0$; when $\paramtrueidx_j < 0$ the argument is identical swapping $\meanfunc^{+}$ and $\meanfunc^{-}$. Thus, there exists $\universalradius>0$ small enough such that for all $\param \in \ball{\paramdim}{\paramtrue}{\universalradius}$, $\paramidx_j \in (a_1, a_2)$ for some $0 < a_1 < a_2$. Since $i$ will only change the scaling of $\log \meanfunc(\param)$, it suffices to verify (M1) through (M3) for $i=1$.

For any $\param \in \ball{\paramdim}{\paramtrue}{\universalradius}$ and $\dumderiv_j \in [\derivnum]$, $\dumderivvec = (0,\dots,\dumderiv_j,0,\dots,0)$ satisfies
\*[
	\partial^{\dumderivvec}[-\log\meanfunc^{+}(\param)]
	=
	(-1)^{\dumderiv_j} (\dumderiv_j-1)! \paramidx_j^{-\dumderiv_j},
\]
and for every $\dumderivvec \in \Nats^\paramdim$ not of this form, $\partial^{\dumderivvec}\log\meanfunc^{+}(\param) = 0$.
Thus, by boundedness of $\paramidx_j$, (M1) holds for $\meanfunc^{+}$.

For (M2), the Hessian satisfies $\partial^2[-\log\meanfunc(\param)]_{jj} = \paramidx_j^{-2}$ and $\partial^2[-\log\meanfunc(\param)]_{ij} = 0$ otherwise. Thus, by boundedness of $\paramidx_j$ in $\ball{\paramdim}{\paramtrue}{\universalradius}$, the eigenvalues are all nonnegative and bounded above as required.

(M3) holds trivially since $\paramidx_j \in (a_1, a_2)$.

\end{proof}

%% file: sections/computational.tex
\section{Computational Considerations}\label{sec:computational}

In this section we describe the necessary computational and implementation details for applying \AGHQ{} to models of the type we consider in \cref{sec:exampleslowdim,sec:highdimexamples}.

\subsection{The \texttt{aghq} Package}\label{subsec:aghqpackage}

All of the computations described in this paper
are implemented in the \texttt{R} package \texttt{aghq}, current version \texttt{0.4.0}, on \texttt{CRAN}. 
The user only needs to provide an unnormalized log-posterior and two derivatives (which can be obtain automatically, see \cref{subsec:autodiff}).
From this input the \texttt{aghq} package performs all subsequent computations automatically, including: optimization and approximate normalization; approximate moments; and marginal densities, distribution functions, and quantiles. This section gives the details on how the package performs these computations automatically without requiring any additional user input.

The \AGHQ{} procedure employed in \cref{sec:exampleslowdim} (i.e., for low-dimensional models) is implemented as follows. The user provides a list \texttt{ff} containing the following elements, each of which are functions of $\param$,
\*[
	&\texttt{fn}: \log\poststar{\param},\\
	&\texttt{gr}: \partial_{\param}\log\poststar{\param},\\
	&\texttt{he}: \partial^{2}_{\param}\log\poststar{\param}.
\]
In all of our examples, \texttt{ff} is obtained via a call to \texttt{TMB::MakeADFun} (see \cref{subsec:autodiff}), and the user therefore only has to construct a \texttt{TMB} template implementing $\log\poststar{\param}$. This construction is problem-specific.

Given \texttt{ff}, a \texttt{numeric} number of (one-dimensional) quadrature points \texttt{k}, and a \texttt{numeric} vector of length $\paramdim=\text{dim}(\param)$ of starting values for the optimization \texttt{start}, the command 

\*[
	\texttt{quad <- aghq(ff,k,start)}
\]
calculates $\log\GHmarg$ using product GHQ as the base grid (other grids satisfying $\property{\quadnum}{\paramdim}$ are also supported). The command \texttt{get_log_normconst(quad)} returns the $\log\GHmarg$ prescribed by \cref{THM:MAINRESULT}.

The object \texttt{quad} has class \texttt{aghq}, and the commands
\*[
	&\texttt{summary(quad)}\\
	&\texttt{plot(aghq)}
\]
will compute and print or plot univariate marginal densities according to \cref{eqn:approx-marg-dist-defn} and approximate moments according to \cref{eqn:approx-moment-defn},
which are exactly the quantities for which the theoretical guarantees of \cref{fact:pos-mean-compute,fact:pos-marg-compute} apply. 
Also computed are approximate quantiles and cumulative distribution functions, which are not covered by the theoretical guarantees of the present work.

As of version \texttt{0.4.0}, the user must set 
\begin{center}
\texttt{control = default_control(method_summaries=`correct')}
\end{center}
to turn on the computation of moments and marginals according to \cref{eqn:approx-moment-defn,eqn:approx-marg-dist-defn}. This was done for backwards compatibility, and the correct computation will be made the default setting in the eventual \texttt{1.0.0} version release of \texttt{aghq}.

The use of parameter transformations is ubiquitous and convenient in Bayesian models. While \cref{THM:MAINRESULT} does not require any transformation to be made (only the assumptions of \cref{sec:assumptions} to hold), \cite{nayloradaptive} point out that often a simple transformation, like \texttt{log} or \texttt{logit}, can yield a transformed parameter whose posterior is closer to being log-quadratic than that of the parameter of inferential interest, and that this can improve the finite-sample accuracy of the quadrature and/or the speed and stability of the optimization. 

The \texttt{aghq} package provides an interface for parameter transformations. Suppose inferential interest is in parameter $\phi$, but the user implements \texttt{ff} to depend on a transformed parameter $\theta = h(\phi)$ where $h:\Reals\to\Reals$ is monotonic and invertable. It is desirable for the quadrature to be done on the $\theta$ scale, but all summary methods to return results for $\phi = h^{-1}(\theta)$. The user creates a transformation object of class \texttt{aghqtrans} using the command
\*[
	\texttt{trans <- make_transformation(totheta = h,fromtheta = hinv)},
\]
where $\texttt{h} = h$ and $\texttt{hinv} = h^{-1}$. These functions are passed through \texttt{match.fun} internally. 

The quadrature
\*[
	\texttt{quad <- aghq(ff,k,start,transformation = trans)}
\]
is then performed in exactly the same way, but the \texttt{summary} and \texttt{plot} commands will now return inferences for $\phi = h^{-1}(\theta)$. When $\paramdim>1$, $h$ is interpreted as a vectorized scalar-to-scalar function; fully multivariate transformations are not yet supported.

As a concrete example, suppose \texttt{ff} is a template implementing $\log\dist(\param,\data)$ for the infectious disease model of \cref{sec:exampleslowdim}. Recall the parameters of interest are $(\alpha,\beta)$, but the quadrature was done on the posterior of the transformed parameters $\theta_{1}=\log\alpha$ and $\theta_{2}=\log\beta$. The full code to implement one instance of this example is
\*[
	&\texttt{quad <- aghq(ff,7,c(0,0),}\\
	&\qquad\texttt{make_transformation(`log',`exp'),}\\
	&\qquad\texttt{control = default_control(method_summaries=`correct'))}\\
	&\texttt{summary(quad); plot(quad)}.
\]

In \cref{sec:highdimexamples}, we describe the use of \AGHQ{} within a more complicated framework for making approximate Bayesian inferences. This full framework is also implemented within the \texttt{aghq} package. The simplest way, which we describe here, is for the user to implement a \texttt{TMB} template computing $-\log\dist(\highdimparam,\param,\data)$ and set $\texttt{random} = \highdimparam$. This provides a list \texttt{ff} containing elements (again functions of $\param$):
\*[
	&\texttt{fn}: -\log\LAapprox(\param,\data),\\
	&\texttt{gr}: -\partial_{\param}\log\LAapprox(\param,\data).
\]
The requirement to implement the \emph{negative} log-posterior is for compatibility with \texttt{TMB} and its automatic Laplace approximation. This is handled internally by \texttt{aghq}.

The command
\*[
	\texttt{quad <- marginal_laplace_tmb(ff,k,start)}
\]
performs the computations necessary to use $\approxdist(\highdimparam \setdelim \data)$ (\cref{eqn:Wapprox}). The \texttt{summary} and \texttt{plot} methods provide inferences for $\param$ based on the AGHQ-normalized marginal Laplace approximation $\log\LAapprox(\param|\data)$, providing an implementation of the method of \citet{laplace}. The user obtains \texttt{M} samples from the mixture of Gaussians $\approxdist(\highdimparam \setdelim \data)$ using the command
\*[
	\texttt{sample_marginal(quad,M)}.
\]
This can be done automatically within \texttt{summary}, by setting the \texttt{max_print} option to be greater than $\text{dim}(\highdimparam)$, in which case \texttt{summary} will compute and return sample-based summary statistics of $\highdimparam$. By default, \texttt{max_print} is set to $30$, and summaries are computed using $1000$ samples, but these can be easily changed by the user.

\subsection{Computing Posterior Summaries}\label{app:computesummaries}

The \texttt{marginal_posterior}, \texttt{compute_moment}, \texttt{compute_pdf_and_cdf}, and \texttt{compute_quantiles} functions are all automatically called within \texttt{aghq} and the \texttt{summary} and \texttt{plot} methods for \texttt{aghq} objects, and automatically handle any parameter transformations provided using \texttt{make_transformation}. They are also exported directly so that the user has further control over the computation of summary statistics.

The \texttt{compute_moment} function computes the approximate moment $\approxEE[\meanfunc(\param) \setdelim \data]$ of a function $g:\R^{\paramdim}\to\R^{+}$ according to \cref{eqn:approx-moment-defn}. The user provides a list \texttt{gg} containing the following elements, each of which are functions of $\param$,
\*[
	&\texttt{fn}: \log g(\param),\\
	&\texttt{gr}: \partial_{\param}\log g(\param),\\
	&\texttt{he}: \partial^{2}_{\param}\log g(\param).
\]
The \texttt{make_moment_function} helper helps to automate this process. The user calls 
\begin{center}
\texttt{make_moment_function(g)}, 
\end{center}
where $\texttt{g} = g$, and \texttt{make_moment_function} creates the appropriate list, using numeric derivatives. Other, more detailed options are described in the package documentation.

The \texttt{marginal_posterior} function computes the approximate marginal posterior $\approxdist(\margparamval \setdelim \data)$ at any point $\texttt{q}\in\R$ according to \cref{eqn:approx-marg-dist-defn}. If unspecified by the user (the default), the evaluation points \texttt{q} are chosen automatically, using a default based on a one-dimensional adapted GHQ rule. 

For computing raw and central moments, the user may instead pass a numeric scalar \texttt{nn} to \texttt{compute_moment}, as well as \texttt{type = `raw'} or \texttt{type = `central'}. In this case, \texttt{compute_moment} automatically constructs an appropriate input list for the function $g(\param) = \paramidx_{j}^{\texttt{nn}}$ (\texttt{type=`raw'}) or $g(\param) = (\paramidx_{j} - \approxEE(\paramidx_{j}|\data))^{\texttt{nn}}$ (\texttt{type=`central'}), for $j\in[\paramdim]$, and returns the corresponding vector of approximate moments. To ensure positivity (which is required both theoretically for \cref{fact:pos-mean-compute} and computationally for \AGHQ{} to be applicable), the function automatically detects whether 
$\min_{\quadpointvec\in\quadpointset}(\logpostchol\quadpointvec + \parammode)_j^{\texttt{nn}} < 0$, adds a buffer value $a > -\min_{\quadpointvec\in\quadpointset}(\logpostchol\quadpointvec + \parammode)_j^{\texttt{nn}}$, and outputs $\approxEE(\paramidx_j^{\texttt{nn}} - a|\data) + a$. 

Data suitable for creating plots of the approximate probability density and cumulative distribution functions are computed using \texttt{compute_pdf_and_cdf}. Unlike \texttt{marginal_posterior} and \texttt{compute_moment}, the output of these functions are not covered by \cref{fact:pos-mean-compute,fact:pos-marg-compute}. For any $\psi\in\R$, denote the Lagrange polynomial interpolant of $\log\GHapproxidx{j}{\psi}$ by $\PolyP_{j}(\psi)$ and define $$\GHapproxidxpoly{j}{\psi} = \exp\left\{ \PolyP_{j}(\psi) \right\}.$$The marginal CDF is defined by
\*[
	\postcdfidx{j}{\psi}
	= \int_{-\infty}^{\psi} \postidx{j}{\psi^{\prime}} \dee \psi^{\prime},
\]
and approximated by choosing a fine grid $x_{1},\ldots,x_{L}$ for some large $L\in\N$ and computing
\*[
	\GHapproxcdfidx{j}{\psi}
	= \sum_{l: x_l \leq \psi} \GHapproxidxpoly{j}{x_l}(x_{l+1}- x_l).
\]
The choice of grid is again handled internally by \texttt{compute_pdf_and_cdf}, with no input required by the user.

Finally, marginal quantiles are computed by \texttt{compute_quantiles}. For any level $\alpha \in (0,1)$, \texttt{compute_quantile} outputs:
\*[
	\approxquantpoly{j}{\alpha}
	= \min\left\{x \in \{x_1,\dots,x_L\} \Bigsetdelim \GHapproxcdfidx{j}{x} \geq \alpha \right\}.
\]
We reiterate that all of the quantities described in this section are computed and displayed to the user automatically by \texttt{summary.aghq}.

\subsection{Software Package Versions}

For \AGHQ{}, we use \texttt{CRAN} version 0.4.0 of the \texttt{aghq} package, 
which may be installed using the command \texttt{install.packages(`aghq')}.
For \MCMC{}, we use the \texttt{tmbstan} package \citep{tmbstan}, version 1.0.2 from \texttt{CRAN}, which implements the state-of-the-art No-U-Turn sampler \citep{nuts}, the self-tuning version of Hamiltonian Monte Carlo that is the default in the popular \texttt{STAN} language \citep{stan}. 

\subsection{Automatic Differentiation}\label{subsec:autodiff}

Approximate computation of $\GHmarg$ requires two derivatives of $\log\dist(\param,\data)$, and this is often too burdensome to be done by hand or numerically. Automatic Differentiation (AD) \citep{tmb,stan,bbsvi,hybrid,torch,pyhessian} provides exact derivatives of any differentiable function that can be represented by a computer. Other prominant methods for Bayesian inference, including Hamiltonian Monte Carlo \citep{stan} and Stochastic Variational Inference \citep{bbsvi}, also require differentiation of complicated (and in the latter case, intractable) objective functions, and the cited implementations of these methods use AD for this purpose. In our examples, we use \texttt{TMB} \citep{tmb}, but as described in \cref{subsec:aghqpackage}, any manner by which the derivatives are obtained is compatible with the \texttt{aghq} package. Because of the wide availability of AD software, including in \texttt{R} (see \citealt{stan,torch}), the requirement of two derivatives of $\log\dist(\param,\data)$ is computationally benign.

\subsection{Optimization Software}

Computing $\GHmarg$ requires computing $\parammode = \argmax_\param\log\dist(\param,\data)$, and this requires numerical optimization. By \cref{assn:hessian}, $\log\dist(\param,\data)$  is locally convex for sufficiently large $n$, and we therefore use convex optimization techniques. When $\paramspace = \R^{\paramdim}$, we use trust region optimization as implemented in the \texttt{trustOptim} \citep{trustoptim} or \texttt{trust} \citep{trustdense} packages. Box parameter constraints can be removed via parameter transformations (\cref{subsec:infectiousdisease,subsec:astro,subsec:aghqpackage}) or handled using more advanced optimization tools. General constraints, including box and non-linear constraints, are accomodated by using the \texttt{IPOPT} package \citep{ipopt} for constrained optimization (\cref{subsec:astro}).

A referee pointed out that the concentration behaviour of the log-likelihood implied by \cref{assn:kderiv,assn:hessian,assn:limsup,assn:consistency,assn:prior} may make optimization challenging for large $n$. If any such difficulty is encountered, we recommend dividing the objective function by $n$ when computing $\parammode$. By \cref{assn:kderiv}, the log-likelihood is bounded in probability when scaled by $n$, and this precise knowledge of the scaling behaviour of the objective function is a useful feature of models satisfying these assumptions.
However, such adjustments were not necessary to obtain stable results in the optimization step for any of the examples we considered.

%% file: sections/simulations.tex
\section{Simulations}\label{sec:simulations}

The results of \cref{sec:convergence} provide guarantees on the accuracy of approximating posterior distributions and posterior summary statistics using \AGHQ.
Naturally, such results require certain assumptions (see \cref{sec:assumptions}) about the model, and are all statements about guarantees as the sample size tends to infinity. 
Additionally, all of these guarantees are upper bounds, and individual models may or may not achieve faster rates of convergence. 
\cref{thm:mainresult} cannot be tight in all cases since, for example, if the posterior is a normal distribution then \AGHQ{} will exactly approximate the density, resulting in zero error.
However, we conjecture that for many models, \AGHQ's dependence on $n$ and $\quadnum$ is no better than in our upper bounds.

In the absence of theoretical lower bounds, we use simulation to show an example of a simple model in which the empirical error rate is not lower than our prescribed upper bound.
For this simple model, we observe that the convergence rate given in \cref{thm:mainresult} is realized at very small sample sizes ($n \ll 100$), empirically demonstrating the tightness of the upper bound. 
Further, in contrast to the existing literature, our results are stochastic in nature. We have designed our simulation to demonstrate this, as it is important that not only a single ideal realization of data achieves the desired rate, but that such datasets occur with high probability under the model.

We consider the following simple model:
\begin{equation}\begin{aligned}\label{eqn:simmodel}
\dataidx_i \setdelim \lambda &\overset{ind}{\sim} \text{Poisson}(\lambda), i\in[n], \\
\lambda &\sim \text{Exponential}(1),
\end{aligned}\end{equation}
with posterior
\*[
	\lambda \setdelim \data &\sim \text{Gamma}\Big(1 + \sum_{i=1}^{n}\dataidx_{i},n+1\Big).
\]
We have chosen this conjugate model because the posterior and normalizing constant are known exactly, facilitating computation of error rates. In contrast, \cite{adaptive_GH_2020} use an example in which their integral is not known exactly, and use \AGHQ{} with a large number of quadrature points in place of the exact answer. Consequently, their simulation confirms that the variance of \AGHQ{} diminishes with the number of quadrature points, but does not demonstrate anything about its bias. 

If $\relerror(\data) \approx C n^{- \lfloor \frac{\quadnum+2}{3} \rfloor}$ for some constant $C$, then $\log \relerror(\data) \approx \log(C) - \lfloor \frac{\quadnum+2}{3} \rfloor \log n $. Therefore, we compute $\log \relerror(\data) + \lfloor \frac{\quadnum+2}{3} \rfloor \log n$ for many simulated datasets and various values of $n$ and $\quadnum$, and observe that there is no pattern in $n$ in the resulting plots.
 The full details of this simulation procedure are described in Algorithm \ref{alg:simulation2}. 

Figure \ref{fig:relraterandom} demonstrates the results over $1000$ simulations with $\lambda = 5$. We consider $n$ up to $100$, and $\quadnum \in \{3,5,7,11\}$, which correspond respectively to rates of $\Ordp(n^{-1}),\Ordp(n^{-2}),\Ordp(n^{-3})$, and $\Ordp(n^{-4})$ by \cref{thm:mainresult}.
Each point represents a realization of $\log \relerror(\data) + \lfloor \frac{\quadnum+2}{3} \rfloor \log n$ for a dataset simulated from this model, which will equal $\log \constt$ for some constant $\constt > 0$ if the rate of \cref{thm:mainresult} is achieved.
The exact value of the vertical axis is of only secondary interest; the relevant observation is that there is no visible pattern with respect to $n$.

We observe no pattern in \cref{fig:relraterandom}, implying that the rate of \cref{thm:mainresult} is tight for this simple model. Computations are all done on the log scale for numerical stability, however due to the simplicity of the model we observed roundoff error when computing $\log\relerror(\data)$ for some of the simulated datasets when $k = 11$ and $n > 80$, an artifact that appears in the lower right corner of \cref{fig:relraterandom} (d). This is due to roundoff error when computing $\log \relerror(\data)$, and is not related to the properties of the \AGHQ{} procedure.

\begin{algorithm}[p]
\textbf{Input}: max sample size $N\in\N$, quadrature points parameter $k\in\N$, number of simulations $M\in\N$, mean response $\lambda\in\R$.\\
Let $r_k = \lfloor\frac{k+2}{3} \rfloor$.\\
\textbf{For} $n = 1,\ldots,N$ \textbf{do}:
\begin{itemize}
\item \textbf{For} $l = 1,\ldots,M$, \textbf{do}:
\begin{enumerate}
\item Generate the $l^{\text{th}}$ dataset of length $n$:
$\data_{l} = \dataidx_{1,l},\dots,\dataidx_{n,l}\isim\text{Poisson}\left(\lambda\right)$.
\item Compute the approximate normalizing constant:
$\approxdist^{\text{\upshape\tiny AGHQ}}(\data_{l})$ as in \cref{eqn:aghq-normalizing}.
\item Compute the relative error: 
$E_{n,l} = \absbig{\frac{\dist(\data_{l})}{\approxdist^{\text{\upshape\tiny AGHQ}}(\data_{l})} - 1}$.
\item Compute the \emph{de-trended} log-relative error:
$D_{n,l} = \log E_{n,l} + r_k\log n$.
\end{enumerate}
\end{itemize}
\textbf{Output}: Sampled de-trended errors: $(D_{n,m})_{n\in[N],m\in[M]}$.
\caption{Computing sample rates in the simulation study}
\label{alg:simulation2}
\end{algorithm}

\begin{figure}[p]
\centering
\subfloat[$k = 3$, $\lvert \frac{\dist(\data)}{\GHmarg} - 1 \rvert = \Ordp(n^{-1})$]{\includegraphics[width=.5\textwidth]{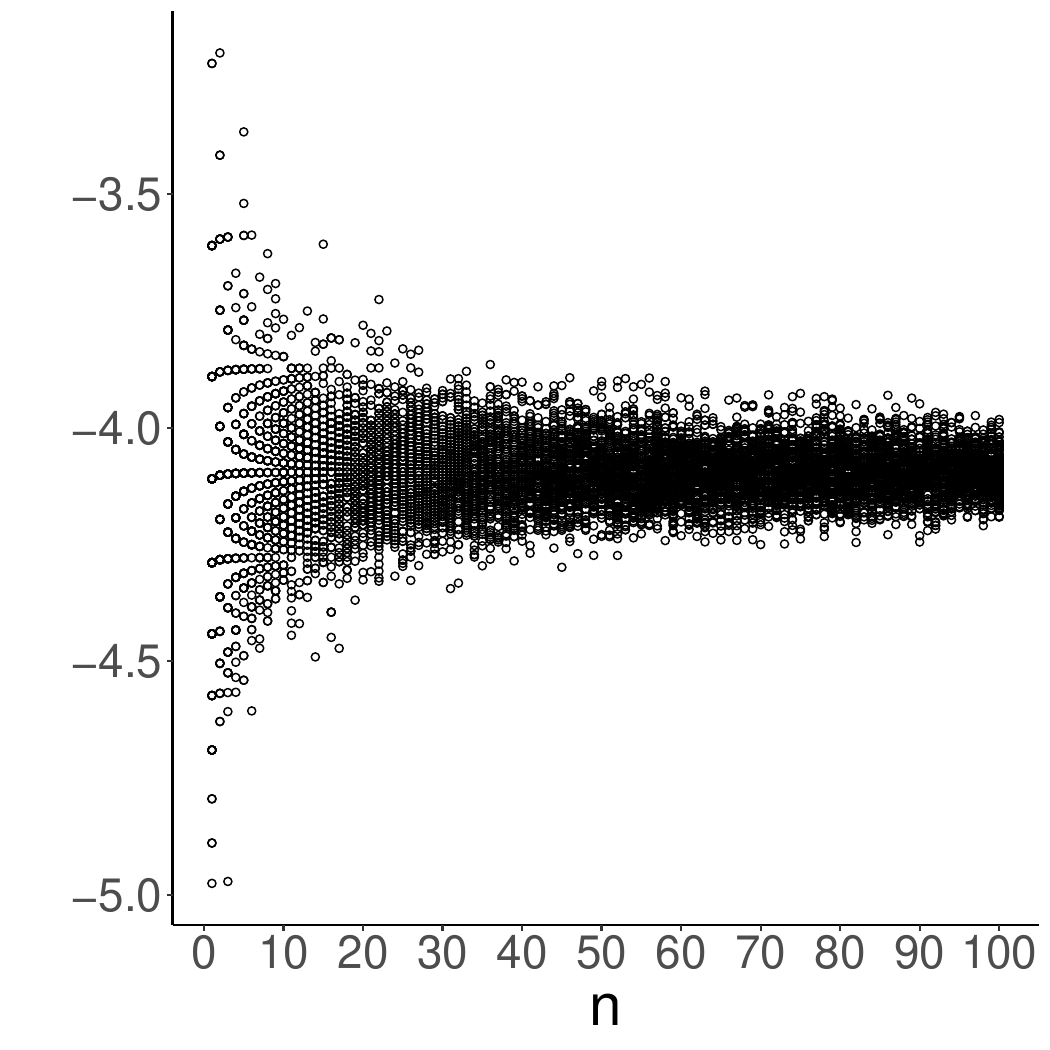}}
\subfloat[$k = 5$, $\lvert \frac{\dist(\data)}{\GHmarg} - 1 \rvert = \Ordp(n^{-2})$]{\includegraphics[width=.5\textwidth]{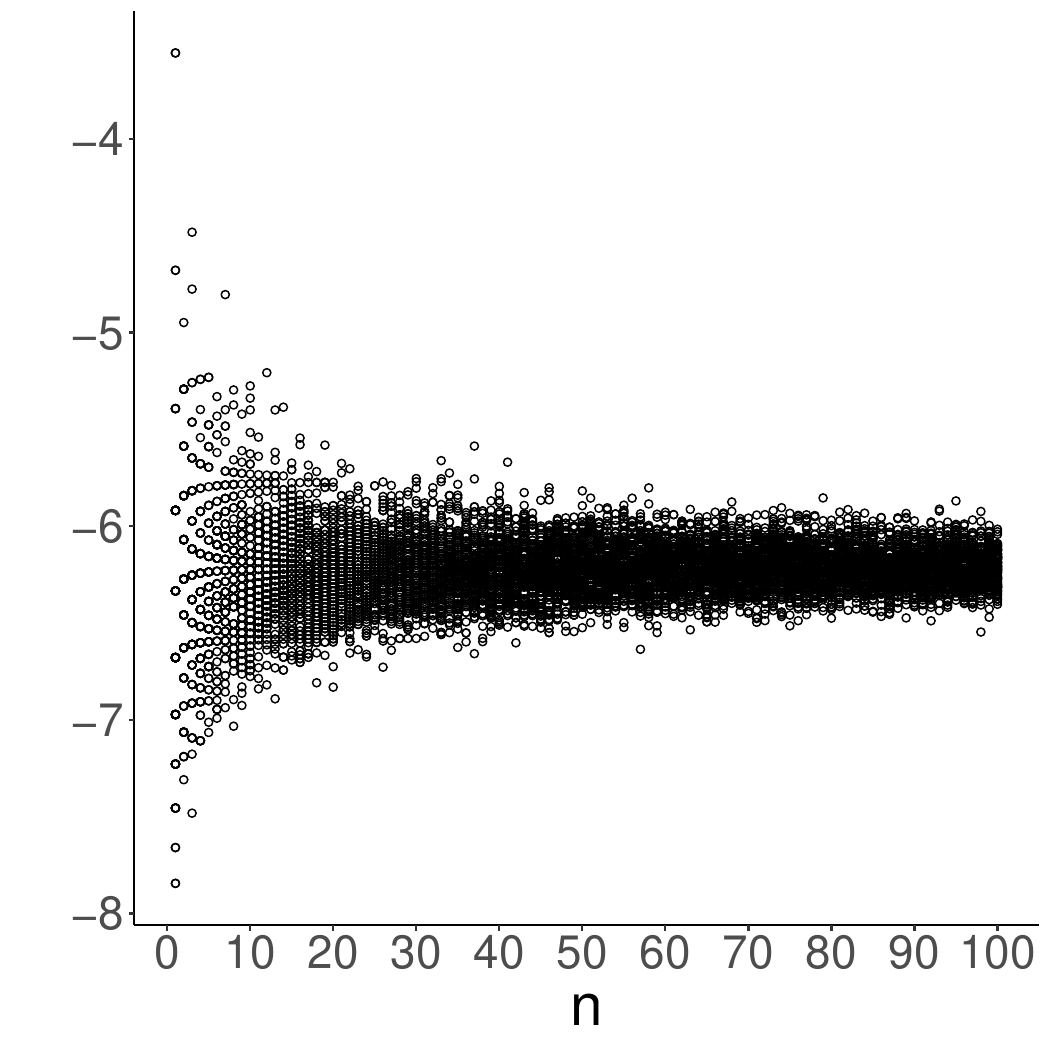}}\\
\subfloat[$k = 7$, $\lvert \frac{\dist(\data)}{\GHmarg} - 1 \rvert = \Ordp(n^{-3})$]
{\includegraphics[width=.5\textwidth]{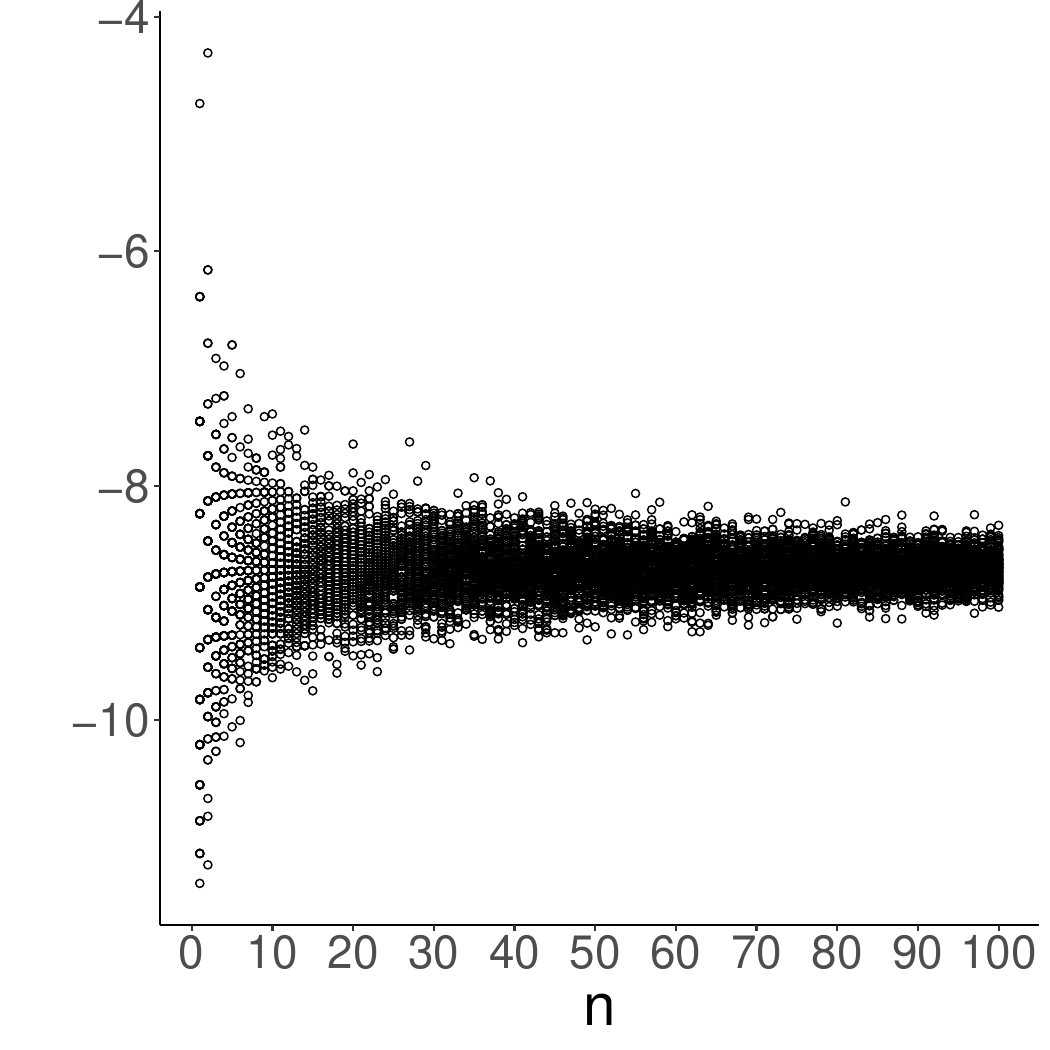}}
\subfloat[$k = 11$, $\lvert \frac{\dist(\data)}{\GHmarg} - 1 \rvert = \Ordp(n^{-4})$]{\includegraphics[width=.5\textwidth]{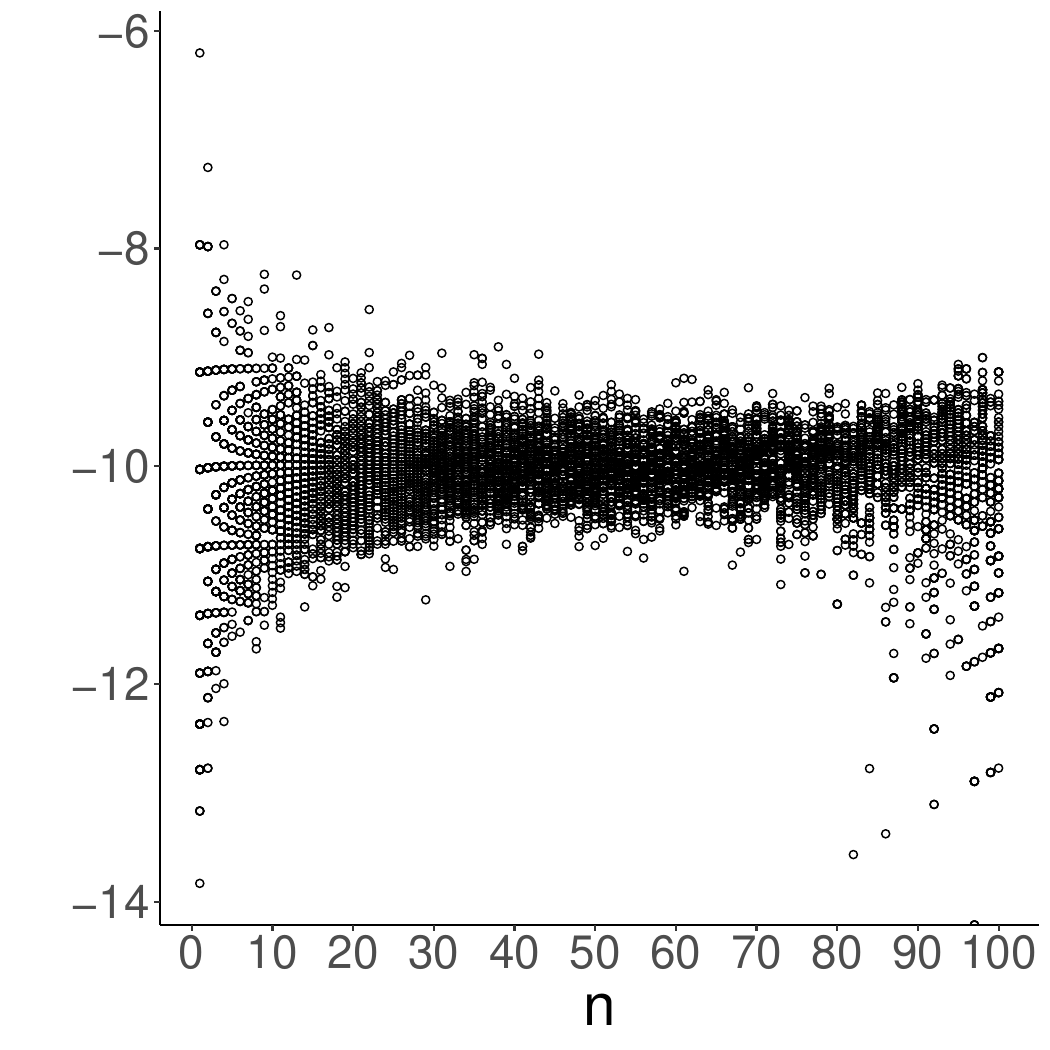}}
\caption{Realized de-trended error rates $\log \relerror(\data) + \lfloor \frac{\quadnum+2}{3} \rfloor \log n$ for data generated from the model (see \cref{eqn:simmodel}), $k = 3,5,7,11$, and $n < 100$.}
\label{fig:relraterandom}
\end{figure}

%% file: sections/loaloa-simulations.tex
\section{Further detail for \cref{SUBSEC:ZEROINF}}\label{sec:loaloa-simulations}

In this section we discuss the results of running \MCMC{} for the zero-inflated binomial geostatistical regression from \cref{subsec:zeroinf} 
and include a brief simulation study to assess the empirical accuracy of the adaptive quadrature-based approximations used in \cref{subsec:marginallaplace}. 

\subsection{\MCMC{} Results}\label{subsec:mcmc-results}

We present the results of running the NUTS sampler through \texttt{tmbstan} using the default settings, with a computation time of $66$ hours for running eight parallel chains of $10,000$ iterations each, including a warmup of $1,000$ iterations. We stress that we are confident an expert user of \MCMC{} could tune the algorithm to produce favourable results, however the observed runtime of almost three days illustrates that doing so would be inconvenient and laborious even for such an expert. In contrast, the \AGHQ{} strategy runs in approximately $90$ seconds without problem-specific tuning, and if any tuning were required it could be done much more efficiently due to the short running time.

\cref{fig:loaloa-mcmc-compare} shows the predicted suitability and incidence maps from the \MCMC{} run, alongside those from \AGHQ{} (\cref{fig:loaloaresults}) for comparison. While the predicted incidence probabilities are visually similar, the \MCMC{} results appear to fail to identify the spatial pattern in suitability, which is the main practical reason to consider this model in the first place. Closer inspection reveals the problem is a failure to accurately sample from the posterior for $\beta_{\texttt{suit}}$, leading to inflated estimates of $\phi(\mb{s})$ at all locations. \cref{fig:loaloa-pairs} shows pairs plots of the two intercepts from the \texttt{tmbstan} output, which illustrate the divergent transitions responsible for the inflated posterior of $\beta_{\texttt{suit}}$. Also shown are corresponding plots of posterior samples from the \AGHQ{} fit for comparison; note the difference in scale for $\beta_{\texttt{suit}}$. Further explanation of the meaning of ``divergent transition'' and advice for tuning the sampler can be found in the \texttt{STAN} documentation at \href{https://mc-stan.org/misc/warnings.html\#divergent-transitions-after-warmup}{https://mc-stan.org/misc/warnings.html\#divergent-transitions-after-warmup} or in \citet{visualizationbayesian}. In particular, \citet{visualizationbayesian} suggest that divergent transitions clustered in one region of the parameter space as is clearly seen in \cref{fig:loaloa-pairs} indicates a serious problem with the ability of the sampler to adequately explore the posterior. The available advice amounts to either changing tuning parameters, which would lead to an increase in computational cost, or rewriting the model entirely. We reiterate that while an expert user may be able to tune \MCMC{} in a problem-specific manner or implement a different type of sampler that would yield satisfactory results for this problem, tuning of this nature is extremely inconvenient due to the already astronomical computational cost of running \MCMC{} in this example, and the complexity of the model. In contrast, the \AGHQ{}-based approximation strategy of \cref{subsec:marginallaplace} runs in minutes without problem-specific tuning, and is hence a potentially appealing practical alternative, that would be made more appealing by the development of convergence theory for it.

\begin{figure}[p]
\caption{\label{fig:loaloa-mcmc-compare} Estimated posterior mean (a,c) suitability probabilities and (b,d) incidence rates for the \texttt{loaloa} example of \cref{sec:highdimexamples} using (a,b) \MCMC{} and (c,d) \AGHQ{}.}
\centering
\makebox{
\subfloat[{$\EE\left[\phi(\cdot) | \mb{Y} \right]$, \MCMC{}}]{\includegraphics[width=0.4\textwidth]{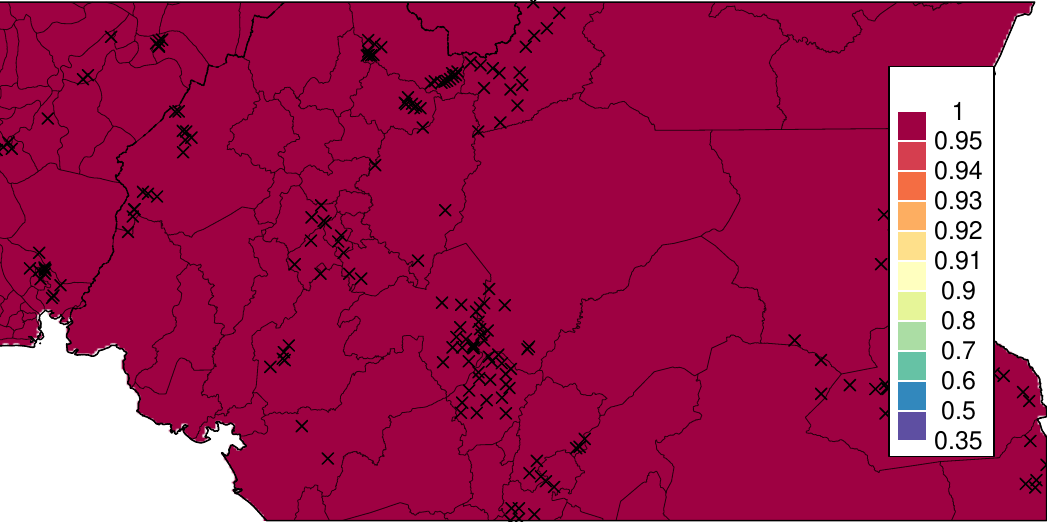}}
\subfloat[{$\EE\left[\phi(\cdot)\times p(\cdot) | \mb{Y}\right]$, \MCMC{}}]{\includegraphics[width=0.4\textwidth]{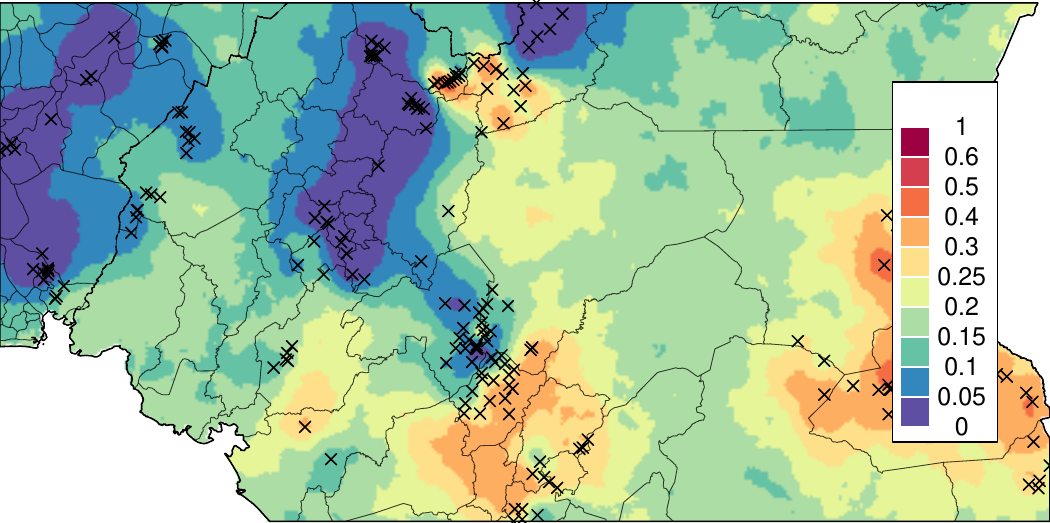}}
} \\
\makebox{
\subfloat[{$\EE\left[\phi(\cdot) | \mb{Y} \right]$, \AGHQ{}}]{\includegraphics[width=0.4\textwidth]{figures/loaloazip/loaloa-zip-postmean.pdf}}
\subfloat[{$\EE\left[\phi(\cdot)\times p(\cdot) | \mb{Y}\right]$, \AGHQ{}}]{\includegraphics[width=0.4\textwidth]{figures/loaloazip/loaloa-risk-postmean.pdf}}
} \\
\end{figure}

\begin{figure}[p]
\caption{\label{fig:loaloa-pairs} Pairs plots for (a) \MCMC{} and (b) \AGHQ{} posterior samples of $\beta_{\texttt{suit}}$ (left) and $\beta_{\texttt{inc}}$ (right). Divergent transitions ($\textcolor{red}{\bullet}$) cause the \MCMC{} algorithm to put non-negligible posterior mass on very large values of $\beta_{\texttt{suit}}$, causing the high estimated posterior mean for $\phi(\cdot)$ at all locations (\cref{fig:loaloa-mcmc-compare}). The clustering of these transitions indicates a potentially serious problem with the algorithm \citep{visualizationbayesian}.}
\centering
\subfloat[{$\widetilde{\dist}(\beta_{\texttt{suit}},\beta_{\texttt{inc}}|\data)$, \MCMC{}}]{\includegraphics[width=0.5\textwidth]{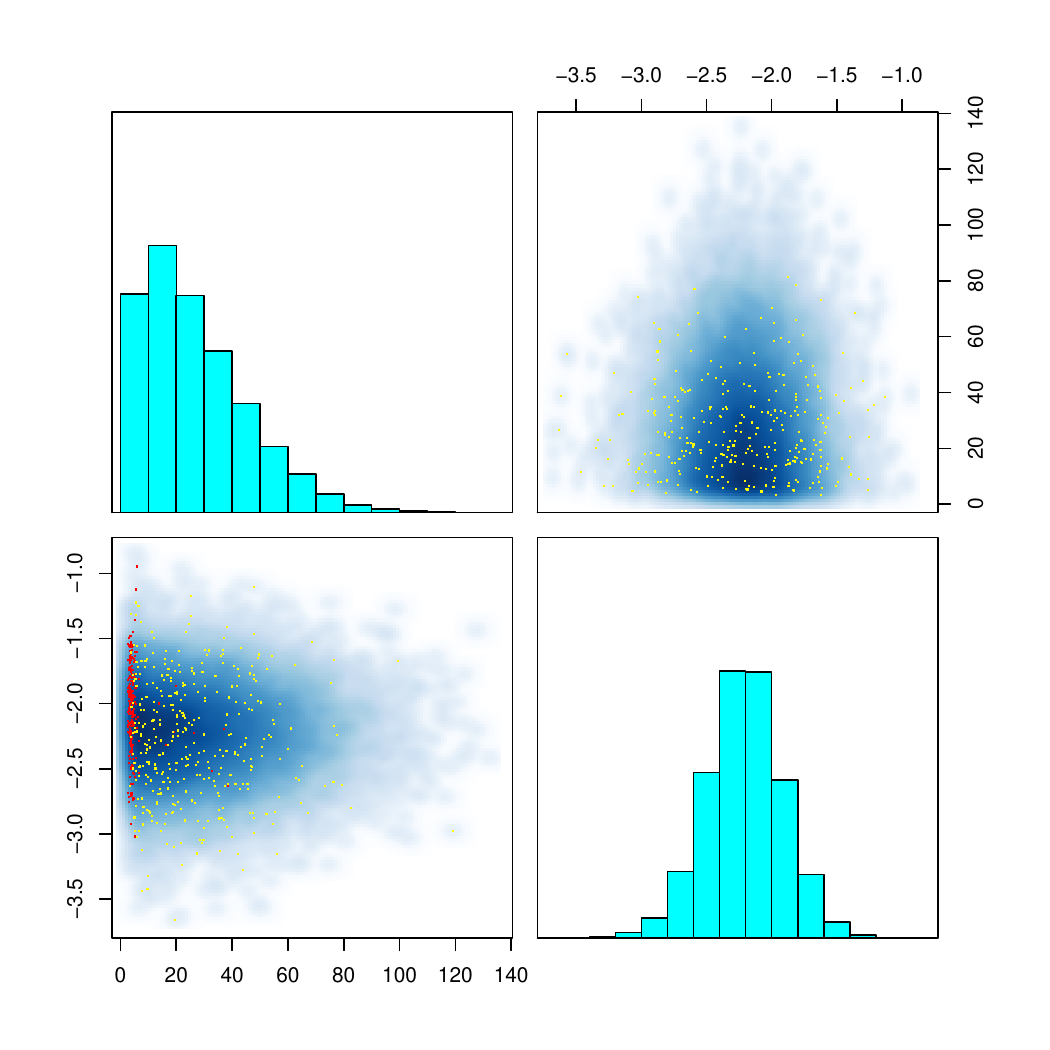}}
\subfloat[{$\widetilde{\dist}(\beta_{\texttt{suit}},\beta_{\texttt{inc}}|\data)$, \AGHQ{}}]{\includegraphics[width=0.5\textwidth]{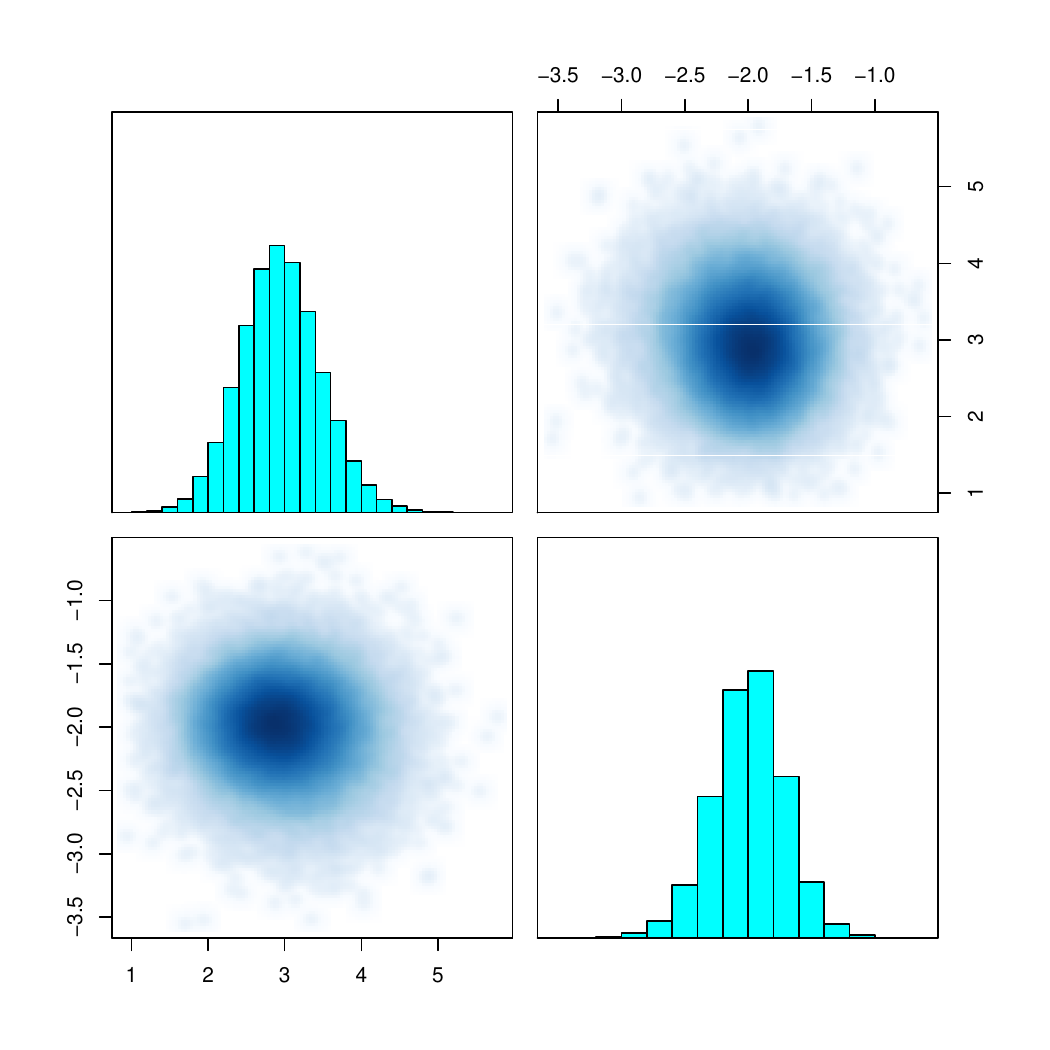}} \\
\end{figure}

\subsection{\MCMC{} Results with $\beta$ Fixed}\label{subsec:mcmc-beta-fixed}

We re-ran the \MCMC{} algorithm with the two intercept parameters fixed at their initial \AGHQ{} estimates of $\beta_{\texttt{suit}} = 2.912111$ and $\beta_{\texttt{inc}} = -1.984655$. The wall time for $8$ parallel chains of $10,000$ iterations of each was $19.5$ hours, compared to a wall time of $224$ seconds for \AGHQ{} with $\quadnum=7$. \AGHQ{} ran in the time taken for approximately $32$ iterations of \MCMC{}. \cref{fig:loaloa-ks} shows the estimated KS statistics between the $2n=380$ marginal distributions of $u(\mb{s}_{i})$ and $v(\mb{s}_{i}),i\in[n]$, for \AGHQ{} and \MCMC{}. The incidence spatial field $v(\mb{s})$ is more accurately estimated than the suitability field $u(\mb{s})$, and both show broad agreement with some villages having moderate disagreement.

\begin{figure}[p]
\caption{\label{fig:loaloa-ks} KS statistics, calculated as maximal difference in approximate marginal posterior empirical cumulative distribution functions from $72,000$ samples for \AGHQ{} and \MCMC{}, for (a) zero-inflation spatial effects $u(\mb{s}_{i})$ and (b) incidence spatial effects $v(\mb{s}_{i}),i\in[n]$.}
\centering
\subfloat[KS (\AGHQ{}/\MCMC{}), $u(\mb{s}_{i})$]{\includegraphics[width=0.5\textwidth]{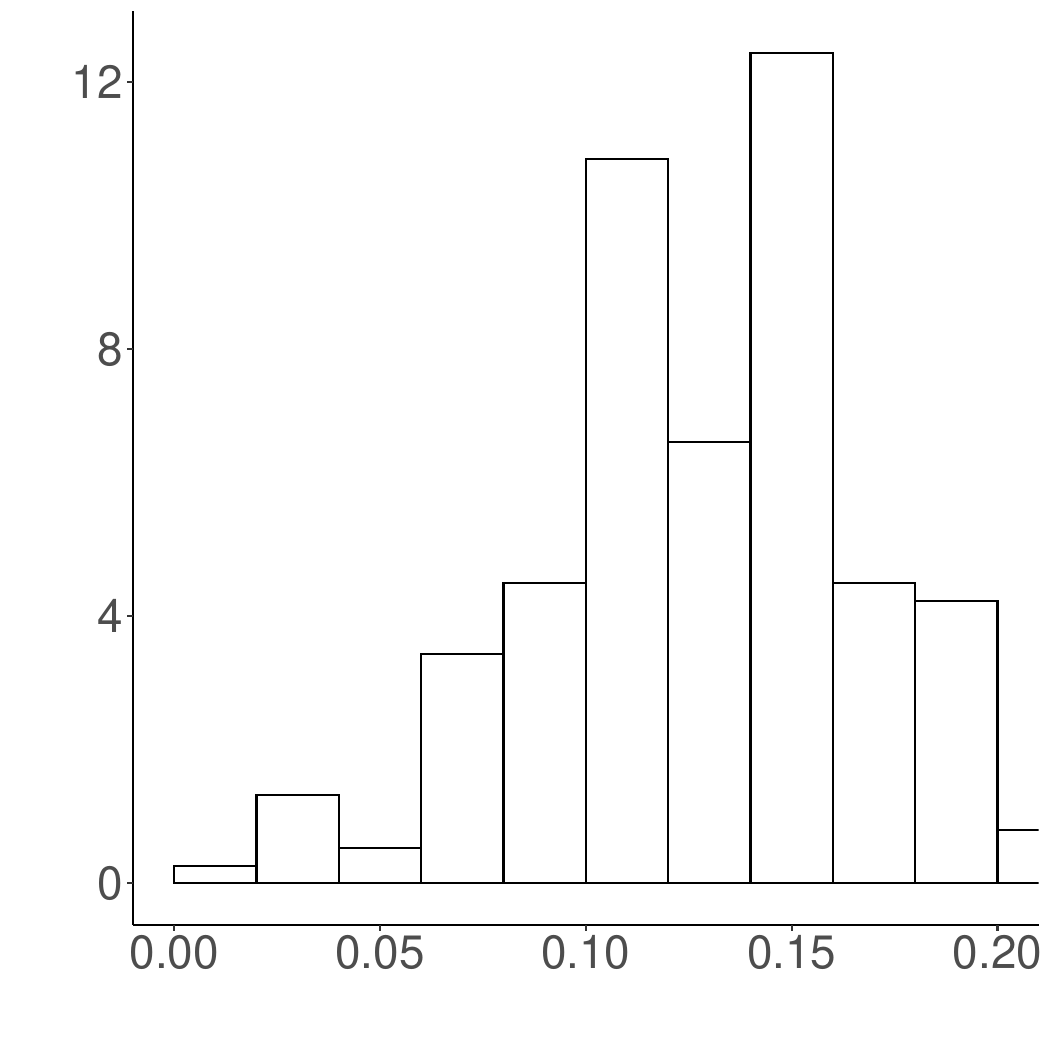}}
\subfloat[KS (\AGHQ{}/\MCMC{}), $v(\mb{s}_{i})$]{\includegraphics[width=0.5\textwidth]{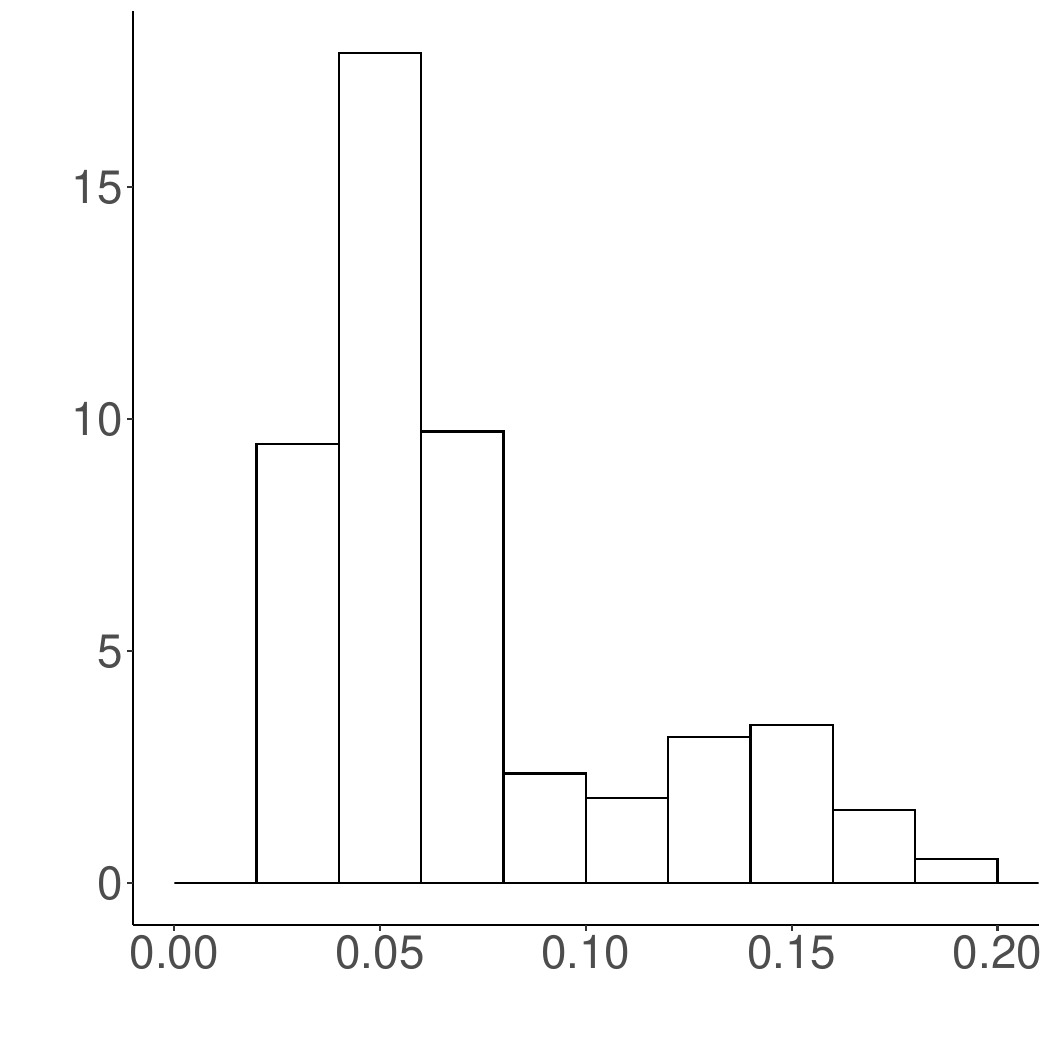}} \\
\end{figure}

\subsection{Empirical Accuracy of \AGHQ{}}\label{subsec:empirical-accuracy-aghq}

In this section we present a brief simulation study to assess the accuracy of using \AGHQ{} to fit the spatial model of \cref{subsec:zeroinf}. We observe that the empirical root-mean-square-error in the posterior mean of the parameter vector decreases with increased simulated sample size, suggesting empirical convergence to the true data-generating parameter. Performing a simulation of this nature is only feasible because of the favourable runtime of \AGHQ{} in this example; running the $250$ total simulations we ran would take approximately $(66\cdot 250) / (24\cdot 365) = 1.88$ years with \MCMC{} based on the one run we completed, and only if we could tune that algorithm to produce satisfactory results.

The simulation procedure is described in \cref{alg:loaloa-simulation}. We use sample sizes 
\*[
	\samplesizes = \{200, 500, 1000, 5000, 10000\}
\]
with $\simsize = 50$ simulated datasets for each size. 
To make the simulation as realistic as possible, we set the true parameter equal to 
the \AGHQ{} estimated posterior mean of $\highdimparam = (\mb{U},\beta_{\texttt{suit}},\mb{V},\beta_{\texttt{inc}})$ from \cref{subsec:zeroinf}. \cref{fig:loaloa-boxplot} shows the empirical RMSE and average coverage of quantile-based pointwise approximate $95\%$ credible intervals for the parameter $\mb{W}$ as well as the suitability probabilities $\phi(\cdot)$ and (conditional) incidence probabilities $p(\cdot)$ from each simulation. The RMSE for all three sets of parameters decreases on average with higher simulated sample size, with $p(\cdot)$ appearing especially accurately estimated and $\phi(\cdot)$ quite accurate as well. The credible intervals for $\mb{W}$ and $\phi(\cdot)$ appear conservative---the less severe of the two types of possible inaccurate coverage---while those for $p(\cdot)$ appear to generally agree with their nominal level.

\begin{algorithm}[p]
\textbf{Input}: size $n_{0}$ of original \texttt{loaloa} dataset, sample sizes $\samplesizes\subseteq\Nats_{\geq n_0}$ to simulate data for, number of simulations $\simsize\in\Nats$ of each sample size to do, true parameters $\highdimparam_{0} = (\mb{U}_{0},\beta_{\texttt{suit},0},\mb{V}_{0},\beta_{\texttt{inc},0})\in\R^{\paramdimbig}$.

\textbf{For} $n \in \samplesizes$ \textbf{do}:
\begin{itemize}
	\item Choose village indices $\JJ \subseteq[n_{0}]$ uniformly and with replacement such that $\abssmall{\JJ} = n$ and each $j\in[n_{0}]$ appears at minimum once in $\JJ$.
\item \textbf{For} $l = 1,\ldots,\simsize$, \textbf{do}:
\begin{enumerate}
	\item Generate the $l^{\text{th}}$ dataset of length $n$, $\data_{l} = \bracevec{Y_{l,i},i\in[n]}$ as follows. 

	\textbf{For} $i = 1,\ldots,n$, \textbf{do}:
	\begin{enumerate}
		\item Let $p_{i} = \left[ 1 + \exp(-\beta_{\texttt{inc},0} - V_{\JJ_{i}})\right]^{-1}$ and $\phi_{i} = \left[ 1 + \exp(-\beta_{\texttt{suit},0} - U_{\JJ_{i}})\right]^{-1}$,
		\item Generate $Z_{i}\sim\text{Binomial}(N_{\JJ_{i}},p_{i})$ and $X \sim \text{Unif}(0,1)$,
		\item If $X \leq \phi_{i}$ set $Y_{l,i} = Z_{i}$, else set $Y_{l,i} = 0$.
	\end{enumerate}
	\item Fit the model (\cref{subsec:zeroinf}) using the \AGHQ{} procedure (\cref{subsec:marginallaplace}) to the data $\data_{l}$.
	\item Compute the estimates:
		\begin{itemize}
			\item Approximate posterior mean $\widehat{\mb{W}}_{l,n} = \approxEE(\mb{W}|\data_{l})$,
			\item Approximate pointwise $95\%$ credible interval: $$\left(\mb{W}_{l,n}^{(\text{lower})},\mb{W}_{l,n}^{(\text{upper})}\right) = \bracevec{(\approxquant{j}{.025},\approxquant{j}{.975}):j\in[\paramdimbig], \data = \data_{l}},$$
			\item (similarly for $\phi(\cdot)$ and $p(\cdot)$).
		\end{itemize}
	\item Compute the metrics:
		\begin{itemize}
			\item Root-Mean-Square-Error $\text{RMSE}_{l,n} = \left[\frac{1}{\paramdimbig}\left(\widehat{\mb{W}}_{l,n} - \mb{W}_{0} \right)^{\tpose}\left(\widehat{\mb{W}}_{l,n} - \mb{W}_{0} \right)\right]^{1/2}$,
			\item Average coverage: $$\text{COVR}_{l,n} = \frac{1}{m}\sum_{j=1}^{m}\ind\left[\left(\mb{W}_{l,n}^{(\text{lower})}\right)_{j} \leq \left(\mb{W}_{0}\right)_{j}\right]\times\ind\left[\left(\mb{W}_{l,n}^{(\text{upper})}\right)_{j} \geq \left(\mb{W}_{0}\right)_{j}\right],$$
			\item (similarly for $\phi(\cdot)$ and $p(\cdot)$).
		\end{itemize}
\end{enumerate}
\end{itemize}
\textbf{Output}: Sampled RMSE and coverage values for $\mb{W},\phi(\cdot)$ and $p(\cdot)$ for $(n,l)\in\samplesizes\times[\simsize]$.
\caption{Simulations to assess empirical accuracy of \AGHQ{} for \cref{subsec:zeroinf}}
\label{alg:loaloa-simulation}
\end{algorithm}

\begin{figure}[p]
\caption{\label{fig:loaloa-boxplot} Simulated RMSE (a--c) and average coverage (d--f) for the parameter vector $\mb{W}$ (a,d), suitability probabilities $\phi(\cdot)$ (b,e), and (conditional) incidence probabilities $p(\cdot)$ (c,f) for the zero-inflated binomial model of \cref{subsec:empirical-accuracy-aghq}.}
\centering
\makebox{
	\subfloat[RMSE, $\mb{W}$]{\includegraphics[width=0.3\textwidth]{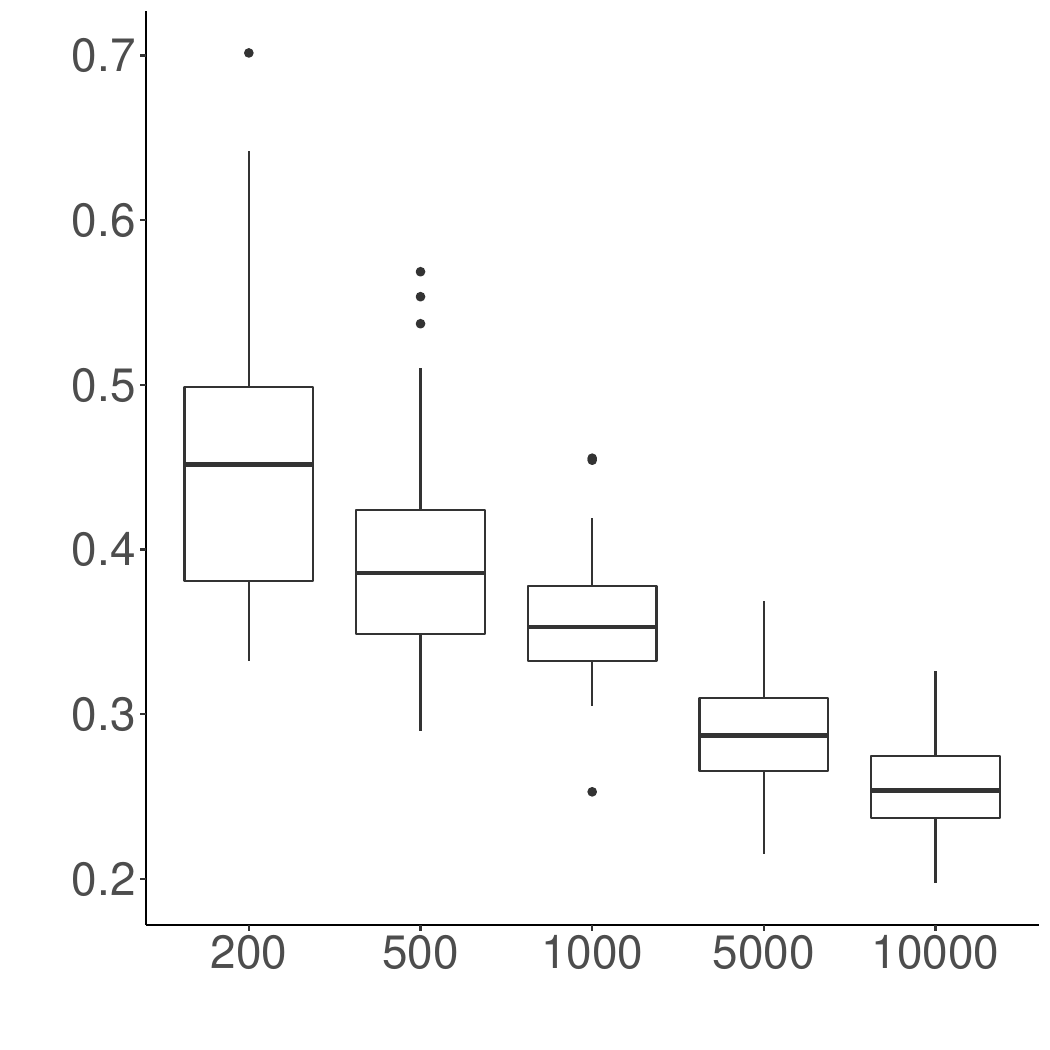}}
	\subfloat[RMSE, $\phi(\cdot)$]{\includegraphics[width=0.3\textwidth]{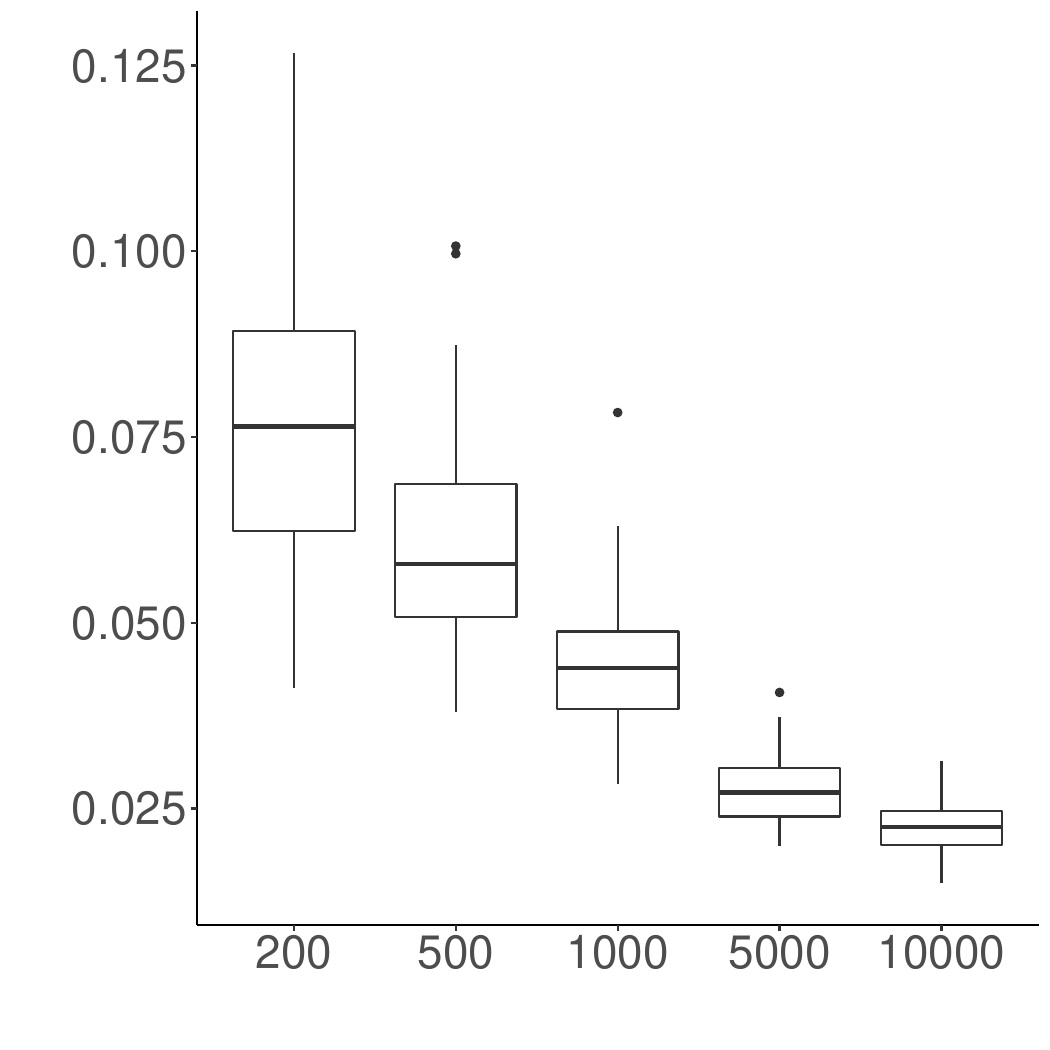}}
	\subfloat[RMSE, $p(\cdot)$]{\includegraphics[width=0.3\textwidth]{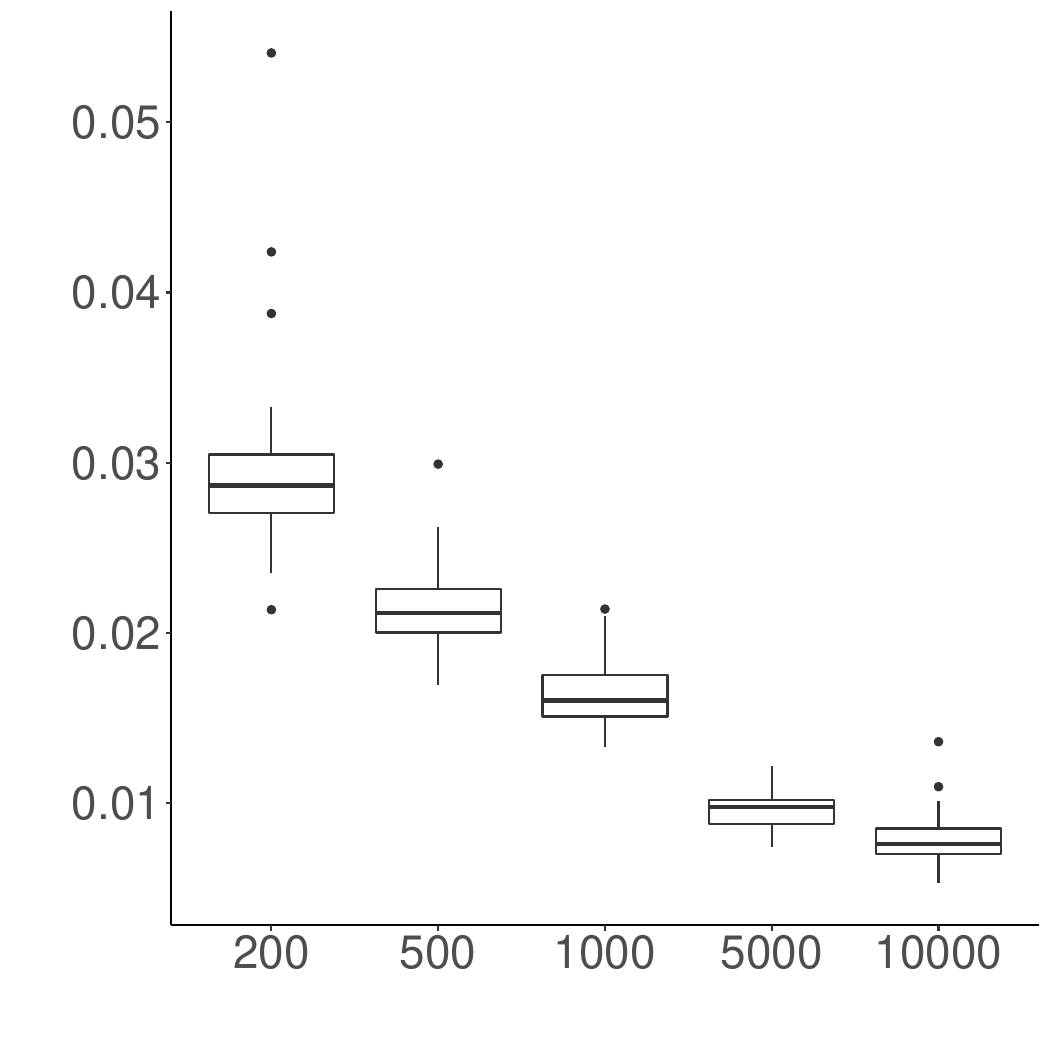}}
} \\
\makebox{
	\subfloat[Coverage, $\mb{W}$]{\includegraphics[width=0.3\textwidth]{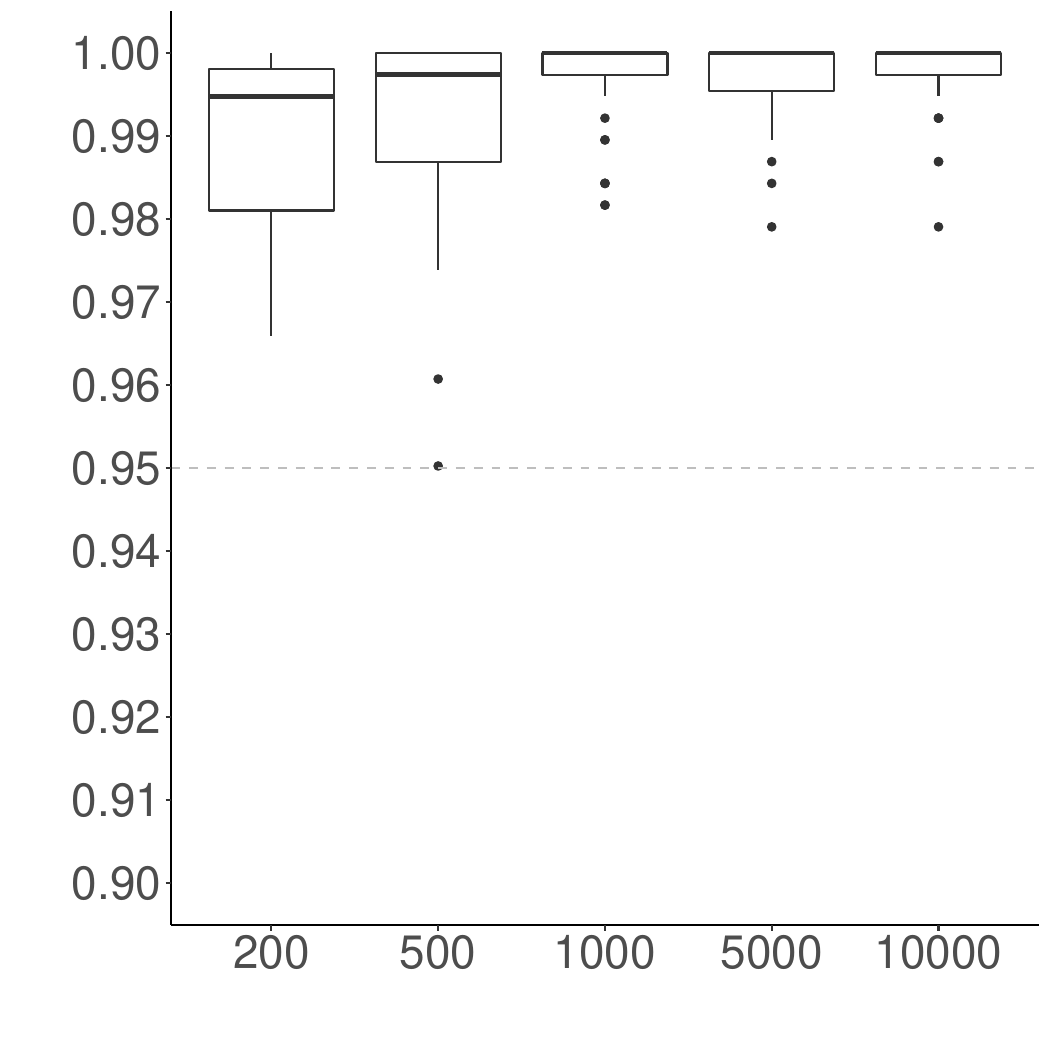}}
	\subfloat[Coverage, $\phi(\cdot)$]{\includegraphics[width=0.3\textwidth]{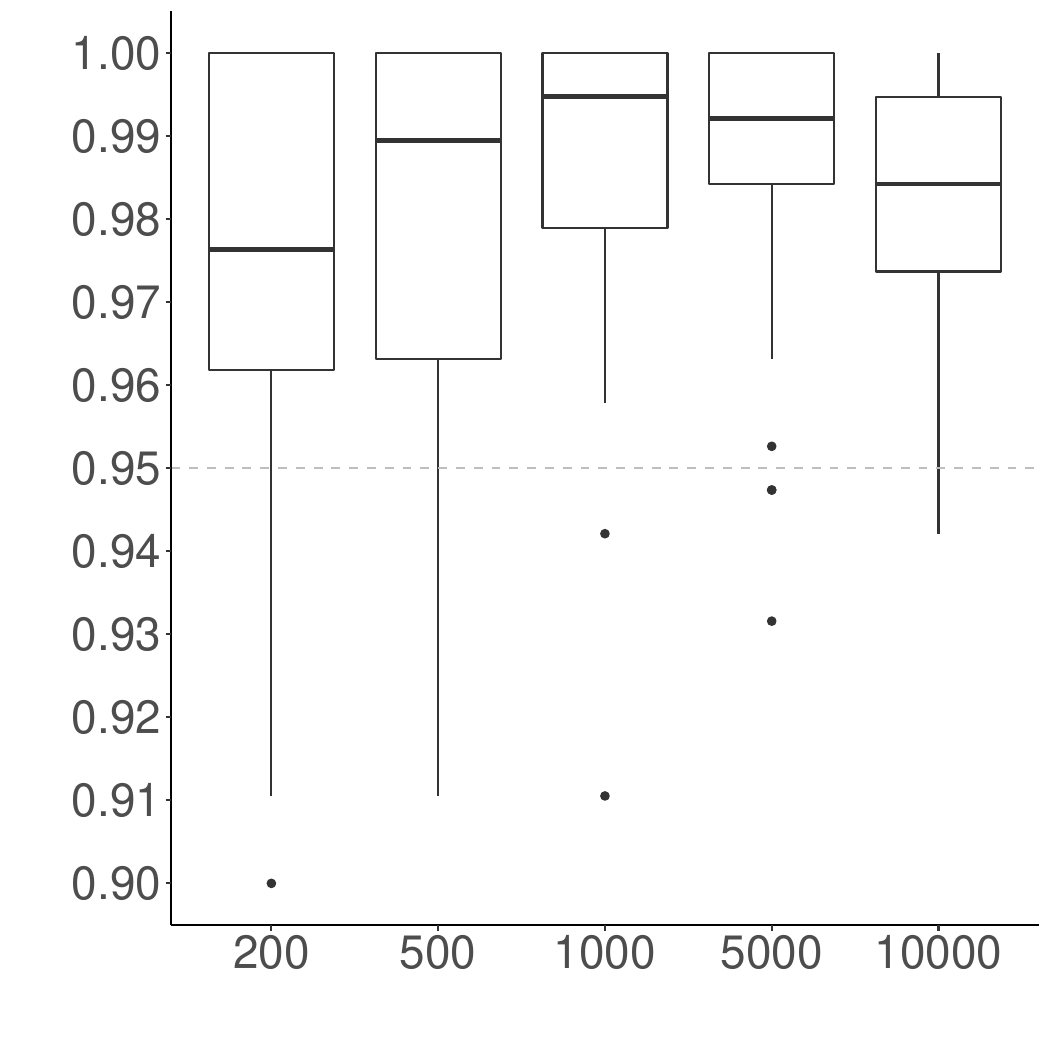}}
	\subfloat[Coverage, $p(\cdot)$]{\includegraphics[width=0.3\textwidth]{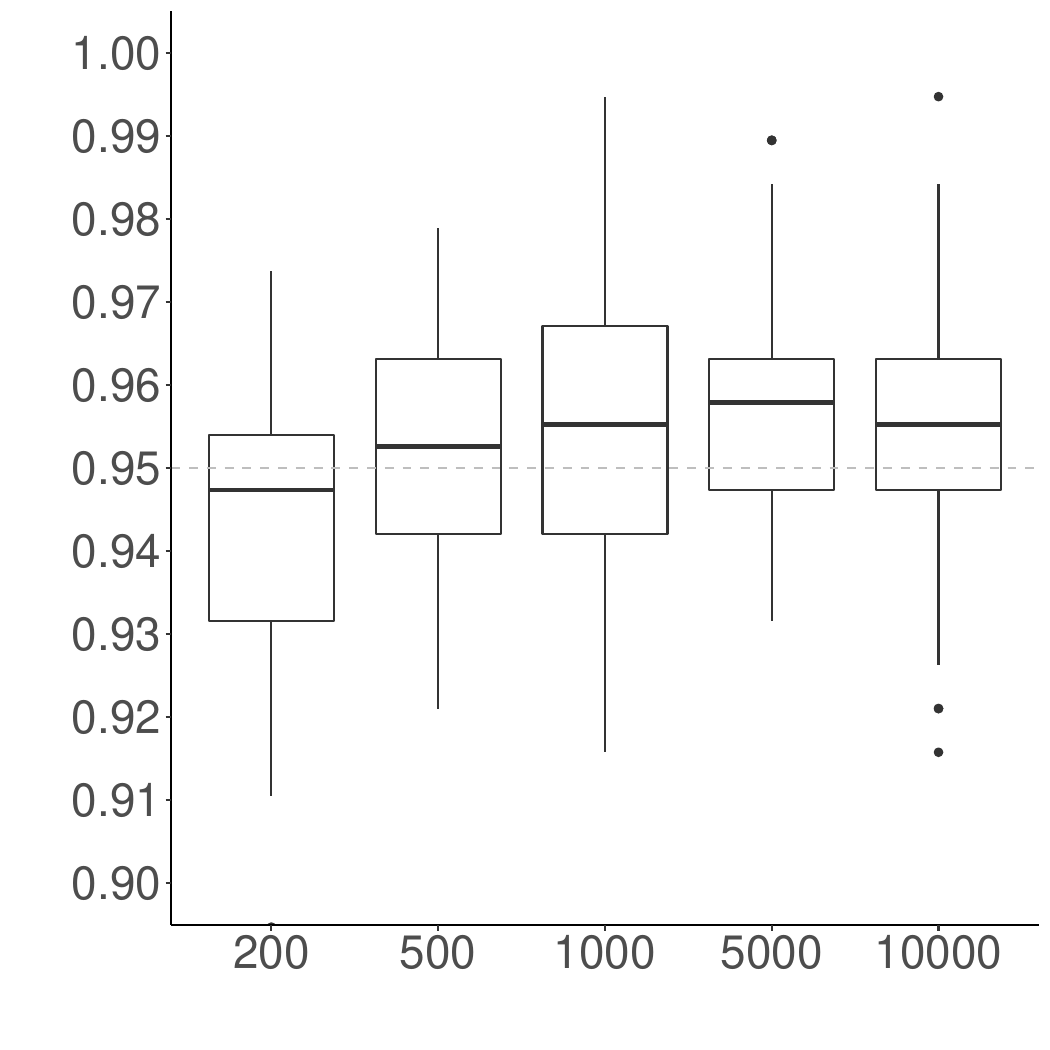}}
} \\
\end{figure}

%% file: sections/glossary.tex
\section{Glossary of Commonly Used Terms}
\centering
\begin{tabular}{ |m{1.7cm} ? >{\centering}m{5.5cm}| m{6cm}|  }
 \hline
 \textbf{Symbol}     & \textbf{Mathematical Description}  & \textbf{Verbal Description} \\
 \thickhline
 $\param$   &   -  & Model Parameter \\
 \hline
 $\data$   &   -  & Data vector \\
 \hline
$\llhood(\param)$   &   $\log \dist(\data \setdelim \param)$  & Log-likelihood function \\
 \hline
 $\dist(\data)$ & - & Marginal likelihood\\
 \hline
 $\approxdist(\data)$ & \cref{eqn:aghq-normalizing} & Approximate marginal likelihood \\
 \hline
 $\dist(\param)$ & - & Prior distribution on $\paramspace$\\
 \hline 
 $\llandp(\param)$ & $\llhood(\param) + \log\dist(\param)$ & Unnormalized log-posterior\\
 \hline
 $\paramtrue$ & - & Closest parameter to truth\\
 \hline 
 $\parammle$ & $\argmax_{\param\in\paramspace}\llhood(\param)$ & Maximum likelihood estimator\\
 \hline
 $\parammode $ & $\argmax_{\param\in\paramspace}\llandp(\param)$ & Posterior mode\\
 \hline  
 $\logposthess(\param) $ & $-\frac{\partial^{2}}{\partial \param\partial \param\tpose}\llandp(\param)$ & Negative Hessian of $\llandp$\\
 \hline 
 $\logpostchol$ & $[\logposthess(\parammode)]^{-1} = \logpostchol\logpostchol\tpose$ & Lower Cholesky decomposition \\
 \hline 
 $\quadnum$ & - & Desired polynomial accuracy\\
 \hline 
 $\hermitezerovec$ & \cref{subsec:adaptivequadrature} & Quadrature point ($\in \Reals^\paramdim$)\\
 \hline
 $\weight(\hermitezerovec)$ & \cref{subsec:adaptivequadrature} & Quadrature weight assigned to $\hermitezerovec$  \\
 \hline
 $\quadpointset$ &  \cref{subsec:adaptivequadrature} & Set of all quadrature points\\
 \hline
 $\hermitezerobound $ & $\sup_{\hermitezerovec \in \hermiteset(\quadnum,\paramdim)} \norm{\hermitezerovec}_2$ &Max quadrature point norm \\
 \hline
 $\tau^{(j)}_{< \dummytau}$ & \multicolumn{2}{c|}{$ \left\{(t_3,\dots,t_{2\quadnum}) \in \PosInts^{2\quadnum -3} \setdelim \sum_{s=3}^{2\quadnum} t_s = j \text{ and } \sum_{s=3}^{2\quadnum} s t_s \leq \dummytau - 1 \right\}$} \\
 \hline 
 $\tau^{(j)}_{\geq \dummytau}$ & \multicolumn{2}{c|}{ $\left\{(t_3,\dots,t_{2\quadnum}) \in \PosInts^{2\quadnum-3} \setdelim \sum_{s=3}^{2\quadnum} t_s = j \text{ and } \sum_{s=3}^{2\quadnum} s t_s \geq \dummytau \right\}$} \\
 \hline
 $\tau^{(j)}_{= \dummytau}$ & \multicolumn{2}{c|}{ $\left\{(t_3,\dots,t_{2\quadnum}) \in \PosInts^{2\quadnum-3} \setdelim \sum_{s=3}^{2\quadnum} t_s = j \text{ and } \sum_{s=3}^{2\quadnum} s t_s = \dummytau \right\}$} \\
 \hline
 $\tsum$ &  \multicolumn{2}{c|}{$\sum_{s=3}^{a} s t_s$ for any $\tvec = (t_3,\dots,t_{a}) \in \PosInts^{a-3}$ for $a \geq 3$}   \\
 \hline
 $\multicoeff{j}{\tvec}$ & ${j \choose t_3,\dots,t_{2\quadnum}}$ & Multinomial coefficients \\
 \hline 
 $\hessbig $ & $\sup_{\param \in \ball{\paramdim}{\paramtrue}{\universalradius}} \eigen_{1}(\logposthess(\param))  \leq n \hessbig$ & Maximal  Eigenvalue\\
 \hline  
 $\hessbigmode  $ & $\eigen_{1}(\logposthess(\parammode))/n$ & Largest normalized Eigenvalue of the negative Hessian at the posterior mode \\
 \hline 
  $\hesssmall $ & $\inf_{\param \in \ball{\paramdim}{\paramtrue}{\universalradius}} \eigen_{\paramdim}(\logposthess(\param))  \geq n \hesssmall$ & Minimal Eigenvalue\\
 \hline  
$\hesssmallmode  $ & $\eigen_{\paramdim}(\logposthess(\parammode))/n$ & Smallest normalized Eigenvalue of the negative Hessian at the posterior mode \\
\hline
$\cdfpost{j}(x)$ & \cref{subsec:posterior} & True $j$th marginal CDF of $\param$\\
 \hline 
 $ \quant{j}{\alpha}$ & $\inf \{x \setdelim \cdfpost{j}(x) \geq \alpha\}$ & True $\alpha$-quantile for the $j$th marginal of $\param$\\
 \hline
$\cdfapprox{j}(x)$ & \cref{sec:approximatebayesianinference} & Approx $j$th marginal CDF of $\param$\\
 \hline 
 $\approxquant{j}{\alpha}$  & $\inf \{x \setdelim \cdfapprox{j}(x) \geq \alpha\}$ & Approx $\alpha$-quantile for the $j$th marginal of $\param$\\
 \hline
 $\derivboundmode$ & \multicolumn{2}{c|}{$\sup_{\dumderivvec: \abssmall{\dumderivvec}\leq \derivnum} \sup_{\param\in\ball{\paramdim}{\parammode}{\hermitezerobound\lvert\logpostchol\rvert \lor \shrinkingrad}} \absbig{\partial^{\dumderivvec} \llandp(\parammode)}$}\\
 \hline
\end{tabular}

%% file: main.bbl
\begin{thebibliography}{}

\bibitem[\protect\citeauthoryear{{Almutiry}, {Warriyar K.V}, and
  {Deardon}}{{Almutiry} et~al.}{2020}]{epi}
{Almutiry}, W., V.~{Warriyar K.V}, and R.~{Deardon} (2020).
\newblock Continuous time individual-level models of infectious disease:
  {EpiILMCT}.
\newblock arXiv:2006.00135v1.

\bibitem[\protect\citeauthoryear{Bojanov and Petrov}{Bojanov and
  Petrov}{2001}]{bojanov01interval}
Bojanov, B. and P.~P. Petrov (2001).
\newblock Gaussian interval quadrature formula.
\newblock {\em Numerische Mathematik\/}~{\em 87}, 625--643.

\bibitem[\protect\citeauthoryear{Braun}{Braun}{2014}]{trustoptim}
Braun, M. (2014).
\newblock {trustOptim}: An {R} package for trust region optimization with
  sparse {Hessians}.
\newblock {\em Journal of Statistical Software\/}~{\em 60}, 1--16.

\bibitem[\protect\citeauthoryear{Brown}{Brown}{2011}]{geostatsp}
Brown, P. (2011).
\newblock Model-based geostatistics the easy way.
\newblock {\em Journal of Statistical Software\/}~{\em 73}, 423--498.

\bibitem[\protect\citeauthoryear{Cagnone and Monari}{Cagnone and
  Monari}{2013}]{latent}
Cagnone, S. and P.~Monari (2013).
\newblock Latent variable models for ordinal data by using the adaptive
  quadrature approximation.
\newblock {\em Computational Statistics\/}~{\em 28}, 597--619.

\bibitem[\protect\citeauthoryear{Carpenter, Gelman, Hoffman, Lee, Goodrich,
  Betancourt, Brubaker, Guo, Li, and Riddell}{Carpenter et~al.}{2017}]{stan}
Carpenter, B., A.~Gelman, M.~D. Hoffman, D.~Lee, B.~Goodrich, M.~Betancourt,
  M.~Brubaker, J.~Guo, P.~Li, and A.~Riddell (2017).
\newblock Stan: A probabilistic programming language.
\newblock {\em Journal of Statistical Software\/}~{\em 76}.

\bibitem[\protect\citeauthoryear{Davis and Rabinowitz}{Davis and
  Rabinowitz}{1975}]{numint}
Davis, P.~J. and P.~Rabinowitz (1975).
\newblock {\em Methods of Numerical Integration}.
\newblock Academic Press.

\bibitem[\protect\citeauthoryear{{Dick}, {Gantner}, {Le Gia}, and
  {Schwab}}{{Dick} et~al.}{2019}]{dick19inversion}
{Dick}, J., R.~N. {Gantner}, Q.~T. {Le Gia}, and C.~{Schwab} (2019).
\newblock Higher order quasi-{Monte Carlo} integration for {Bayesian} {PDE}
  inversion.
\newblock {\em Computers and Mathematics with Applications\/}~{\em 77},
  144--172.

\bibitem[\protect\citeauthoryear{Diggle and Giorgi}{Diggle and
  Giorgi}{2016}]{geostatlowresource}
Diggle, P.~J. and E.~Giorgi (2016).
\newblock Model-based geostatistics for prevalence mapping in low-resource
  settings.
\newblock {\em Journal of the American Statistical Association\/}~{\em 111},
  1096--1120.

\bibitem[\protect\citeauthoryear{Duvenaud and Adams}{Duvenaud and
  Adams}{2015}]{bbsvi}
Duvenaud, D. and R.~P. Adams (2015).
\newblock {Black-box stochastic variational inference in five lines of Python}.
\newblock {\em NIPS Workshop on Black-box Learning and Inference\/}.

\bibitem[\protect\citeauthoryear{Eadie and Harris}{Eadie and
  Harris}{2016}]{gwen2}
Eadie, G.~M. and W.~E. Harris (2016).
\newblock Bayesian mass estimates of the {M}ilky {W}ay: the dark and light
  sides of parameter assumptions.
\newblock {\em The Astrophysical Journal\/}~{\em 829}.

\bibitem[\protect\citeauthoryear{Falbel and Luraschi}{Falbel and
  Luraschi}{2020}]{torch}
Falbel, D. and J.~Luraschi (2020).
\newblock {\em {torch: tensors and neural metworks with 'GPU' acceleration}}.
\newblock R package version 0.1.1.

\bibitem[\protect\citeauthoryear{Fan and Lv}{Fan and Lv}{2008}]{Fan}
Fan, J. and J.~Lv (2008).
\newblock Sure independence screening for ultrahigh dimensional feature space.
\newblock {\em Journal of the Royal Statistical Society, Series B: Statistical
  Methodology\/}~{\em 70}, 849--911.

\bibitem[\protect\citeauthoryear{Fuglstad, Simpson, Lindgren, and Rue}{Fuglstad
  et~al.}{2019}]{pcpriormatern}
Fuglstad, G.-A., D.~Simpson, F.~Lindgren, and H.~Rue (2019).
\newblock Constructing priors that penalize the complexity of {G}aussian random
  fields.
\newblock {\em Journal of the American Statistical Association\/}~{\em 114},
  445--452.

\bibitem[\protect\citeauthoryear{Gabry, Simpson, Vehtari, Betancourt, and
  Gelman}{Gabry et~al.}{2019}]{visualizationbayesian}
Gabry, J., D.~Simpson, A.~Vehtari, M.~Betancourt, and A.~Gelman (2019).
\newblock Visualization in {Bayesian} workflow.
\newblock {\em Journal of the Royal Statistical Society, Series A: Statistics
  in Society\/}~{\em 182}, 389--402.

\bibitem[\protect\citeauthoryear{Geirsson, Hrafnkelsson, Simpson, and
  Sigurdarson}{Geirsson et~al.}{2020}]{geirsson20lgm}
Geirsson, O.~P., B.~Hrafnkelsson, D.~Simpson, and H.~Sigurdarson (2020).
\newblock {LGM} split sampler: An efficient {MCMC} sampling scheme for latent
  {Gaussian} models.
\newblock {\em Statistical Science\/}~{\em 35}, 218--233.

\bibitem[\protect\citeauthoryear{Genz and Keister}{Genz and
  Keister}{1996}]{genz1996fully}
Genz, A. and B.~D. Keister (1996).
\newblock Fully symmetric interpolatory rules for multiple integrals over
  infinite regions with {Gaussian} weight.
\newblock {\em Journal of Computational and Applied Mathematics\/}~{\em 71},
  299--309.

\bibitem[\protect\citeauthoryear{Geyer}{Geyer}{2020}]{trustdense}
Geyer, C.~J. (2020).
\newblock {\em trust: Trust Region Optimization}.
\newblock R package version 0.1-8.

\bibitem[\protect\citeauthoryear{Giorgi and Diggle}{Giorgi and
  Diggle}{2017}]{prevmap}
Giorgi, E. and P.~J. Diggle (2017).
\newblock {PrevMap}: An {R} package {for} prevalence mapping.
\newblock {\em Journal of Statistical Software\/}~{\em 78}.

\bibitem[\protect\citeauthoryear{Giorgi, Schluter, and Diggle}{Giorgi
  et~al.}{2018}]{loaloazero}
Giorgi, E., D.~K. Schluter, and P.~J. Diggle (2018).
\newblock Bivariate geostatistical modelling of the relationship between {L}oa
  loa prevalence and intensity of infection.
\newblock {\em Environmetrics\/}~{\em 29}.

\bibitem[\protect\citeauthoryear{{Heiss} and {Winschel}}{{Heiss} and
  {Winschel}}{2008}]{heiss08sparse}
{Heiss}, F. and V.~{Winschel} (2008).
\newblock Likelihood approximation by numerical integration on sparse grids.
\newblock {\em Journal of Econometrics\/}~{\em 144}, 62--80.

\bibitem[\protect\citeauthoryear{Hoffman and Gelman}{Hoffman and
  Gelman}{2014}]{nuts}
Hoffman, M.~D. and A.~Gelman (2014).
\newblock The no-{U}-turn sampler: Adaptively setting path lengths in
  {Hamiltonian Monte Carlo}.
\newblock {\em Journal of Machine Learning Research\/}~{\em 15}, 1593--1623.

\bibitem[\protect\citeauthoryear{Jin and Andersson}{Jin and
  Andersson}{2020}]{adaptive_GH_2020}
Jin, S. and B.~Andersson (2020).
\newblock A note on the accuracy of adaptive {G}auss–{H}ermite quadrature.
\newblock {\em Biometrika\/}~{\em 107}, 737--744.

\bibitem[\protect\citeauthoryear{Kass and Steffey}{Kass and
  Steffey}{1989}]{kassandsteffy}
Kass, R. and D.~Steffey (1989).
\newblock Approximate {Bayesian} inference in conditionally independent
  hierarchical models (parametric empirical {Bayes} models).
\newblock {\em Journal of the American Statistical Association\/}~{\em 84},
  717--726.

\bibitem[\protect\citeauthoryear{Kass, Tierney, and Kadane}{Kass
  et~al.}{1990}]{validitylaplace}
Kass, R.~E., L.~Tierney, and J.~B. Kadane (1990).
\newblock The validity of posterior expansions based on {L}aplace's method.
\newblock {\em Bayesian and Likelihood Methods in Statistics and
  Econometrics\/}, 473--488.

\bibitem[\protect\citeauthoryear{Kristensen, Nielson, Berg, Skaug, and
  Bell}{Kristensen et~al.}{2016}]{tmb}
Kristensen, K., A.~Nielson, C.~W. Berg, H.~Skaug, and B.~M. Bell (2016).
\newblock {TMB}: automatic differentiation and {L}aplace approximation.
\newblock {\em Journal of Statistical Software\/}~{\em 70}.

\bibitem[\protect\citeauthoryear{Liu and Pierce}{Liu and
  Pierce}{1994}]{adaptive_GH_1994}
Liu, Q. and D.~A. Pierce (1994).
\newblock A note on {G}auss-{H}ermite quadrature.
\newblock {\em Biometrika\/}~{\em 81}, 624--629.

\bibitem[\protect\citeauthoryear{Margossian, Vehtari, Simpson, and
  Agrawal}{Margossian et~al.}{2020}]{hybrid}
Margossian, C.~C., A.~Vehtari, D.~Simpson, and R.~Agrawal (2020).
\newblock {Hamiltonian Monte Carlo using an adjoint-differentiated Laplace
  approximation}.
\newblock {\em arXiv:2004.12550v3\/}.

\bibitem[\protect\citeauthoryear{Monnahan and Kristensen}{Monnahan and
  Kristensen}{2018}]{tmbstan}
Monnahan, C. and K.~Kristensen (2018).
\newblock {No-U-turn} sampling for fast {Bayesian} inference in {ADMB} and
  {TMB}: Introducing the adnuts and tmbstan {R} packages.
\newblock {\em PLOS ONE\/}~{\em 13}, 1--10.

\bibitem[\protect\citeauthoryear{Naylor and Smith}{Naylor and
  Smith}{1982}]{nayloradaptive}
Naylor, J. and A.~F.~M. Smith (1982).
\newblock Applications of a method for the efficient computation of posterior
  distributions.
\newblock {\em Journal of the Royal Statistical Society, Series C: Applied
  Statistics\/}~{\em 31}, 214--225.

\bibitem[\protect\citeauthoryear{Pinheiro and Bates}{Pinheiro and
  Bates}{1995}]{non-linear-latent}
Pinheiro, J.~C. and D.~M. Bates (1995).
\newblock Approximations to the log-likelihood function in the nonlinear
  mixed-effects model.
\newblock {\em Journal of computational and Graphical Statistics\/}~{\em 4},
  12--35.

\bibitem[\protect\citeauthoryear{Rue}{Rue}{2001}]{fastsamplinggmrf}
Rue, H. (2001).
\newblock {{Fast sampling of Gaussian Markov random fields}}.
\newblock {\em Journal of the Royal Statistical Society, Series B: Statistical
  Methodology\/}~{\em 63}, 325--338.

\bibitem[\protect\citeauthoryear{Rue, Martino, and Chopin}{Rue
  et~al.}{2009}]{inla}
Rue, H., S.~Martino, and N.~Chopin (2009).
\newblock Approximate {B}ayesian inference for latent {G}aussian models by
  using integrated nested {L}aplace approximations.
\newblock {\em Journal of the Royal Statistical Society, Series B: Statistical
  Methodology\/}~{\em 71}, 319--392.

\bibitem[\protect\citeauthoryear{Schillings and Schwab}{Schillings and
  Schwab}{2016}]{schillings}
Schillings, C. and C.~Schwab (2016).
\newblock Scaling limits in computational {Bayesian} inversion.
\newblock {\em ESAIM: Mathematical Modelling and Numerical Analysis\/}~{\em
  50}, 1825--1856.

\bibitem[\protect\citeauthoryear{Schlather, Malinowski, Menck, Oesting, and
  Strokorb}{Schlather et~al.}{2015}]{randomfields}
Schlather, M., A.~Malinowski, P.~J. Menck, M.~Oesting, and K.~Strokorb (2015).
\newblock Analysis, simulation and prediction of multivariate random fields
  with package randomfields.
\newblock {\em Journal of Statistical Software\/}~{\em 63}.

\bibitem[\protect\citeauthoryear{Stringer}{Stringer}{2021}]{aghqsoftware}
Stringer, A. (2021).
\newblock Implementing adaptive quadrature for {B}ayesian inference: the aghq
  package.
\newblock arXiv:2101.04468.

\bibitem[\protect\citeauthoryear{Stringer, Brown, and Stafford}{Stringer
  et~al.}{2021}]{casecrossover}
Stringer, A., P.~Brown, and J.~Stafford (2021).
\newblock Approximate {B}ayesian inference for case crossover models.
\newblock {\em Biometrics\/}~{\em 77}, 785--795.

\bibitem[\protect\citeauthoryear{Tang and Reid}{Tang and
  Reid}{2020}]{tang2020modified}
Tang, Y. and N.~Reid (2020).
\newblock Modified likelihood root in high dimensions.
\newblock {\em Journal of the Royal Statistical Society, Series B: Statistical
  Methodology\/}~{\em 82}, 1349--1369.

\bibitem[\protect\citeauthoryear{Taylor and Diggle}{Taylor and
  Diggle}{2014}]{inlamcmc}
Taylor, B.~M. and P.~J. Diggle (2014).
\newblock {INLA or MCMC? A tutorial and comparative evaluation for spatial
  prediction in log-Gaussian Cox processes}.
\newblock {\em Journal of Statistical Computation and Simulation\/}~{\em 84},
  2266--2284.

\bibitem[\protect\citeauthoryear{Tierney and Kadane}{Tierney and
  Kadane}{1986}]{laplace}
Tierney, L. and J.~B. Kadane (1986).
\newblock Accurate approximations for posterior moments and marginal densities.
\newblock {\em Journal of the American Statistical Association\/}~{\em 81},
  82--86.

\bibitem[\protect\citeauthoryear{{van der Vaart}}{{van der
  Vaart}}{1998}]{vaart_1998}
{van der Vaart}, A.~W. (1998).
\newblock {\em Asymptotic Statistics}.
\newblock Cambridge University Press.

\bibitem[\protect\citeauthoryear{Wachter and Biegler}{Wachter and
  Biegler}{2006}]{ipopt}
Wachter, A. and L.~T. Biegler (2006).
\newblock On the implementation of a primal-dual interior point filter line
  search algorithm for large-scale nonlinear programming.
\newblock {\em Mathematical Programming\/}~{\em 106}, 25--57.

\bibitem[\protect\citeauthoryear{{Winkelbauer}}{{Winkelbauer}}{2012}]{winkelbauer12gaussian}
{Winkelbauer}, A. (2012).
\newblock Moments and absolute moments of the normal distribution.
\newblock arXiv:1209.4340.

\bibitem[\protect\citeauthoryear{Wood}{Wood}{2020}]{simplifiedinla}
Wood, S. (2020).
\newblock {Simplified integrated nested Laplace approximation}.
\newblock {\em Biometrika\/}~{\em 107}, 223--230.

\bibitem[\protect\citeauthoryear{Wood, Pya, and S\"{a}fken}{Wood
  et~al.}{2016}]{smoothestimation}
Wood, S., N.~Pya, and B.~S\"{a}fken (2016).
\newblock Smoothing parameter and model selection for general smooth models.
\newblock {\em Journal of the American Statistical Association\/}~{\em 111},
  1548--1575.

\bibitem[\protect\citeauthoryear{Yao, Gholami, Keutzer, and Mahoney}{Yao
  et~al.}{2020}]{pyhessian}
Yao, Z., A.~Gholami, K.~Keutzer, and M.~W. Mahoney (2020).
\newblock Pyhessian: Neural networks through the lens of the hessian.
\newblock {\em arXiv:1912.07145v3\/}.

\end{thebibliography}
